\documentclass[11pt,a4paper]{article}
\usepackage{jheppub}
\usepackage{mathrsfs}
\usepackage{amsmath}
\usepackage{mathtools}
\usepackage{setspace}
\usepackage{amsfonts}
\usepackage{tikz-cd}
\usepackage{url}
\usepackage{amsthm}

\newtheorem{theorem}{Theorem}[section]
\newtheorem{corollary}{Corollary}[theorem]
\newtheorem{lemma}[theorem]{Lemma}
\newtheorem{prop}[theorem]{Proposition}

\theoremstyle{remark}
\newtheorem{defi}{Definition}[section]
\newtheorem{beispiel}{Example}[section]
\newtheorem{bem}{Remark}[section]

\newcommand{\vect}[1]{\ensuremath{\boldsymbol{#1}}}

\newcommand{\extder}{\ensuremath{\mathbf{d}}}

\newcommand{\liephase}{\ensuremath{\boldsymbol{\mathcal{L}}}}
\newcommand{\Sym}{\mathsf{Sym}}
\newcommand{\Gau}{\mathsf{Gau}}
\newcommand{\Asym}{\mathsf{Asym}}

\newcommand{\SU}{\mathrm{SU}}

\newcommand{\vp}{\vec{\phi}}
\newcommand{\va}{\vec{A}}
\newcommand{\vao}{\vec{A}^{(0)}}

\newcommand{\vpi}{\vec{\pi}}

\begin{document}

\abstract{We investigate the asymptotic
symmetry group of a $\SU(2)$-Yang-Mills
theory coupled to a Higgs field in the
Hamiltonian formulation. This extends
previous work on the asymptotic structure 
of pure electromagnetism by Henneaux and 
Troessaert, and on electromagnetism 
coupled to scalar fields and pure 
Yang-Mills fields by Tanzi and Giulini. 
We find that there are no obstructions 
to global electric and magnetic charges,
though that is rather subtle in the 
magnetic case. Again it is the 
Hamiltionian implementation of boost
symmetries that need a careful and 
technically subtle discussion of 
fall-off and parity conditions of all 
fields involved.
}
	
\title{\large Asymptotic Symmetries of $\SU(2)$ Yang-Mills-Higgs Theory in Hamiltonian 
Formulation}
\author[a]{Lena Janshen}\author[b,c]{and Domenico Giulini}
\affiliation[a]{Mathematisches Institut, University of Göttingen, 37073 Göttingen}
\affiliation[b]{Institute for Theoretical Physics, Leibniz University of Hannover, 30167 Hannover}
\affiliation[c]{Center of Applied Space Technology and Microgravity (ZARM), University of Bremen, 28359 Bremen}
\emailAdd{lena.janshen@mathematik.uni-goettingen.de}
\emailAdd{giulini@itp.uni-hannover.de}
\maketitle
\begin{singlespace}
		
\section{Introduction}
\label{sec:introduction}
The study of asymptotic field-structures and 
their associated symmetry groups and charges 
for long-ranging field configurations remains 
one of the most interesting and challenging 
problems in all kinds of geometric field 
theories. The Hamiltonian analysis of these 
structures has been pioneered by Henneaux and
Troessaert in their investigations of various 
field theories, including the pure 
electromagnetic case \cite{Henneaux-ED}. 
They showed how to simultaneously ensure the 
existence of a Hamiltonian phase space,
a proper symplectic two-form, and 
a Hamiltionian action of the Poincar\'e group. 
Their analysis made perfectly 
clear that one has to expect severe
fine-tuning conditions to be imposed 
on the fields and hence the analytical 
characterisation of one's phase space in order 
to meet these conditions. Subsequently 
it was shown by Tanzi and Giulini in 
\cite{Tanzi:2020} that such conditions prevent 
globally charged states to exist in pure 
$\SU(N)$-Yang-Mills theory for $N\geq 2$, which 
is a quite unexpected result. In \cite{Tanzi:2021}
these authors also showed that no such surprises 
exist  in scalar electromagnetism and the 
abelian Higgs model.  

The present paper can be seen as a direct 
continuation of the mentioned work on pure 
Yang-Mills as well as that on electromagnetism 
coupled to scalar fields. The latter led 
to the result that there are no obstructions
to globally charged states in case of 
massive fields or those with Higgs-type
potentials, whereas the massless case 
seems to contain irremovable obstructions. 
Facing this situation, it seems natural to ask 
what happens if non-abelian gauge-fields 
are coupled to scalar fields. This is what we 
shall focus on here in case of $\SU(2)$ with a 
Higgs fields in the adjoint representation. 
It turns out that this case cannot be understood 
by simple generalisations of the previous cases, 
which essentially means that all calculations 
need to be perfomed again from scratch in order  
not to miss out on any of the new terms appearing
in the couplings arising from the Higgs field 
(not present in \cite{Tanzi:2020}) and the 
non-abelian self-coupling (not present in 
\cite{Tanzi:2021}). In addition to previous 
studies, we also attempt to add some mathematical
rigour in the justification of boundary terms of
the relevant functions on phase space%
\footnote{As phase space is infinite dimensional, functions on that space are usually referred to 
as ``functionals'' in the physics literature.},
which are essential in a twofold way: they are 
necessary for the existence of a corresponding 
Hamiltonian flow, and they also influence the 
value of the function on those phase-space points 
corresponding to the long-ranging field 
configurations.

\section{Notation and conventions}
\label{sec:Notation}
For the Minkowski metric we use the $(-,+,+,+)$ 
signature convention. Greek indices $\mu,\nu,...$ 
denote spacetime components. Latin indices 
$i,j,...$ denote spatial components while 
$\bar{a},\bar{b},...$ denote angular components, 
e.g.\ $\bar{a}\in\{\theta,\varphi\}$. Bold 
symbols like $\extder$, $\liephase, \boldsymbol{\wedge}$ denote operators in phase 
space while the corresponding non-bold ones 
denote operators on space or spacetime in the 
ususal fashion. If a function $F$ on phase 
space vanishes identically after restriction 
to the constraint surface $\mathcal{C}$, 
i.e. if $F\big\vert_\mathcal{C}=0$, we will 
sometimes simply write $F\approx 0$ 
(see section\,\ref{Setting}) following 
Dirac's notation from \cite{Dirac-book}.
The notation $\frac{\delta}{\delta\phi}$ will 
be used in the following way:
\begin{align}
    \frac{\delta F}{\delta\phi}:=\frac{\delta^{\text{EL}}f}{\delta^{\text{EL}}\phi},
\end{align}
with the Euler-Lagrange derivative 
$\delta^{\text{EL}}$\footnote{This expression can be interpreted as a gradient if $\extder F$ 
does not involve any boundary terms.}. We also use the notation $\delta$ for other kinds of variations, e.g.\ $\delta_{\vec{\epsilon}}$ for infinitesimal gauge transformations.

We make use of the Bachmann-Landau $O$- and $\text{o}$-notation. For the characterisation
of asymptotic fall-off or growth behaviour we 
shall employ the Bachmann-Landau notation involving the upper case $O$ and lower 
case $\text{o}$. As we will encounter 
asymptotic expansions in the radial parameter 
$r$ with terms proportional to some $r^{-n}$ 
for some $n\in\mathbb{N}$, we can define an operator $\mathcal{O}_{r^{-n}}$ which just 
picks out the term proportional to $r^{-n}$
in that expansion. Likewise we write  $\text{o}(r^{-\mathbb{N}})$ and $\mathcal{O}_{r^{-\mathbb{N}}}$ for a 
function that falls off faster then $r^{-n}$ 
$\forall n\in\mathbb{N}$ and the projection on that fast fall-off part. Finally, we shall use 
the symbol $\rhd$ for the formal action of 
the gauge group on the phase space.

\section{Mathematical setting}
\label{Setting}
In this section we shall essentially follow 
standard notions of Hamiltonian field, 
as e.g. generally outlined in \cite{TanziPHD} and whose notation we shall follow.
From that we recall that the mathematical
structure of a \textit{Hamiltonian system} 
consists of a symplectic manifold, the \textit{phase space} $(\mathbb{P},\Omega)$, together with a class of functions $F\in\text{Obs}\subset C^{\infty}(\mathbb{P})$, the \textit{observables} of the system, with a distinguished observable, the \textit{Hamiltonian function} $H$ which determines the time evolution as the flow of the vector field $X^H$ which, 
in turn, is defined by 
\begin{align}
    \extder H=-i_{X^H}\Omega.\label{Hameq}
\end{align}
Generalizing the equation \eqref{Hameq} to an observable $F$, it defines the corresponding \textit{Hamiltonian vector field} $X^F$. Conversely, given a Hamiltonian vector field 
$X$, the corresponding observable is called \textit{generator}, which in non-gauge theory 
is determined up to an additive constant. In contrast, 
in gauge theory, the Hamiltonian is only defined on 
a submanifold $\mathcal{C}\subset\mathbb{P}$, called the  \textit{constraint surface}, which is defined
to be the set of unique-valued Hamiltonian with corresponding (non-unique) Hamiltonian vector fields being tangential to $\mathcal{C}$\footnote{If one finds a submanifold $\mathcal{C}$ that carries a uniquely defined Hamiltonian it can be that the corresponding Hamiltonian vector field is not tangential to $\mathcal{C}$. Then the system is not dynamically consistent and one has to shrink $\mathcal{C}$ until this is the case. It can be that $\mathcal{C}$ is then empty.}. The non-uniqueness of $X^H$ is given by the degeneracy directions of $i_{\mathcal{C}}^*\Omega$. The Hamiltonian vector fields that span the degeneracy directions are called infinitesimal \textit{proper gauge transformations}. The generators of those Hamiltonian vector fields are functions that vanish on $\mathcal{C}$. Physical states are defined to lie in $\mathcal{C}$ and two states that are connected by a gauge transformations are the same physical state.\\

We are considering a Hamiltonian description of a field theory which has infinite degrees of freedom. For the definition of the phase space we need the notion of a infinite dimensional symplectic manifold. A suitable notion is a smooth manifold modelled onto a locally convex vector space. For the class of locally convex vector spaces there still exists a notion of derivatives for functions between those spaces that fulfills usual properties that are necessary to build a smooth structure on the manifold. For details on the definition of smooth manifolds and the so called \textit{Bastiani calculus} I refer to \cite{Neeb}.

In the framework of \cite{Neeb} one can also define differential forms on such manifolds.
\begin{defi}
    A \textit{symplectic form} $\Omega$ on a smooth manifold $\mathbb{P}$ is a differential 2-form with $\extder\Omega=0$ which is \textit{weakly non-degenerate}, i.e.
    \begin{align}
        \forall\,p\in\mathbb{P}:\;\text{if}\;\Omega_p(v_1,v_2)=0\,\forall v_2\in T_p\mathbb{P}\;\Rightarrow\; v_1=0.
    \end{align}
\end{defi}
In distinction to the finite-dimensional case, $\Omega_p$ does in general not induce an isomorphism $\Omega^{\flat}_p:T_p\mathbb{P}\longrightarrow T_p^*\mathbb{P}$ as in general $T_p\mathbb{P}$ and $T_p^*\mathbb{P}$ are not isomorphic as topological vector spaces.
Hence for the existence of the Hamiltonian vector field corresponding to $F\in\text{Obs}$ we have the requirement $(\extder F)_p\in\text{Im}(\Omega^{\flat}_p)\;\forall p\in\mathbb{P}$.

The phase space $\mathbb{P}$ in Hamiltonian field theory will be some subset of the space of pairs of smooth sections in some vector bundle $E\rightarrow\Sigma$ over a 3 manifold $\Sigma$, i.e.\ $(\phi,\pi)\in\mathbb{P}\subset\Gamma(E)\times\Gamma(E)$. This subset will usually be characterized by fall-off and other boundary conditions. $\Gamma(E)$ is equipped with the $C^{\infty}$-topology \cite[Example II.1.4]{Neeb}, while $\mathbb{P}$ might be equipped with a $C^{\infty}$-topology induced from the $C^{\infty}$-topology of fields on a compactification $\bar{\Sigma}$ of $\Sigma$, which gives rise to certain fall-off conditions (see \ref{AppendixFalloff}).
Observables are \textit{local functionals}, i.e.\ of the form
\begin{align}\label{locfunc}
    F[\phi,\pi]=\int_{\Sigma}\text{d}^3x\,f(\phi(x),\pi(x),...,\nabla^{\alpha}\phi(x),\nabla^{\alpha}\pi(x),...),
\end{align}
where $\alpha$ is a multiindex with $|\alpha|\leq k$ for some $k\in\mathbb{N}$. The symplectic form in such a field theory has the formal gestalt\footnote{As we will see, it might be necessary to add a boundary term to the symplectic form. But also then the requirement for a Hamiltonian vector field for an observable to exist is managed by boundary terms for the observable.}
\begin{align}
    \Omega_{[\phi,\pi]}=\int_{\Sigma}\text{d}^3x\,\extder\pi\boldsymbol{\wedge}\cdot\extder\phi.
\end{align}
For the requirement $(\extder F)_p\in\text{Im}(\Omega^{\flat}_p)\;\forall\,p\in\mathbb{P}$ to be fulfilled for local functionals it might be necessary to add a suitable \textit{boundary term} to the integral expression (see also \cite[Chapter 4.1]{TanziPHD}).

The addition of boundary terms has important physical consequences. Take an observable $G$ with $G|_{\mathcal{C}}=0$. Such functions are called \textit{constraint functions}. The vanishing on $\mathcal{C}$ will sometimes to be denoted by $G\approx 0$. For it to have a corresponding $X^G$ we might have to add a boundary term $G^{\partial}$ with $G^{\partial}|_{\mathcal{C}}\neq 0$. In that case, $X^G$ is not an infinitesimal proper gauge transformation. We call it infinitesimal \textit{improper} gauge transformation. In practice, if the integrand of $G$ has a particular slow fall-off this \textit{long-reaching} gauge transformation might actually be improper. If additionally $X^G$ is a \textit{symmetry}, i.e.\ $X^H(G+G^{\partial})|_{\mathcal{C}}=0$ (as $X^H$ is tangent to $\mathcal{C}$, $X^H(G)|_{\mathcal{C}}$ is always true) it will be called an \textit{asymptotic} symmetry. 
\begin{defi}
    If all infinitesimal symmetries are fundamental symmetries of a group action of the \textit{symmetry group} $\Sym$ and all proper gauge transformation form a normal subgroup $\Gau\subset\Sym$, then the asymptotic symmetry group is
    \begin{align}
        \Asym=\Sym/\Gau.
    \end{align}
\end{defi}

\section{ $\SU(2)$ Yang-Mills-Higgs theory}
Let $(M,{}^4g)$ be Minkowki space.
The field in Yang-Mills theory is a connection on a $\SU (2)$ principal bundle on $M$:
\begin{align}
P=\SU (2)\times M \underset{pr_2}{\longrightarrow} M.
\end{align}
The bundle is trivial because every principal bundle over a contractible space is trivial\footnote{See Corollary 7.3 in \cite{Mitchell2006NotesOP}}. A connection on $P$ is given by a Lie algebra valued 1-form on $P$. The Lie algebra of $\SU(2)$ is $\mathfrak{su}(2)\cong \text{Skew}(\mathbb{R}^3)\cong (\mathbb{R}^3,\times)$ where $\times$ is the vector product on $\mathbb{R}^3$. On a trivial principal bundle the connection is completely determined by the Yang-Mills field $A\in C^{\infty}(M,(\mathbb{R}^3,\times))$, which is the pull back of the connection under a section $s\in\Gamma(P)$.\\

The Higgs field lives in the adjoint representation of $\SU(2)$, i.e. consider the associated vector bundle
\begin{align}
E=(\mathbb{R}^3,\times)\times M\underset{pr_2}{\longrightarrow} M
\end{align}
with the adjoint representation $\text{Ad}:\SU(2)\rightarrow \text{Aut}(\mathbb{R}^3,\times)$, $g\mapsto\text{Ad}_g=D(g)(\cdot)$, where $D:\SU(2)\rightarrow \mathrm{SO}(3)$ is the representation as a rotation matrix on $\mathbb{R}^3$.
The Higgs field is a section $\phi\in\Gamma(E)\cong C^{\infty}(M,(\mathbb{R}^3,\times))$. The connection defines a covariant derivative on $\Gamma(E)$ via the adjoint representation, i.e.\
\begin{align}
D\phi:=\text{d}\phi+e [(\text{Ad}_{*})_e(A)](\phi)=\text{d}\phi+e A\times\phi.
\end{align}
$\mathfrak{su}(2)$ carries a natural inner product, the Killing form, which is on $(\mathbb{R}^3,\times)$ given by the euclidean scalar product notated by a dot $\cdot\,$.

The curvature 2 form is the exterior derivative of the connection 1 form on the principal bundle. Pull back via a section $s$ onto the base manifold defines the Yang-Mills field strength:
\begin{align}
F=\text{d}A+e A(\wedge\otimes\times)A.
\end{align}
Now we are able to write down the Yang-Mills-Higgs action:
\begin{align}
S[A,\phi,\dot{A},\dot{\phi}]=\int_{M}\text{d}^4x\sqrt{-{}^4g}\bigg(-\frac{1}{4}\vec{F}^{\mu\nu}\cdot\vec{F}_{\mu\nu}-\frac{1}{2}D^{\mu}\vp\cdot D_{\mu}\vp-V(\vp)\bigg)+(\text{boundary}).
\end{align}
The Higgs field is characterized by the Mexican-hat potential
\begin{align}
V(\vp)=(\|\vp\|^2-a^2)^2,
\end{align}
where $a> 0$.
Note, one could also consider a family $V_{\lambda}(\vp):=\lambda(\|\vp\|^2-a^2)^2$. The investigation presented here only applies for the cases $\lambda>0$. In that regime, $\lambda$ is only a rescaling, so we can assume w.l.o.g.\ $\lambda=1$. The boundary term is yet to be chosen to make the Hamiltonian theory to exist. Let's assume its existence for a moment, but it will be dealt with later on.

\section{Hamiltonian of $\SU(2)$ Yang-Mills-Higgs theory}
In this section we review the Hamiltonian formulation of $\SU(2)$ Yang-Mills-Higgs theory on a non-dynamical Minkowski spacetime. We take a similar notation and line of argument as \cite{Tanzi:2020}, where the free Yang-Mills case is discussed. The Hamiltonian formulation of $\SU(2)$ Yang-Mills-Higgs theory is also discussed in \cite{Waida}.

Take the Hamiltonian formulation with respect to the (3+1)-split $M\cong\Sigma\times\mathbb{R}\cong\mathbb{R}^3\times\mathbb{R}$ with the corresponding split of the Minkowski metric
\begin{align}
{}^4g=\left(
\begin{array}{c|c}
-1 & 0\\
\hline
0 & g
\end{array}
\right),
\end{align}
where $g=\delta$ is the standard Euclidean metric. The action takes the form $S=\int\text{d}t\,L[A,\phi,\dot{A},\dot{\phi}]$, where
\begin{align}
L[A,\phi,\dot{A},\dot{\phi}]=\int\text{d}^3x\sqrt{g}\,\bigg(\frac{1}{2}\vec{F}^{0i}\cdot\vec{F}_{0i}-\frac{1}{4}\vec{F}^{ij}\cdot\vec{F}_{ij}+\frac{1}{2}D^{0}\vp\cdot D_0 \vp-\frac{1}{2} D^{i}\vp\cdot D_i \vp-V(\vp)\bigg)\nonumber\\+(\text{boundary})
\end{align}
is the Lagrangian function with respect to this space-time split.
Define canonical momenta via:
\begin{align}
\vpi^i:=\frac{\delta L}{\delta\dot{\va}_i}=\sqrt{g}\vec{F}^{i0}\; ,\;\vpi^0:=\frac{\delta L}{\delta \dot{\va}_0}\approx 0\;,\;
\vec{\Pi}:=\frac{\delta L}{\delta \dot{\vp}}=\sqrt{g}D^0\vp.
\end{align}
$\vpi^0\approx 0$ is a primary constraint\footnote{For the discussion of types of constraints, see \cite{HenneauxTeitelboim}.}. By a straightforward calculation, one obtains the Hamiltonian\footnote{A similar calculation is done in \cite[Example 2.2]{Troessaert} for electromagnetism.}:
\begin{align}
H_0[A,\pi,\phi,\Pi;\mu]&=\int\text{d}^3 x\,\bigg(\frac{\vpi^i\cdot\vpi_i}{2\sqrt{g}}+\frac{\vec{\Pi}\cdot\vec{\Pi}}{2\sqrt{g}}+\frac{\sqrt{g}}{4}\vec{F}_{ij}\cdot\vec{F}^{ij}+\frac{\sqrt{g}}{2}D^i\vp\cdot D_i\vp+V(\vp)\nonumber\\&+\va_0\cdot(D_i\vpi^i+\vp\times\vec{\Pi})+\vec{\mu}\cdot\vpi^{0}\bigg)+(\text{boundary}).
\end{align}
It has to be checked whether the primary constraint $\pi^0\approx 0$ is time dependent. This is indeed the case and this enforces a secondary constraint, called \textit{Gauss constraint}:
\begin{align}
\vec{\mathscr{G}}:=D_i\vpi^i+e\vp\times\vec{\Pi}\approx 0.
\end{align}
One easily checks that the Gauss constraint is preserved by time evolution and hence no further constraints a present in the theory. The extended Hamiltonian with respect to the constraint algebra is:
\begin{align}
H_{\text{ext}}[A,\pi,\phi,\Pi;\mu,\lambda]&=\int\text{d}^3x\,\bigg(\frac{\vpi^i\cdot\vpi_i}{2\sqrt{g}}+\frac{\vec{\Pi}\cdot\vec{\Pi}}{2\sqrt{g}}+\frac{\sqrt{g}}{4}\vec{F}_{ij}\cdot\vec{F}^{ij}+\frac{\sqrt{g}}{2}D^i\vp\cdot D_i\vp+V(\vp)\nonumber\\&+\vec{\lambda}\cdot\vec{\mathscr{G}}+\vec{\mu}\cdot\vpi^{0}\bigg)+(\text{boundary}),
\end{align}
where the Lagrange multiplier $\vec{\lambda}$ absorbed $\vec{A}_0$.
The (formal) symplectic form on the not yet specified phase space is:
\begin{align}
\Omega[A,\pi,\phi,\Pi]=\int\text{d}^3x\,\extder\vpi^{\mu}\boldsymbol{\wedge}\cdot\extder\va_{\mu}+\int\text{d}^3x\,\extder\vec{\Pi}\boldsymbol{\wedge}\cdot\extder\vec{\phi},
\end{align}
where $\extder$ and $\wedge$ are, respectively, the exterior derivative and the wedge product in phase space.

\subsection{Constraint algebra and gauge transformations}
The constraints are first class, which means that the set of local constraint functions $\vpi^{(0)}$ and $\vec{\mathscr{G}}$ is a subalgebra of the pointwise Poisson algebra of local functionals. One calculates easily,
\begin{align}
\{\vec{\pi}^0(x)\otimes,\vec{\pi}^{0}(x')\}&=0,\\
\{\vec{\pi}^0(x)\otimes,\vec{\mathscr{G}}(x')\}&=0,\\
\{\vec{\mathscr{G}}(x)\otimes,\vec{\mathscr{G}}(x')\}&=(.\times .)\cdot\vec{\mathscr{G}}(x)\,\delta(x-x'),
\end{align}
where $\otimes$ is the tensor product for vectors in $\mathbb{R}^3$ and $(.\times .)$ is the structure constant of the Lie algebra $(\mathbb{R}^3,\times)$.

Gauge transformations are Hamiltonian vector fields generated by first-class constraints.
The Hamiltonian vector field generated by the constraint $\vpi^0$ is on the \textit{constraint surface} $\mathcal{C}$ everywhere vanishing, hence there is no degeneracy vector field of the pullback of $\Omega$ to $\mathcal{C}$. This enables us to state $\vpi^0=0$ as a global condition without loosing non-degeneracy of $\Omega$.
From now on the symplectic form is
\begin{align}
\Omega[A,\pi,\phi,\Pi]=\int\text{d}^3x\,\extder\vpi^{i}\boldsymbol{\wedge}\cdot\extder\va_{i}+\int\text{d}^3x\,\extder\vec{\Pi}\boldsymbol{\wedge}\cdot\extder\vec{\phi}.
\end{align}
Let's define
\begin{align}
G_{\vec{\lambda}}:=\int \text{d}^3 x\sqrt{g}\,\vec{\lambda}\cdot\vec{\mathscr{G}}.
\end{align}
For the Hamiltonian vector field corresponding to $G_{\vec{\lambda}}$ to exist, the expression might be corrected by a boundary term $B$ such that
\begin{align}
\bar{G}_{\vec{\lambda}}:=G_{\vec{\lambda}}+B\not\approx 0.
\end{align}
$\vec{\lambda}$ that have the property $\bar{G}_{\vec{\lambda}}\not\approx 0$ are called \textit{improper} gauge transformations, while $\vec{\lambda}$ with $\bar{G}_{\vec{\lambda}}\approx 0$ are called \textit{proper} gauge transformations. Proper gauge transformations are the redundancies in the description of the theory, while improper gauge transformations are symmetries that actually change the physical state.\\
The only possible boundary term that has to be added to $G_{\vec{\lambda}}$ to have a finite and functionally differentiable $\bar{G}_{\vec{\lambda}}$ is
\begin{align}
B=-\oint\text{d}^2 x\,\mathcal{O}_{1}(\vec{\lambda}\cdot\vec{\pi}^r).
\end{align}
Proper gauge transformations are generated by $\bar{G}_{\vec{\epsilon}}$, where $\vec{\epsilon}$ has a fast enough fall-off such that $\mathcal{O}_{1}(\vec{\epsilon}\cdot\vec{\pi}^r)=0$. Then:
\begin{align}\label{infgauge}
\delta_{\vec{\epsilon}}\va_i=\{\va_i,\bar{G}_{\vec{\epsilon}}\}=-D_i\vec{\epsilon}\; ,\; \delta_{\vec{\epsilon}}\pi^i=\{\vpi^i,\bar{G}_{\vec{\epsilon}}\}=e\vec{\epsilon}\times \vpi^i\; ,\nonumber\\
\delta_{\vec{\epsilon}}\vp=\{\vp,\bar{G}_{\vec{\epsilon}}\}=e\vec{\epsilon}\times\vp\; ,\;\delta_{\vec{\epsilon}}\vec{\Pi}=\{\vec{\Pi},\bar{G}_{\vec{\epsilon}}\}=e\vec{\epsilon}\times\vec{\Pi}.
\end{align}
Whether improper gauge transformations exist depends on the fall-off- and boundary conditions for the canonical variables.

\section{Fall-off and boundary conditions}

\subsection{Fall-off conditions for the Yang-Mills variables}

There are three reasons for demanding fall-off- and boundary conditions. For a Hamiltonian description of a relativistic field theory to exist, the Hamiltonian has to be finite, the symplectic form has to be finite and there has to be an action of the Poincar\'{e} algebra on phase space.

Let's start with the Hamiltonian
\begin{align}
H[A_i,\phi,\pi^i,\Pi;\lambda]&=\int\text{d}^3x\bigg(\frac{1}{2}\vec{\pi}^i\cdot\vec{\pi}_i+\frac{1}{2}\vec{\Pi}\cdot\vec{\Pi}+\frac{1}{4}\vec{F}^{ij}\cdot\vec{F}_{ij}+\frac{1}{2}D^i\vec{\phi}\cdot D_i \vec{\phi}+V(\vec{\phi})\bigg)\nonumber\\&+ \int\text{d}^3 x\vec{\lambda}\cdot\vec{\mathscr{G}}-\oint\text{d}^2 x\,\mathcal{O}_{1}(\vec{\lambda}\cdot\vec{\pi}^r)+(\text{boundary?}).	
\end{align}
At first we ensure that the non-interacting Yang-Mills part $\int\text{d}^3x\,(2^{-1}\vpi^i\cdot\vpi_i+4^{-1}\vec{F}^{ij}\cdot\vec{F}_{ij})$ is finite. For that we choose the following fall-off conditions:
\begin{align}
\va_i(r,x^{\bar{a}})&\sim\sum_{n=1}^{\infty}\frac{1}{r^n}\va_i^{(n-1)}(x^{\bar{a}}),
\label{afalloff}\\
\vpi^i(r,x^{\bar{a}})&\sim\sum_{n=1}^{\infty}\frac{1}{r^{n+1}}\vpi^i_{(n-1)}(x^{\bar{a}})\label{pifalloff},
\end{align}
where we used the notation $\sim$ for an \textit{asymptotic expansion}. These fall-off conditions can be justified by geometrical reasons (see appendix \ref{AppendixFalloff}).
\begin{bem}
    With these fall-off conditions the symplectic form is logarithmically divergent. This divergence will be handled by parity conditions (\ref{parity}).
\end{bem}

\subsection{Fall-off and boundary conditions for the Higgs variables}
Inspired by appendix \ref{AppendixFalloff} we choose the the fall-off
\begin{align}\label{Higgsfalloff}
\vp(r,x^{\bar{a}})\sim\vp_{\infty}(x^{\bar{a}})+\sum_{n=1}^{\infty}\frac{1}{r^n}\vp^{(n)}(x^{\bar{a}}).
\end{align}
This alone does not make the terms of the Hamiltonian that contain $\vp$ finite. First of all, let's consider the potential term  $\int\text{d}^3x\,V(\vp)$. Using the ansatz \eqref{Higgsfalloff} the integrand has to be $O(r^{-4})$.
\begin{prop}
	\label{higgsb}
	Under the assumption of the fall-off behaviour being \eqref{Higgsfalloff}, $V(\vp)\in O(r^{-4})$ $\Leftrightarrow$ $\vp_{\infty}\cdot\vp^{(1)}=0$ and $\|\vp_{\infty}(x)\|^2=a^2$ $\forall x\in S^2$\footnote{$\|.\|$ is the natural norm on $\mathfrak{su}(2)\cong\mathbb{R}^3$ induced by the scalar product.}.
\end{prop}
\begin{proof}With the ansatz
	$\vp=\vec{\phi}_{\infty}+\frac{1}{r}\vp^{(1)}+\frac{1}{r^2}\vec{\phi}^{(2)}+O(r^{-3})$,
	\begin{align}
	V(\vp)&=\bigg(\left\| \vp_{\infty}+\frac{1}{r}\vp^{(1)}+\frac{1}{r^2}\vp^{(2)}+O(r^{-3})\right\|^2-a^2\bigg)^2\nonumber \\ \nonumber
	&=\bigg(\|\vp_{\infty}\|^2-a^2+\frac{1}{r^2}\Vert\vp^{(1)}\Vert^2+\frac{1}{r^4}\Vert\vp^{(2)}\Vert^2+\frac{1}{r}\vp_{\infty}\cdot\vp^{(1)}+\frac{1}{r^2}\vp_{\infty}\cdot\vp^{(2)}+\frac{1}{r^3}\vp^{(1)}\cdot\vp^{(2)}+\mathcal{O}(r^{-3})\bigg)^2\\ 
	&=\frac{1}{r^2}(\vp_{\infty}\cdot\vp^{(1)})^2+\frac{1}{r^3}(\vp_{\infty}\cdot\vp^{(1)})(\vp_{\infty}\cdot\vp^{(2)})+\frac{1}{r^3}(\vp_{\infty}\cdot\vp^{(1)})\Vert\vp^{(1)}\Vert^2\nonumber\\&+(\|\vp_{\infty}\|^2-a^2)\bigg(\frac{1}{r^2}(\vp_{\infty}\cdot\vp^{(1)})^2+\frac{1}{r^3}(\vp_{\infty}\cdot\vp^{(1)})(\vp_{\infty}\cdot\vp^{(2)})+\frac{1}{r^3}(\vp_{\infty}\cdot\vp^{(1)})\Vert\vp^{(1)}\Vert^2\bigg)\nonumber\\&+(\|\vp_{\infty}\|^2-a^2)^2+O(r^{-4}).
	\label{pot}
	\end{align}
	Certainly $V(\vp)\in O(r^{-4})$ $\Leftrightarrow$ $\|\vp_{\infty}\|=a^2$ and $\vp_{\infty}\cdot\vp^{(1)}=0$.
\end{proof}
\begin{bem}
	In principle one might eliminate higher order terms in $V(\vp)$ by using parity conditions, such that after performing the radial integral, the spherical integral is zero. But as $(\vp_{\infty}\cdot\vp^{(1)})^2$ and $\|\vp_{\infty}\|^2-a^2$ are even functions on $S^2$, this is not possible.
\end{bem}
\begin{bem}
	The condition $\vp_{\infty}\cdot\vp^{(1)}=0$ is gauge-invariant. A gauge transformation is generated by a Lie algebra valued function $\vec{\epsilon}$ with the asymptotic behaviour $\vec{\epsilon}=\vec{\epsilon}_{(0)}+\frac{1}{r}\vec{\epsilon}_{(1)}+O(r^{-2})$ (see section \ref{gauge}). So
	\begin{align}
	\delta_{\vec{\epsilon}}(\vp_{\infty}\cdot\vp^{(1)})&=\mathcal{O}_{r^{-1}}(\delta_{\vec{\epsilon}}(\vp\cdot\vp))\nonumber\\&=2e((\vec{\epsilon}_{(0)}\times\vp_{\infty})\cdot\vp^{(1)}+(\vec{\epsilon}_{(0)}\times\vp^{(1)})\cdot\vp_{\infty}+(\vec{\epsilon}^{(1)}\times\vp_{\infty})\cdot\vp_{\infty}=0.
	\end{align}
	$\|\vp_{\infty}\|^2=a^2$ is obviously also gauge invariant.
\end{bem}
\begin{defi}
	$\vp_{\infty}$ is called the \textit{Higgs vacuum component} or simply the \textit{Higgs vacuum}, because $\vp_{\infty}$ has zero potential energy, i.e. $V(\vp_{\infty}(x))=0$ $\forall x\in\Sigma$ and has also zero kinetic energy (see \ref{timevac}).
\end{defi}
It remains to set a boundary condition for the canonical momentum of the Higgs field. In order to make the term $\int_{\Sigma}\text{d}^3x\frac{1}{2}\vec{\Pi}\cdot\vec{\Pi}$ of the Hamiltonian finite and to let the Poincar\'{e} transformations \ref{aspoin} leave the fall-off \eqref{Higgsfalloff} invariant we have to choose necessarily
\begin{align}\label{Higgsmomfalloff}
\vec{\Pi}(r,x^{\bar{a}})\sim\sum_{n=2}^{\infty}\frac{1}{r^n}\vec{\Pi}^{(n-2)}(x^{\bar{a}}).
\end{align}

\subsection{Further boundary conditions}
The fall-off conditions \eqref{afalloff} for $\va_i$ and \eqref{Higgsfalloff} for $\vp$ do not make the term $\int\text{d}^3 x\, \frac{1}{2}D^i\vec{\phi}\cdot D_i\vec{\phi}$ of the Hamiltonian finite. For that to be the case it has to be $D_i\vec{\phi}\in O(r^{-2})$, while in general
\begin{align}
D_i\vec{\phi}=\underbrace{\partial_i\vec{\phi}_{\infty}}_{\in O(r^{-1})}+\underbrace{\partial_i(r^{-1}\vec{\phi}^{(1)})}_{\in O(r^{-2})}+e \underbrace{r^{-1}\vec{A}^{(0)}_i\times\vec{\phi}_{\infty}}_{\in O(r^{-1})}+\underbrace{er^{-2}\vec{A}^{(0)}_i\times\vec{\phi}^{(1)}}_{\in O(r^{-2})}+\underbrace{er^{-2}\va_i^{(1)}\times\vp_{\infty}}_{\in O(r^{-2})}+O(r^{-3}).
\end{align}
The terms of order $r^{-1}$ have to vanish. Therefore $D_i^{(0)}\vp_{\infty}:=\partial_i\vec{\phi}_{\infty}+er^{-1}\vec{A}^{(0)}_i\times\vec{\phi}_{\infty}=0$, and by taking the cross product with $\vec{\phi}_{\infty}$ from the left, the equation can be solved for $\vec{A}^{(0)}_i$:
\begin{align}
\frac{1}{r}\vec{A}^{(0)}_i=-\frac{1}{ea^2}\vec{\phi}_{\infty}\times\partial_i\vec{\phi}_{\infty}+\frac{1}{r}\vec{\phi}_{\infty}\bar{A}_i,\;\bar{A}_i:=\frac{1}{a^2}\vec{\phi}_{\infty}\cdot\vec{A}^{(0)}_i.\label{ai}
\end{align}
In spherical coordinates $r$ and $\bar{x}^{\bar{a}}\in\{\theta,\varphi\}$, $\bar{a}\in\{1,2\}$ it follows (see also \ref{spherical}):
\begin{align}
\vao_r=\vp_{\infty}\bar{A}_r\;\text{and}\;\vao_{\bar{a}}=-\frac{1}{ea^2}\vp_{\infty}\times\partial_{\bar{a}}\vp_{\infty}+\vp_{\infty}\bar{A}_{\bar{a}}.\label{barrbara}
\end{align}
Until this point we made sure that the part of the Hamiltonian that does not generate improper or proper gauge transformations is finite, i.e.
\begin{align}
\int\text{d}^3x\,\bigg(\frac{1}{2}\vpi^i\cdot\vpi_i+\frac{1}{2}\vec{\Pi}\cdot\vec{\Pi}+\frac{1}{4}\vec{F}^{ij}\cdot\vec{F}_{ij}+\frac{1}{2}D^i\vp\cdot D_i\vp\bigg)<\infty.
\end{align}
Let's consider $\Gamma[\vec{\epsilon}]=\int\text{d}^3x\,\vec{\epsilon}\cdot\vec{\mathscr{G}}-\oint\text{d}^2x\,\mathcal{O}_{1}(\vec{\epsilon}\cdot\vpi^r)$. There are no boundary terms of infinite values, because $\vpi^i$ fulfills the fall-off conditions \eqref{pifalloff} and $\vec{\epsilon}$ will fulfill fall-off conditions that prevent this scenario. Infinitesimal gauge transformations $\vec{\epsilon}$ have to preserve the fall-off conditions \eqref{afalloff}, \eqref{pifalloff}, \eqref{Higgsfalloff} and \eqref{Higgsmomfalloff}. By the transformation rules \eqref{infgauge},
\begin{align}\label{epsilonfalloff}
\vec{\epsilon}(r,x^{\bar{a}})\sim\sum_{n=0}^{\infty}\frac{1}{r^n}\vec{\epsilon}^{(n)}(x^{\bar{a}}).
\end{align}
For $\int\text{d}^3x\,\vec{\epsilon}\cdot\vec{\mathscr{G}}<\infty$, it has to be $\vec{\mathscr{G}}\in O(r^{-4})$. $\vec{\mathscr{G}}=D_i\vpi^i+e\vp\times\vec{\Pi}$, therefore
\begin{align}
&\mathcal{O}_{r^{-2}}(\vec{\mathscr{G}})=e\vp_{\infty}\times\vec{\Pi}^{(0)}=0\;\Rightarrow\; \vec{\Pi}^{(0)}=\bar{\Pi}\vp_{\infty},\label{gr2}\\
&\mathcal{O}_{r^{-3}}(\vec{\mathscr{G}})=D_i^{(0)}\vpi^i_{(0)}+e\vp^{(1)}\times\vec{\Pi}^{(0)}+e\vp_{\infty}\times\vec{\Pi}^{(1)}=0\label{gr3}.
\end{align}
While the first equation gives a concrete boundary condition at this point, the second equation will be solved later on when the action of the Poincar\'{e} group forces futher boundary conditions on the fields (section \ref{aspoin}).
\begin{bem}Reminder: The conditions \eqref{gr2} and \eqref{gr3} are actual boundary conditions and not constraints.
\end{bem}
All of the fall-off conditions and boundary conditions together guarantee a finite Hamiltonian. To generate a well defined Hamiltonian vector field the Hamiltonian has to be differentiable (see section \ref{Setting}). Indeed we do not have to introduce further boundary terms to $H$, apart from $\oint\text{d}^2x\,\mathcal{O}_1(\vec{\epsilon}\cdot\vpi^r)=\oint\text{d}^2x\,\vec{\epsilon}^{(0)}\cdot\vpi^r_{(0)}$. The two terms of $H$ that could in principle also need additional boundary terms are $\int\text{d}^3x\,2^{-1}D^i\vp\cdot D_i\vp$ and $\int\text{d}^3x\,4^{-1}\vec{F}^{ij}\cdot\vec{F}_{ij}$. The latter one does not lead to a boundary term, because this term would come from the integral $\int\text{d}^3x\,4^{-1}\partial_i(\vec{F}^{ij}\cdot \va_j)$. By the fall-off \eqref{afalloff}, $\partial_i(\vec{F}^{ij}\cdot\va_j)\in O(r^{-4})$, which gives a zero boundary term. For $\int\text{d}^3x\,2^{-1}D^i\vp\cdot D_i\vp$ it is a little bit more involved, because the finiteness of this integral is not achieved by a fall-off behaviour. Indeed, $\partial_i((D^i\vp)\cdot\vp)\in O(r^{-3})$ and
\begin{align}\label{noboundaryinthehiggs}
\int\text{d}^3x\,\partial_i((D^i\vp)\cdot\vp)&=\oint\text{d}^2x\,\sqrt{\bar{\gamma}}r^2\mathcal{O}_{r^{-2}}((D^r\vp)\cdot\vp)\nonumber\\&=\oint\text{d}^2x\,\sqrt{\bar{\gamma}}(-\vp^{(1)}\cdot\vp_{\infty}+(\va_r^{(0)}\times\vp^{(1)})\cdot\vp_{\infty}+(\va_r^{(1)}\times\vp_{\infty})\cdot\vp_{\infty})\nonumber\\&=0,
\end{align}
where we used $D^r_{(0)}\vp_{\infty}=0$ which is equivalent to $\va^{(0)}_r=\bar{A}_r\vp_{\infty}$ (\ref{barrbara}) and $\vp_{\infty}\cdot\vp^{(1)}=0$ (\ref{higgsb}). Having that, the fall-off conditions \eqref{afalloff}, \eqref{pifalloff}, \eqref{Higgsfalloff} and \eqref{Higgsmomfalloff} together with the boundary conditions
\begin{align}\label{boundary}
\|\vp_{\infty}\|^2=a^2,\;\;\vp_{\infty}\cdot\vp^{(1)}=0,\;\;\vec{\Pi}^{(0)}=\bar{\Pi}\vp_{\infty},\;\;\mathcal{O}_{r^{-3}}(\vec{\mathscr{G}})=0,\nonumber\\
\va_r^{(0)}=\bar{A}_r\vp_{\infty},\;\;\va_{\bar{a}}^{(0)}=\bar{A}_{\bar{a}}\vp_{\infty}-(ea^2)^{-1}\vp_{\infty}\times\partial_{\bar{a}}\vp_{\infty}
\end{align}
lead to a set of canonical variables with a differentiable and finite Hamiltonian function.

If the symplectic form would be finite at this point one could determine the Hamiltonian vector field. But we can do one intresting remark about the time development at this point. The Higgs part of the symplectic form is $\Omega^{\text{Higgs}}=\int\text{d}^3x\,\extder{\vec{\Pi}}\cdot\boldsymbol{\wedge}\extder\vp$. Later we will see that $\vp$ and $\vec{\Pi}$ fulfill the right boundary conditions such that $\Omega^{\text{Higgs}}$ is finite. Using that without further characterization, we can calculate the equations of motion for the Higgs vacuum $\vp_{\infty}$.
\begin{prop}\label{timevac}
	The time development of the Higgs vacuum $\vp_{\infty}$ is a gauge transformation $\partial_t\vp_{\infty}=e\vec{\epsilon}^{(0)}\times\vp_{\infty}$.
\end{prop}
\begin{proof}
	By the Hamiltonian equation coming from the symplectic form,
	\begin{align}
	\partial_t\vp_{\infty}=\mathcal{O}_1\bigg(\frac{\delta H}{\delta\vec{\Pi}}\bigg)=\underbrace{\mathcal{O}_1(\vec{\Pi})}_{=0}+e\mathcal{O}_1(\vec{\epsilon}\times\vp)=e\vec{\epsilon}^{(0)}\times\vp_{\infty}.
	\end{align}
\end{proof}
\begin{bem}
	At this point it is not clear whether $\partial_t\vp_{\infty}$ is a proper or improper gauge transformation. This depends on the boundary term $\oint\text{d}^2x\,\vec{\epsilon}^{(0)}\cdot\vpi^r_{(0)}$. As it turns out in \ref{aspoin}, we have to choose the boundary condition $\vpi^r_{(0)}=\bar{\pi}^r\vp_{\infty}$ which makes only the $\vec{\epsilon}^{(0)}$ that are parallel to $\vp_{\infty}$ to improper gauge transformations. Then $\partial_t\vp_{\infty}$ is obviously a proper gauge transformation.
\end{bem}

\subsection{Fall-off in spherical coordinates}\label{spherical}
For a lot of the discussions in this work it is convenient to use spherical coordinates on $\Sigma\cong\mathbb{R}^3$. Spherical coordinates are defined via 
\begin{align}
x(r,\theta,\varphi)=r\begin{pmatrix} \sin(\theta)\cos(\varphi)\\ \sin(\theta)\sin(\varphi)\\ \cos(\theta)\end{pmatrix}. 
\end{align}
The fall-off of the Yang-Mills field expressed in spherical coordinates is
\begin{align}
\va_r&=\va_i\frac{\partial x^i}{\partial r}=\bigg(\frac{1}{r}\va_i^{(0)}+\frac{1}{r^2}\va_i^{(1)}+...)\bigg)\underbrace{(e_r)^i}_{\in O(1)}=\frac{1}{r}\va_r^{(0)}+\frac{1}{r^2}\va^{(1)}_r+...,\\
\va_{\theta}&=\va_i\frac{\partial x^i}{\partial \theta}=\bigg(\frac{1}{r}\va_i^{(0)}+\frac{1}{r^2}\va_i^{(1)}+...\bigg)\underbrace{r(e_{\theta})^i}_{\in O(r)}=\va^{(0)}_{\theta}+\frac{1}{r}\va^{(1)}_{\theta}+...,\\
\va_{\varphi}&=\va_i\frac{\partial x^i}{\partial \varphi}=\bigg(\frac{1}{r}\va_i^{(0)}+\frac{1}{r^2}\va_i^{(1)}+...\bigg)\underbrace{r(e_{\varphi})^i}_{\in O(r)}=\va^{(0)}_{\varphi}+\frac{1}{r}\va^{(1)}_{\varphi}+...\,.
\end{align}
For the transformation of the canonical momenta it has to be recalled that these are 1-densities, i.e. $\vpi^i=-\sqrt{g}\vec{E}^i$, where $\vec{E}^i$ is a vector field and $g$ is the flat Riemannian metric on $\Sigma$. In spherical coordinates this metric has the form:
\begin{align}
g=\left(
\begin{array}{c|c}
1 & 0\\
\hline
0 & r^2\bar{\gamma}
\end{array}
\right),
\end{align}
where $\bar{\gamma}$ is the round metric on $S^2$. $\Rightarrow$ $\sqrt{g}=r^2\sqrt{\bar{\gamma}}$. It follows (assuming the fall-off behaviour $\vec{E}^i=r^{-2}\vec{E}^i_{(0)}+\mathcal{O}(r^{-3})$):
\begin{align}
\vpi^r&=-r^2\sqrt{\bar{\gamma}}\underbrace{\frac{\partial r}{\partial x^i}}_{\in\ O(1)}\vec{E}^i=\vpi_{(0)}^r+\frac{1}{r}\vpi^r_{(1)}+...\; ,\\
\vpi^{\bar{a}}&=-r^2\sqrt{\bar{\gamma}}\underbrace{\frac{\partial x^{\bar{a}}}{\partial x^i}}_{\in O(r^{-1})}\vec{E}^i=\frac{1}{r}\vpi^{\bar{a}}_{(0)}+\frac{1}{r^2}\vpi^{\bar{a}}_{(1)}+...\,,
\end{align}
where $x^{\bar{a}}\in\{\theta,\varphi\}$.

\subsection{Asymptotic electromagnetic structure with magnetic charge}
\label{gauge}
With the boundary conditions \eqref{ai} for the asymptotic Yang-Mills field there are distinguished special kinds of gauge transformations, i.e.\ gauge transformations that stabilize the Higgs vacuum. They give rise to infinitesimal $\mathrm{U}(1)$ gauge transformations of the free components $\bar{A}_i$ of the asymptotic Yang-Mills field: i.e.\ $\mathcal{O}_1(\delta_{\vec{\lambda}}\vp)=\vec{\lambda}^{(0)}\times\vp_{\infty}\stackrel{!}{=}0$ $\Rightarrow\,\vec{\epsilon}^{(0)}=\bar{\epsilon}\vp_{\infty}$ and therefore:
\begin{align}\label{gaugipara}
\delta_{\vec{\epsilon}^{(0)}}\va_i^{(0)}:=\mathcal{O}_{r^{-1}}(\delta_{\vec{\epsilon}}\va_i)&=-\partial_i\vec{\epsilon}^{(0)}+e\vec{\epsilon}^{(0)}\times\vao_i\nonumber\\&=-(\partial_i\bar{\epsilon})\vp_{\infty}-\bar{\epsilon}\partial_i\vp_{\infty}-\frac{\bar{\epsilon}}{a^2}\vp_{\infty}\times(\vp_{\infty}\times\partial_i\vp_{\infty})=-(\partial_i\bar{\epsilon})\vp_{\infty}.
\end{align}
\begin{lemma}\label{derivative}
	Let $\lambda\in C^{\infty}(S^2)$. Then, under the assumption $\va_{i}^{(0)}=\bar{A}_{i}\vp_{\infty}-(ea^2)^{-1}\vp_{\infty}\times\partial_{i}\vp_{\infty}$, it is $D_{i}^{(0)}(\bar{\lambda}\vp_{\infty}):=\partial_i(\bar{\lambda}\vp_{\infty})+e\va_i^{(0)}\times\bar{\lambda}\vp_{\infty}=(\partial_{i}\lambda)\vp_{\infty}$.
\end{lemma}
\begin{proof}
	See equation \eqref{gaugipara}.
\end{proof}
As the asymptotic componenent of the Yang-Mills momentum has to be $\vpi^i_{(0)}=\bar{\pi}^i\vp_{\infty}$, which will be shown in the section \ref{aspoin}, $\delta_{\bar{\epsilon}}\bar{\pi}^i=0$, like in electromagnetism.

The spatial Yang-Mills curvature is getting the form of a slightly modified spatial field-strenght tensor of of electromagnetism:
\begin{align}\label{asymptoticF}
\vec{F}_{ij}^{(0)}=&\,\partial_i\va^{(0)}_j-\partial_j\va_i^{(0)}+e\va^{(0)}_i\times\va^{(0)}_j\nonumber\\
=&\,\partial_i(\bar{A}_j\vp_{\infty})-\partial_j(\bar{A}_i\vp_{\infty})-(ea^2)^{-1}(\partial_i(\vp_{\infty}\times\partial_j\vp_{\infty})-\partial_j(\vp_{\infty}\times\partial_i\vp_{\infty})\nonumber\\
&-e\bar{A}_i\vp_{\infty}\times(ea^2)^{-1}(\vp_{\infty}\times\partial_j\vp_{\infty})+e\bar{A}_j\vp_{\infty}\times(ea^2)^{-1}(\vp_{\infty}\times\partial_i\vp_{\infty})\nonumber\\&+e(ea^2)^{-2}(\vp_{\infty}\times\partial_i\vp_{\infty})\times(\vp_{\infty}\times\partial_j\vp_{\infty})\nonumber\\=&\,(\partial_i\bar{A}_j-\partial_j\bar{A}_i)\vp_{\infty}+(ea^4)^{-1}(\vp_{\infty}\cdot(\partial_i\vp_{\infty}\times\partial_j\vp_{\infty}))\vp_{\infty}=:\bar{F}_{ij}\vp_{\infty}.
\end{align}
The first term of $\bar{F}_{ij}$ matches with the spatial field strenght tensor of electromagnetism. The second term is actually the asymptotic field of a \textit{magnetic monopole}.

In a theory of electromagnetism where the electric field is dual to the magnetic field, the coupling of the magnetic charge current to the electromagnetic field is descibed by the equation $\text{d}F=j^{\text{mag}}$, with the 4 dimensional electromagnetic field strenght tensor $F$ \cite[Equation 2.9]{Goddard:1978}. The magnetic charge in the Cauchy hypersurface $\Sigma$ is
\begin{align}
m=\frac{1}{4\pi}\int_{\Sigma}j^{\text{mag}}=\frac{1}{4\pi}\int_{\Sigma}\text{d}F=\frac{1}{4\pi}\int_{\partial\bar{\Sigma}}F\,.
\end{align}
Taking asymptotically $F_{ij}=\bar{F}_{ij}$, the magnetic charge in the $\SU(2)$ Yang-Mills-Higgs theory with adapted units is
\begin{align}
m=\frac{ea^4}{4\pi}\oint\text{d}^2x\,\frac{1}{2}(ea^4)^{-1}\epsilon_{\bar{a}\bar{b}}\,\vp_{\infty}\cdot(\bar{\nabla}^{\bar{a}}\vp_{\infty}\times\bar{\nabla}^{\bar{b}}\vp_{\infty}),\label{them}
\end{align}
where $\epsilon_{\bar{a}\bar{b}}$ is the volume two form on $S^2$ and $\bar{\nabla}$ the Levi civita connection of the round metric $\gamma_{\bar{a}\bar{b}}$ (see \ref{spherical}). In the appendix \ref{appendix1} it will be shown that $m$ is actually a \textit{topological charge}, i.e.\ it is the \textit{winding number} of the vacuum Higgs field $\vp_{\infty}:S^2\longrightarrow S^2$. In the section \ref{globalstructure} it will be shown that $m$ labels the connection components of the phase space.\\\\
As in electromagnetism, states can also have a electric charge, which is the generator of asymptotically constant $\mathrm{U}(1)$ gauge transformations, i.e. with adapted units
\begin{align}
q=\frac{1}{4\pi a^2}\oint\text{d}^2x\,\bar{\epsilon}\,\vp_{\infty}\cdot\vpi^r_{(0)}\;\text{and}\;\bar{\epsilon}=1.
\end{align}
With the result $\vpi^r_{(0)}=\bar{\pi}^r\vp_{\infty}$ of section \ref{aspoin} we get the electric charge of electromagnetism, i.e. $q=(4\pi)^{-1}\oint\text{d}^2x\,\bar{\pi}^r$.

Depending on further boundary conditions (section \ref{parity}) we will investigate whether a non-zero electric charge is actually allowed in phase space. We will see that this is the case.

In the Hamiltonian study of electromagnetism in the Henneaux/Troessaert paper \cite{Henneaux-ED} it is also found an infinite number of non-zero \textit{soft charges}, which are
\begin{align}
q[\bar{\epsilon}]=\frac{1}{4\pi}\oint\text{d}^2x\,\bar{\epsilon}\,\bar{\pi}^r
\end{align}
for $\bar{\epsilon}\in C^{\infty}(S^2)$ being an even function under parity $x\mapsto-x$. Soft charges corresponding to odd $\bar{\epsilon}$ are also allowed, but they are calculated by a different integral expression (see also section \ref{asym}).

\subsection{Example: Dyon-solution}\label{Dyon}
The fall-off- and boundary conditions we stated should allow for known solutions of $\SU(2)$ Yang-Mills-Higgs theory. The most important solution is the Julia-Zee Dyon, discovered by Bernard Julia and Anthony Zee in 1975 (\cite{Julia.Zee:1975}).

The Dyon has a certain asymptotic behaviour. Using the notation of \cite{Giulini:1995}, the equations (2.5) and \cite[Section III.]{Julia.Zee:1975}, as well as \cite[equation 4.26]{Goddard:1978}, one finds:
\begin{align}
\va_0(x)&\sim \vec{n}\,(1+\frac{b_1}{r}+...)\label{Dyon1}\\
\va_i(x)&\sim \frac{\vec{e}_i\times\vec{n}}{r}+O(e^{-\alpha r})\label{Dyon2}\\
\vp(x)&\sim\vec{n}+O(e^{-\beta r}),
\end{align}
where $\alpha$, $\beta$ and $b_1$ are some constants. The vector space of the coordinate chart $x$ for $\Sigma\cong\mathbb{R}^3$ is identified with $\mathfrak{su}(2)\cong\mathbb{R}^3$ by chosing an orthonormal basis for $\mathfrak{su}(2)$. In this identification is $\vec{n}=\frac{\vec{x}}{r}$ and $\vec{e}_i$ the i-th basisvector. These boundary condition fulfill the fall-off and boundary conditions stated above.  The boundary condition $\va_i^{(0)}=\bar{A}_i\vp_{\infty}-(ea^2)^{-1}\vp_{\infty}\times\partial_i\vp_{\infty}$ follows by $\partial_i\vec{n}=r^{-1}\vec{e}_i-r^{-1}n_i\vec{n}$ $\Rightarrow$ $r^{-1}\vec{e}_i\times\vec{n}=-\vec{n}\times\partial_i{n}$, which fulfills the boundary condition \eqref{ai} with $\bar{A}_i=0$.

The Dyon solution has the magnetic charge $m=1$ and the electric charge $q=b_1$. $m=1$ is easily calculated with
the equation \eqref{them}.
  For the calculation of the electric charge we need the radial asymptotic electric field $\bar{\pi}^r$. While using \eqref{Dyon1} and \eqref{Dyon2} one finds $\vec{\pi}_{(0)}^r=\sqrt{\bar{\gamma}}r^2\vec{F}_{(0)}^{0r}=\va_0^{(1)}0=b_1\vec{n}$. Then $q=\frac{1}{4\pi a^2}\oint\text{d}^2x\,\vec{n}\cdot\vpi^r_{(0)}=b_1$.

  \section{Poincar\'{e} transformations}		
  
  As $\SU(2)$ Yang-Mills-Higgs theory is a relativistic field theory, the phase space $\mathbb{P}$ has to carry an canonical action of the Poincar\'{e} group (which is then automatically Hamiltonian as the Poincar\'{e} algebra is perfect, see \cite[Chapter 3.2, 3.3]{Woodhouse:GQ}. A necessary condition for such an action is, that there is a canonical action of the Poincar\'{e} algebra, which is a Lie algebra homomorphism $X:\mathfrak{poin}\longrightarrow \Gamma(T\mathbb{P})$, such that $\liephase_{X(\mathfrak{poin})}\Omega=0$ and $i_{X(p)}\Omega=\extder P^{p}$, where $P^p$ is a differentiable local functional for every $p\in\mathfrak{poin}$. Whether the action of the Poincar\'{e} algebra is globally integrable to an action of the Poincar\'{e} group is in this infinite dimensional setting a rather technical question, which will not be answered in this paper.\\
  
  In this chapter we check whether the Poincar\'{e} transformations preserve the fall-off- and boundary conditions from the previous chapter and state further boundary conditions that are invariant under Poincar\'{e} transformations.
  In the first section the infinitesimal Poincar\'{e} transformations of the canonical variables are calculated.
  \subsection{Action of the Poincar\'{e} algebra}
  The action of the Poincar\'{e} algebra on Minkowski-space is represented by the vector-fields
  \begin{align}\label{poingens}
  (\xi^{\perp},\vect{\xi}):=\xi\frac{\partial}{\partial x^0}+\xi^i\frac{\partial}{\partial x^i}\;\text{with}\;\xi:=b_ix^i+a^{\perp}\;\text{and}\;\xi^i:=b^{i}_j x^j+a^i
  \end{align}
  where $b_i,b_{ij}=-b_{ji},a^{\perp},a^i\in\mathbb{R}$. $b_i$ parametrize Lorentz boosts, $b_{ij}$ parametrize spatial rotations, $(a^{\perp},a^i)$ are space-time translations (see also \cite{Henneaux-ED}). These vector fields define certain Cauchy-hypersurface deformations. Hypersurface deformations are in general locally parametrized by a lapse function $N$ and a shift function $\vect{N}$. With an analogous calculation like in \cite[Chapter 2]{Tanzi:2021}, the phase-space generator of the action of the deformations on the canonical variables is determined to be
  \begin{align}
  P_{(N,\vect{N})}=\int\text{d}^3x\,(N\mathcal{P}_0+N^i\mathcal{P}_i)+B_{(N,\vect{N})},\label{thegenerator}
  \end{align}
  where
  \begin{align}
  \mathcal{P}_0&=\mathcal{H}=\frac{1}{2}\vpi^i\cdot\vpi_i+\frac{1}{2}\vec{\Pi}\cdot\vec{\Pi}+\frac{1}{4}\vec{F}^{ij}\cdot\vec{F}_{ij}+\frac{1}{2}D^i\vp\cdot D_i\vp+V(\vp)+\vec{\lambda}\cdot\vec{\mathscr{G}},\\
  \mathcal{P}_i&=\vpi^j\cdot\partial_i\va_j-\partial_j(\vpi^j\cdot\va_i)+\vec{\Pi}\cdot\partial_i\vp
  \end{align}
  and $B_{(N,\vect{N})}$ is a boundary term to make the expression functionally differentiable. $B$ is to be determined after the complete specification of the fall-off and boundary conditions.
  
  \begin{bem}\label{Poingauge}
  The gauge parameter associated to the Poincar\'{e} transformations will be denoted by $\vec{\zeta}=\xi^{\perp}\vec{\lambda}$. The fact that the gauge transformation goes along only with the boosts and the time translations is an artefact of the geometrical derivation of \eqref{thegenerator}. Indeed we could combine any Poincar\'{e} transformation (of spatial- or non-spatial type) with a gauge transformation and the result is a different but also appropriate action of the Poincar\'{e} algebra, as a gauge transformation does not change the physical state. But caution: Depending on the asymptotic behaviour, this gauge transformation can actually be improper, hence resulting in a different physical transformation. When those improper gauge transformations are present in the theory the Poincar\'{e} algebra will not be a symmetry algebra anymore, rather the symmetry algebra is the asymptotic symmetry algebra $\mathfrak{asym}$ with an embedding $\mathfrak{poin}\hookrightarrow\mathfrak{asym}$, which is not unique.
  \end{bem}
 Let's determine the local Poincar\'{e} transformations of the canonical variables, assuming that we have found the right boundary terms $B_{(\xi^{\perp},\vect{\xi})}$, with the formula
 \begin{align}
 \delta_{(\xi^{\perp},\vect{\xi})}(\,\cdot\,)=\{(\,\cdot\,),P_{(\xi^{\perp},\vect{\xi})}\}.
 \end{align}
 This calculation results in
 \begin{flalign}\label{Poinis}
 \delta_{(\xi^{\perp},\vect{\xi})}\va_i&=\frac{\xi^{\perp}\vpi_i}{\sqrt{g}}+\xi^j\partial_j \va_i+(\partial_i\xi^j)\va_j-D_i\vec{\zeta}, \nonumber\\
 \delta_{(\xi^{\perp},\vect{\xi})}\vpi^i&=D_j(\sqrt{g}\xi^{\perp} \vec{F}^{ji})-e\sqrt{g}\xi^{\perp}\vp\times D^i\vp+\partial_j(\xi^j\vpi^i)-\vpi_j\partial^j\xi^i+e\vec{\zeta}\times\vpi^i,\nonumber\\
 \delta_{(\xi^{\perp},\vect{\xi})}\vp&=\frac{1}{\sqrt{g}}\xi^{\perp}\vec{\Pi}+\xi^i\partial_i\vp+e\vec{\zeta}\times\vp,\nonumber\\
 \delta_{(\xi^{\perp},\vect{\xi})}\vec{\Pi}&=D_l(\sqrt{g}\xi^{\perp}D^l\vp)-\sqrt{g}\xi^{\perp}\frac{\partial V}{\partial\vp}+\xi^i\partial_i\vec{\Pi}+e\vec{\zeta}\times\vec{\Pi}.
 \end{flalign}
 
 \subsection{Asymptotic Poincar\'{e} transformations\label{aspoin}}
 The aim of this section is to check whether the Poincar\'{e} transformations leave the fall-off conditons \eqref{afalloff}, \eqref{pifalloff}, \eqref{Higgsfalloff}, \eqref{Higgsmomfalloff}, \eqref{epsilonfalloff} and the boundary conditions \eqref{boundary} invariant. If this is not the case we have to strengthen them to guarantee an action of the Poincar\'{e} algebra on phase space. We will see indeed that this has to be done.
 Further on we will see that we can state certain additional parity conditions on the leading orders of asymptotic expansions of the fields without breaking the Poincar\'{e} invariance.
 This conditions will be necessary to make the symplectic form finite (section \ref{parity}).\\
 
 For the asymptotic analysis it is useful to switch to the spherical coordinates $(r,\bar{x}^{\bar{a}})$ on $\Sigma$.
 The transformation generators $(\xi^{\perp},\vect{\xi})$ expressed in spherical coordinates are (see also \cite[Chapter 2.2]{Henneaux-ED})
 \begin{align}
 \xi^{\perp}=rb+T\;,\; \xi^r=W\;,\; \xi^{\bar{a}}=Y^{\bar{a}}+\frac{1}{r}\bar{\gamma}^{\bar{a}\bar{b}}\partial_{\bar{b}} W,
 \end{align}
 where $T,\, b,\, W$ and $Y^{\bar{a}}$ are the following functions of the angles $(\theta,\varphi)$:
 \begin{flalign}
 T(\theta,\varphi)&=a^{\perp}\\
 b(\theta,\varphi)&=\vec{b}\cdot\vec{e}_r(\theta,\varphi)=b_1\text{sin}(\theta)\text{cos}(\varphi)+b_2\text{sin}(\theta)\text{sin}(\varphi)+b_3\text{cos}(\theta)\\
 W(\theta,\varphi)&=\vec{a}\cdot\vec{e}_r(\theta,\varphi)=a_1\text{sin}(\theta)\text{cos}(\varphi)+a_2\text{sin}(\theta)\text{sin}(\varphi)+a_3\text{cos}(\theta)\\
 Y(\theta,\varphi)&=m_1\bigg(-\text{sin}(\varphi)\frac{\partial}{\partial\theta}-\frac{\text{cos}(\theta)}{\text{sin}(\theta)}\text{cos}(\varphi)\frac{\partial}{\partial\varphi}\bigg)\nonumber\\&+m_2\bigg(\text{cos}(\varphi)\frac{\partial}{\partial\theta}-\frac{\text{cos}(\theta)}{\sin(\theta)}\sin(\varphi)\frac{\partial}{\partial\varphi}\bigg)+m_3\frac{\partial}{\partial\varphi},
 \end{flalign}
 where $m^i=-\frac{1}{2}\epsilon^{ijk}b_{jk}$, $\vec{b}_i=b_i$ and $(\vec{a})_i=a_i$ with the $b_{jk}$, $b_i$, $a_i$ and $a^{\perp}$ from the equation \eqref{poingens}.
 
 Additionally, let's introduce some notation. In this section it is useful to make an asymptotic expansion of the covariant derivative in orders of $r^{-n}$, i.e.
 \begin{align}
 D_r&\sim\sum_{n=0}^{\infty}\frac{1}{r^{n+1}}D_r^{(n)}:=\partial_r+\frac{e}{r}(\va_r^{(0)}\times\;)+\sum_{n=1}^{\infty}\frac{e}{r^{n+1}}(\va_r^{(n)}\times\;),\nonumber\\
 D_{\bar{a}}&\sim\sum_{n=0}^{\infty}\frac{1}{r^n}D_{\bar{a}}^{(n)}:=\partial_{\bar{a}}+e(\va_{\bar{a}}^{(0)}\times\;)+\sum_{n=1}^{\infty}\frac{e}{r^n}(\va_{\bar{a}}^{(n)}\times\;).
 \end{align}
 One can pull the index of $D^{(n)}_{\bar{a}}$ up by $\bar{\gamma}^{\bar{a}\bar{b}}$.\\
 
 First of all, the structure of the fall-off behaviour of the fields, being an asymptotic series in $r^{-n}$, is preserved by the transformation rules \eqref{Poinis}, because $\xi^{\perp}=rb+T,\; \xi^r=W,\; \xi^{\bar{a}}=Y^{\bar{a}}+\frac{1}{r}\bar{\gamma}^{\bar{a}\bar{b}}\partial_{\bar{b}} W$ while \eqref{Poinis} only involves the fields, derivatives (which are $\sim r^{-1}$ for fields having an asymptotic series in $r^{-n}$), $\xi^{\perp}$ and $\vect{\xi}$.\\
 
 The next step is to examine whether any Poincar\'{e} transformation of a variable might have higher order terms as allowed by the fall-off conditions. In one strike the invariance of the boundary conditions \eqref{boundary} is checked.
 
 By \eqref{boundary} the asymptotic components of $\va_i$ depend on $\vp_{\infty}$, so it is natural to investigate the behaviour of $\delta_{(\xi^{\perp},\vect{\xi})}\vp_{\infty}$ first. Using \eqref{Poinis} and the fall-off behaviour,
 \begin{flalign}\label{Higgsboost}
 \delta_{(\xi^{\perp},\vect{\xi})}\vp_{\infty}&=\mathcal{O}_1\bigg(\frac{1}{\sqrt{g}}\xi^{\perp}\vec{\Pi}+\xi^i\partial_i\vec{\phi}+e\xi^{\perp}\vec{\zeta}\times\vec{\phi}\bigg)\nonumber\\
 &=Y^{\bar{a}}\partial_{\bar{a}}\vp_{\infty}+e\vec{\zeta}^{(0)}\times\vp_{\infty}.
 \end{flalign}
 The fall-off conditions \eqref{Higgsfalloff}, \eqref{Higgsmomfalloff} and \eqref{epsilonfalloff} prohibit us from possible higher order terms in $\delta_{(\xi^{\perp},\vect{\xi})}\vp$.
 
 The Poincar\'{e} transformation of the leading component of $\vec{A}_r$ is:
 \begin{flalign}\label{arpoin}
 \delta_{(\xi^{\perp},\vect{\xi})}\frac{1}{r}\va_r^{(0)}&=\mathcal{O}_{r^{-1}}\bigg(\frac{\xi^{\perp}\vec{\pi}_r}{\sqrt{g}}+\xi^j\partial_j\va_r+(\partial_r\xi^j)\va_j-D_r\vec{\zeta}\bigg)\nonumber\\
 &=\frac{1}{r}\frac{b}{\sqrt{\bar{\gamma}}}\vpi_r^{(0)}+\frac{1}{r}\vp_{\infty}Y^{\bar{a}}\partial_{\bar{a}}\bar{A}_r+\frac{1}{r}\bar{A}_r Y^{\bar{a}}\partial_{\bar{a}}\vp_{\infty}+\frac{e}{r}\bar{A}_r\vec{\zeta}^{(0)}\times\vp_{\infty}.
 \end{flalign}
 $\va_r^{(0)}$ fulfills the boundary condition $\va_r^{(0)}=\bar{A}_r\vp_{\infty}$. By the identity $\delta_{(\xi^{\perp},\vect{\xi})}\va_r^{(0)}=\delta_{(\xi^{\perp},\vect{\xi})}(\bar{A}_r\vp_{\infty})=(\delta_{(\xi^{\perp},\vect{\xi})}\bar{A}_r)\vp_{\infty}+\bar{A}_r(\delta_{(\xi^{\perp},\vect{\xi})}\vp_{\infty})$ and $\delta_{(\xi^{\perp},\vect{\xi})}\vp_{\infty}=Y^{\bar{a}}\partial_{\bar{a}}\vp_{\infty}+e\vec{\zeta}^{(0)}\times\vp_{\infty}$, we have to propose the boundary condition:
 \begin{align}
 \vpi_r^{(0)}=\bar{\pi}_r\vp_{\infty}.\label{pir}
 \end{align}
 
Let's calculate the asymptotic Poincar\'{e}-transformations of the angular components:
\begin{align}\label{a0poin}
\delta_{(\xi^{\perp},\vect{\xi})}\va^{(0)}_{\bar{a}}&=\mathcal{O}_{1}\bigg(\frac{\xi^{\perp}\vpi_{\bar{a}}}{\sqrt{g}}+\xi^j\partial_j\va_{\bar{a}}+\va_j\partial_{\bar{a}}\xi^j+D_{\bar{a}}\vec{\zeta}\bigg)\nonumber\\
&=\frac{b\vpi^{(0)}_{\bar{a}}}{\sqrt{\bar{\gamma}}}+Y^{\bar{b}}\partial_{\bar{b}}\va^{(0)}_{\bar{a}}+\va^{(0)}_{\bar{b}}\partial_{\bar{a}} Y^{\bar{b}}-D^{(0)}_{\bar{a}}\vec{\zeta}^{(0)}.
\end{align}
This transformation has to preserve the boundary condition $\va^{(0)}_{\bar{a}}=\bar{A}_{\bar{a}}\vp_{\infty}-\frac{1}{ea^2}\vp_{\infty}\times\partial_{\bar{a}}\vp_{\infty}$. This is the case if and only if
\begin{align}\label{a0boundarycondition}
\delta_{(\xi^{\perp},\vect{\xi})}\va_{\bar{a}}^{(0)}=(\delta_{(\xi^{\perp},\vect{\xi})}\bar{A}_{\bar{a}})\vp_{\infty}+\bar{A}_{\bar{a}}\delta_{(\xi^{\perp},\vect{\xi})}\vp_{\infty}-(ea^2)^{-1}\delta_{(\xi^{\perp},\vect{\xi})}(\vp_{\infty}\times\partial_{\bar{a}}\vp_{\infty}).
\end{align}
As calculated before, $\delta_{(\xi^{\perp},\vect{\xi})}\vp_{\infty}=Y^{\bar{a}}\partial_{\bar{a}}\vp_{\infty}+e\vec{\zeta}^{(0)}\times\vp_{\infty}$. Therefore:
\begin{align}
\delta_{(0,\vect{\xi})}(\vp_{\infty}\times\partial_{\bar{a}}\vp_{\infty})=Y^{\bar{b}}\partial_{\bar{b}}(\vp_{\infty}\times\partial_{\bar{a}}\vp_{\infty})+(\partial_{\bar{a}}Y^{\bar{b}})\vp_{\infty}\times\partial_{\bar{b}}\vp_{\infty},
\end{align}
which together with $\delta_{(0,\vect{\xi})}\vp_{\infty}=Y^{\bar{a}}\partial_{\bar{a}}\vp_{\infty}$ and \eqref{a0poin} shows that the boundary condition is preserved, while $\delta_{(0,\vect{\xi})}\bar{A}_{\bar{a}}=Y^{\bar{b}}\partial_{\bar{b}}\bar{A}_{\bar{a}}$.
Concerning the boosts,
\begin{align}
\delta_{(\xi^{\perp},0)}(\vp_{\infty}\times\partial_{\bar{a}}\vp_{\infty})=e(\vec{\zeta}^{(0)}\times\vp_{\infty})\times\partial_{\bar{a}}\vp_{\infty}+e\vp_{\infty}\times\partial_{\bar{a}}(\vec{\zeta}^{(0)}\times\vp_{\infty})\nonumber\\=e\vp_{\infty}(\vec{\zeta}^{(0)}\cdot\partial_{\bar{a}}\vp_{\infty})+ea^2\partial_{\bar{a}}\vec{\zeta}^{(0)}-e\vp_{\infty}(\partial_{\bar{a}}\vec{\zeta}^{(0)}\cdot\vp_{\infty})-e\partial_{\bar{a}}\vp_{\infty}(\vec{\zeta}^{(0)}\cdot\vp_{\infty})=(*).
\end{align}
With the decomposition $\vec{\zeta}^{(0)}=\bar{\zeta}\vp_{\infty}+\vec{\zeta}^{(0)\perp}$,
\begin{align}\label{jo1}
(*)=e\vp_{\infty}(\vec{\zeta}^{(0)\perp}\cdot\partial_{\bar{a}}\vp_{\infty})+ea^2(\partial_{\bar{a}}\vec{\zeta}^{(0)\perp})^{\perp}=-e\vp_{\infty}(\partial_{\bar{a}}\vec{\zeta}^{(0)\perp}\cdot\vp_{\infty})+ea^2(\partial_{\bar{a}}\vec{\zeta}^{(0)\perp})^{\perp}
\end{align}
and
\begin{align}\label{jo2}
-D^{(0)}_{\bar{a}}\vec{\zeta}^{(0)}=-(\partial_{\bar{a}}\bar{\zeta})\vp_{\infty}-(\partial_{\bar{a}}\vec{\zeta}^{(0)\perp})^{\perp}-e\bar{A}_{\bar{a}}\vp_{\infty}\times\vec{\zeta}^{(0)}.
\end{align}
With \eqref{jo1} and \eqref{jo2}, \eqref{a0boundarycondition} is fulfilled if and only if
\begin{align}\label{baratrafo1}
(\delta_{(\xi^{\perp},0)}\bar{A}_{\bar{a}})\vp_{\infty}=\frac{b\vpi_{\bar{a}}^{(0)}}{\sqrt{\bar{\gamma}}}-(\partial_{\bar{a}}\bar{\zeta})\vp_{\infty}-a^{-2}\vp_{\infty}(\partial_{\bar{a}}\vec{\zeta}^{(0)\perp}\cdot\vp_{\infty}).
\end{align}
This can only be the case, if there is a further boundary condition
\begin{align}
\vpi^{\bar{a}}_{(0)}=\bar{\pi}^{\bar{a}}\vp_{\infty}.
\end{align}

\begin{bem}
	Note that the first two terms of \eqref{baratrafo1} are the transformation behaviour from electromagnetism, while the third term is cancelled by a term of the transformation of $-(ea^2)^{-1}\vp_{\infty}\times\partial_{\bar{a}}\vp_{\infty}$. A redefinition of the boost transformation rule, being $\delta^{\text{EM}}_{(\xi^{\perp},0)}\bar{A}_{\bar{a}}:=b\sqrt{\bar{\gamma}}^{-1}\bar{\pi}_{\bar{a}}-\partial_{\bar{a}}\bar{\zeta}$ and $\delta^{\text{EM}}_{(\xi^{\perp},0)}(\vp_{\infty}\times\partial_{\bar{a}}\vp_{\infty}):=ea^2(\partial_{\bar{a}}\vec{\zeta}^{(0)\perp})^{\perp}$ would mirror this fact.
\end{bem}

Let's consider the Poincar\'{e} transformation of the canonical momenta of the Yang-Mills field. These have to preserve the newly found boundary  conditions $\vpi^i=\bar{\pi}^i\vp_{\infty}$ and the fall-off conditions. This leads to additional conditions for the next to leading order of $\va_i$, which have to be invariant under Poincar\'{e} transformations too. Subsequently this leads to additional conditions for the next to leading order of $\vpi^i$, resulting in a bootstrap. The result of this bootstrap performed on the constraint surface is presented in the following proposition.

\begin{prop}\label{mass}
	With the assumption that the local action \eqref{Poinis} of the Poincar\'{e} algebra leaves the boundary and fall-off conditions invariant, it follows:
	\begin{align}
	\mathcal{O}_{r^{-k}}(D_i\vp)=0\;\, \forall k\in\mathbb{N}, \forall\{((A,\phi),(\pi,\Pi)\}\in\mathcal{C},
	\end{align}
	 where $\mathcal{C}$ is the set of all canonical variables fulfilling all previous stated boundary conditions and the constraints.	
\end{prop}
\begin{proof}
	We show this statement per induction. The base case is the boundary condition \eqref{ai} which is equivalent to
	\begin{align}
	\mathcal{O}_{r^{-1}}(D_i\vp)=D_i^{(0)}\vp_{\infty}=0.
	\end{align}
	Induction step: Assume 
	\begin{align}
	\mathcal{O}_{r^{-k}}(D_i\vp)=0\;\forall k\leq n+1.
	\end{align}
	We have to show
	\begin{align}
	\mathcal{O}_{r^{-(n+2)}}(D_i\vp)=0.
	\end{align}
	The assumption is equivalent to $\partial_i\vp+e\va_i\times\vp=O(r^{-(n+2)})$. Taking the cross product with $\vp$ and using $\vp\times O(r^{-(n+2)})=O(r^{-(n+2)})$:
	\begin{align}\label{aiai}
	&\vp\times\partial_i\vp+e\vp\times(\va_i\times\vp)=O(r^{-(n+2)})\nonumber\\\Leftrightarrow\;&\vp\times\partial_i\vp+e\|\vp\|^2\va_i-e\vp(\vp\cdot\va_i)=O(r^{-(n+2)})\nonumber\\
	\Leftrightarrow\;&\va_i=A_i\vp-\frac{1}{e\|\vp\|^2}\vp\times\partial_i\vp+O(r^{-(n+2)}),
	\end{align}
	where $A_i:=\|\vp\|^{-2}\va_i\cdot\vp$. This condition has to be invariant under Lorentz boosts
	\begin{align}\label{aiboost}
	\delta_{(\xi^{\perp},0)}\va_i&=\xi^{\perp}\vpi_i-D_i\vec{\zeta},\;
	\text{i.e.}\nonumber\\\xi^{\perp}\vpi_i-D_i\vec{\zeta}&=(\delta_{(\xi^{\perp},0)}A_i)\vp+A_i(\delta_{(\xi^{\perp},0)}\vp)-\delta_{(\xi^{\perp},0)}\bigg(\frac{1}{e\|\vp\|^2}\vp\times\partial_i\vp\bigg)+O(r^{-(n+2)}).
	\end{align}
	Before conclude a condition for $\vpi_i$ we have to show an intermediate proposition:\\
	
	If $\{((A,\phi),(\pi,\Pi)\}\in\mathcal{C}$ and $\mathcal{O}_{r^{-k}}(D_i\vp)=0$ $\forall k\leq n+1$, then $\vec{\Pi}^{(k')}=0$ $\forall k'\leq n-1$.\\
	
	Proof: We use induction again, until we reach the order $k=n+1$ in $\mathcal{O}_{r^{-k}}(D_i\vp)$. The base case follows from the fact that the boundary condition $\vp^{(1)}\cdot\vp_{\infty}=0$ has to be invariant under boosts, i.e.\
	\begin{align}
	\label{basecase}
	&\;\;\,0=\delta_b(\vp^{(1)}\cdot\vp_{\infty})=(\delta_b \vp^{(1)})\cdot\vp_{\infty}+\vp^{(1)}\cdot\vp_{\infty}\nonumber\\ &\underbrace{=}_{\text{using \ref{boundary}}}\frac{b}{\sqrt{\bar{\gamma}}}\bar{\Pi}a^2+e(\vec{\zeta}^{(0)}\times\vp^{(1)})\cdot\vp_{\infty}+e(\vec{\zeta}^{(1)}\times\vp_{\infty})\cdot\vp_{\infty}+e\vp^{(1)}\cdot(\vec{\zeta}^{(0)}\times\vp_{\infty})\nonumber\\&\;\;\;\,=\frac{b}{\sqrt{\bar{\gamma}}}\bar{\Pi}a^2\;\Leftrightarrow\;\bar{\Pi}=0.
	\end{align}
	Induction step: Let $\vec{\Pi}^{(k')}=0$ $\forall k'\leq n-2$ and  $\mathcal{O}_{r^{-k}}(D_i\vp))=0$ $\forall k\leq n+1$. Using $\vec{\Pi}^{(k')}=0$ $\forall k'\leq n-2$ it follows
	\begin{align}
	\delta_{(\xi^{\perp},0)}r^{-(k'+1)}\vp^{(k'+1)}=r^{-(k'+2)}\xi^{\perp}\vec{\Pi}^{(k')}+e\,\mathcal{O}_{r^{-(k'+1)}}(\vec{\zeta}\times\vp)=e\,\mathcal{O}_{r^{-(k'+1)}}(\vec{\zeta}\times\vp)\,\;\forall k'\leq n-2.
	\end{align}
	Deploy this in the order $\sim r^{-(k'+2)}$ of the equation \eqref{aiboost} which gives
	\begin{align}
	\xi^{\perp}\vpi_i-D_i\vec{\zeta}&=(\delta_{(\xi^{\perp},0)}A_i)\vp+A_i(\delta_{\vec{\zeta}}\vp)\nonumber\\&-\delta_{\vec{\zeta}}\bigg(\frac{1}{e\|\vp\|^2}\vp\times\partial_i\vp\bigg)+O(r^{-(n+1)}).
	\end{align}
	The only term on the right hand side that is not purely a gauge transformation is the first one. Hence
	\begin{align}\label{pipi}
	(\delta^{\vec{\zeta}=0}_{(\xi^{\perp},0)}A_i)\vp=\xi^{\perp}\vpi_i+O(r^{-(n+1)})\;\Rightarrow \;\vpi_i=\pi_i\vp+O(r^{-(n+2)}),
	\end{align}
	where we used that the term $\xi^{\perp}\vpi_i$ in $\delta_{\xi^{\perp}}\va_i$ is independent of the choice of Poincar\'{e} action parameterized by the gauge parameter $\vec{\zeta}$.
	
	Let's study the different orders of the constraint equation using $\vec{\Pi}^{(k')}=0$ $\forall k'\leq n-2$:
	\begin{align}\label{constraint}
	\mathcal{O}_{r^{-(n+1)}}(\vec{\mathscr{G}})=\mathcal{O}_{r^{-(n+1)}}(D_i\vpi^i)+\vp_{\infty}\times\vec{\Pi}^{(n-1)}=0.
	\end{align}
	Using \eqref{pipi} and \eqref{aiai},
	\begin{align}
	D_i\vpi^i&=D_i(\pi^i\vp)+O(r^{-(n+3)})\nonumber\\
	&=\partial_i(\pi^i\vp)-\frac{\pi^i}{\|\vp\|^2}(\vp\times\partial_i\vp)\times\vp+O(r^{-(n+3)})\nonumber\\
	&=(\partial_i\pi^i)\vp+\frac{\pi^i}{\|\vp\|^2}(\partial_i\vp\cdot\vp)\vp+O(r^{-(n+3)}).\\
	\Rightarrow \mathcal{O}_{r^{-(n+1)}}(D_i\vpi^i)&=\mathcal{O}_{r^{-(n+1)}}(\partial_i\pi^i+\|\vp\|^{-2}\pi^i(\partial_i\vp\cdot\vp))\vp_{\infty}\nonumber\\&+\mathcal{O}_{r^{-n}}(\partial_i\pi^i+\|\vp\|^{-2}\pi^i(\partial_i\vp\cdot\vp))\frac{1}{r}\vp^{(1)}+...\; .
	\end{align}
	As $\mathcal{O}_{r^{-3}}(\vec{\mathscr{G}})=\mathcal{O}_{r^{-3}}(D_i\vpi^i)=\mathcal{O}_{r^{-3}}(\partial_i\pi^i+\|\vp\|^{-2}\pi^i(\partial_i\vp\cdot\vp))\vp_{\infty}=(\partial_i\bar{\pi}^i)\vp_{\infty}=0$,
	\begin{align} \mathcal{O}_{r^{-4}}(\vec{\mathscr{G}})=\mathcal{O}_{r^{-4}}(D_i\vpi^i)&=\mathcal{O}_{r^{-4}}(\partial_i\pi^i+\|\vp\|^{-2}\pi^i(\partial_i\vp\cdot\vp))\vp_{\infty}\nonumber\\++\mathcal{O}_{r^{-3}}(\partial_i\pi^i+\|\vp\|^{-2}\pi^i(\partial_i\vp\cdot\vp))r^{-1}\vp^{(1)}\nonumber\\&=\mathcal{O}_{r^{-4}}(\partial_i\pi^i+\|\vp\|^{-2}\pi^i(\partial_i\vp\cdot\vp))\vp_{\infty}=0
	\end{align}
	By bootstrapping this up to the order $\sim r^{-(n+1)}$ while using at every step $\vec{\Pi}^{(k')}=0$ $\forall k'\leq n-2$, we get
	\begin{align}
	\mathcal{O}_{r^{-(n+1)}}(D_i\vpi^i)&=\mathcal{O}_{r^{-(n+1)}}(\partial_i\pi^i+\|\vp\|^{-2}\pi^i(\partial_i\vp\cdot\vp))\vp_{\infty}.
	\end{align}
	Applying the constraint at this order (\ref{constraint}) we can conclude
	\begin{align}
	\vec{\Pi}^{(n-1)}=\Pi^{(n-1)}\vp_{\infty}.
	\end{align}
	Going further, this canonical momentum gives rise to a boost transformation with
	\begin{align}\label{equation}
	\delta_{(\xi^{\perp},0)}^{\vec{\zeta}=0}\vp^{(n)}=b\vec{\Pi}^{(n-1)}=b\Pi^{(n-1)}\vp_{\infty}.
	\end{align}
	Additionally, since $\vec{\Pi}^{(n-2)}=0$ it has to be $\partial_t\vec{\Pi}^{(n-2)}=0$, where
	\begin{align}\label{boooti}
	\partial_t\vec{\Pi}^{(n-2)}&=\mathcal{O}_{r^{-n}}\bigg(D_mD^m\vp-\frac{\partial V}{\partial\vp}+e\vec{\zeta}\times\vec{\Pi}\bigg)\nonumber\\&=-\mathcal{O}_{r^{-n}}\bigg(\frac{\partial V}{\partial\vp}\bigg)=-\mathcal{O}_{r^{-n}}(4(\|\vp\|^2-a^2)\vp)\nonumber\\
	&=-4\mathcal{O}_{r^{-n}}(\|\vp\|^2-a^2)\vp_{\infty}-4\mathcal{O}_{r^{-(n-1)}}(\|\vp\|^2-a^2)r^{-1}\vp^{(1)}-...=0.
	\end{align}
	Via a bootstrap starting from $\partial_t\vec{\Pi}^{(0)}=0$ and going to $\partial_t\vec{\Pi}^{(n-3)}$ one concludes that all expressions $\mathcal{O}_{r^{-l}}(\|\vp\|^2-a^2)$ for $l\leq n-1$ are already zero.
	\begin{align}
	\Rightarrow\;\mathcal{O}_{r^{-n}}(\|\vp\|^2-a^2)&=\frac{1}{r^n}(2\vp_{\infty}\cdot\vp^{(n)}+2\vp^{(1)}\cdot\vp^{(n-1)}+...)=0\nonumber\\
	\Leftrightarrow\;\vp_{\infty}\cdot\vp^{(n)}&=-\vp^{(1)}\cdot\vp^{(n-1)}-...\;.
	\end{align}
	Using $\vec{\Pi}^{(k')}=0$ $\forall k'\leq n-2$,
	\begin{align}
	\delta_{(\xi^{\perp},0)}^{\vec{\zeta}=0}(\vp_{\infty}\cdot\vp^{(n)})&=-\delta_{(\xi^{\perp},0)}^{\vec{\zeta}=0}(\vp^{(1)}\cdot\vp^{(n-1)})-...\nonumber\\
	&=-(\delta^{\vec{\zeta}=0}_{(\xi^{\perp},0)}\vp^{(1)})\cdot\vp^{(n-1)}-\vp^{(1)}\cdot(\delta_{(\xi^{\perp},0)}^{\vec{\zeta}=0}\vp^{(n-1)})-...\nonumber\\&=0.
	\end{align}
	Together with the equation \eqref{equation},
	\begin{align}
	0=\delta_{(\xi^{\perp},0)}^{\vec{\zeta}=0}(\vp_{\infty}\cdot\vp^{(n)})&=\vp_{\infty}\cdot\delta^{\vec{\zeta}=0}_{(\xi^{\perp},0)}\vp^{(n)}=a^2b\Pi^{(n-1)}.\nonumber\\
	\Rightarrow\;\Pi^{(n-1)}&=0.
	\end{align}
	With that, the intermediate proposition is proven. //\\
	
	Let's go back to the equation \eqref{aiboost}. With the intermediate proposition we have
	\begin{align}
	\delta_{(\xi^{\perp},0)}r^{-(k'+1)}\vp^{(k'+1)}=e\,\mathcal{O}_{r^{-(k'+1)}}(\vec{\zeta}\times\vp)\;\forall k'\leq n-1
	\end{align}
	and therefore equation \eqref{aiboost} gets the shape
	\begin{align}
	\xi^{\perp}\vpi_i-D_i\vec{\zeta}&=(\delta_{(\xi^{\perp},0)}A_i)\vp+A_i(\delta_{\vec{\zeta}}\vp)\nonumber\\&-\delta_{\vec{\zeta}}\bigg(\frac{1}{e\|\vp\|^2}\vp\times\partial_i\vp\bigg)+O(r^{-(n+2)}).\\
	\Rightarrow \vpi_i&=\pi_i\vp+O(r^{-(n+3)}).\label{pipipi}
	\end{align}
	This condition has to be invariant under time evolution. In general
	\begin{align}
	\partial_t\vpi^i=D_m\vec{F}^{mi}-\vp\times D^i\vp+e\vec{\zeta}\times\vpi^i.
	\end{align}
	The general form of the Yang-Mills curvature $\vec{F}$ (by the definition 2.1 in \cite{Roeser}) is
	\begin{align}
	\vec{F}\times\vp=(D\circ D)\vp.
	\end{align}
	Using \eqref{aiai},
	\begin{align}
	\vec{F}\times\vp&= D(D\vp)=O(r^{-(n+3)}).\\
	\Rightarrow\;\vec{F}&=F\vp+O(r^{-(n+3)}).\\
	\Rightarrow\;D_m\vec{F}^{mi}&=D_m(F^{mi}\vp)+O(r^{-(n+4)})\nonumber\\&=\partial_m(F^{mi}\vp)+e\va_m\times F^{mi}\vp+O(r^{-(n+4)})\nonumber\\
	&=(\partial_m F^{mi})\vp+F^{mi}\partial_m\vp-\frac{F^{mi}}{\|\vp\|^2}(\vp\times\partial_m\vp)\times\vp+O(r^{-(n+4)})\nonumber\\
	&=(\partial_mF^{mi}+(\partial_m\vp\cdot\vp)F^{mi})\vp+O(r^{-(n+4)}).
	\end{align}
	Back to the invariance of \eqref{pipipi} under the time evolution, i.e. it has to be
	\begin{align}
	\partial_t\vpi^i=(\partial_t\pi^i)\vp+\pi^i\partial_t\vp+O(r^{-(n+3)}).
	\end{align}
	As $D_m\vec{F}^{mi}$ is already parallel to $\vp$ up to the desired order and $\partial_t\vp^{(k)}=\vec{\Pi}^{(k-2)}+\delta_{\vec{\zeta}}\vp=\delta_{\vec{\zeta}}\vp$ $\forall k\leq n$ such that $e\vec{\zeta}\times\vpi^i$ includes the time evolution of $\vp$, it has to be
	\begin{align}
	&\vp\times D^i\vp=O(r^{-(n+3)}).\\
	\Leftrightarrow\;&\vp\times\partial_i\vp+e\vp\times(\va_i\times\vp)=O(r^{-(n+3)})\nonumber\\
	\Leftrightarrow\;&\va_i=A_i\vp-\frac{1}{e\|\vp\|^2}\vp\times\partial_i\vp+O(r^{-(n+3)}).
	\end{align}
	\begin{align}
	\Rightarrow\;D_i\vp&=\partial_i\vp-\frac{1}{\|\vp\|^2}(\vp\times\partial_i\vp)\times\vp+O(r^{-(n+3)})\nonumber\\
	&=\frac{1}{\|\vp\|^2}(\partial_i\vp\cdot\vp)\vp+O(r^{-(n+3)})\nonumber\\
	&=\frac{1}{2\|\vp\|^2}(\partial_i\|\vp\|^2)\vp+O(r^{-(n+3)}).\\ \Leftrightarrow\;\|\vp\|^2 2D_i\vp&=\partial_i(\|\vp\|^2)\vp+O(r^{-(n+3)}).
	\end{align}
	As we have seen before as a result of the intermediate proposition, $\partial_t\vec{\Pi}^{(k')}=0$ $\forall k'\leq n-1$ on the constraint surface. Using additionaly $\mathcal{O}_{r^{-k}}(D_i\vp)=0$ $\forall k\leq n+1$,
	\begin{align}
	0=\partial_t\vec{\Pi}^{(n-1)}=-\mathcal{O}_{r^{-(n+1)}}\bigg(\frac{\partial V}{\partial\vp}\bigg)=-4\mathcal{O}_{r^{-(n+1)}}((\|\vp\|^2-a^2)\vp).
	\end{align}
	With the same bootstrap argument as it was used at equation \eqref{boooti} we can conclude
	\begin{align}
	\mathcal{O}_{r^{-(n+1)}}(\|\vp\|^2-a^2)&=0.\\
	\Rightarrow\;\mathcal{O}_{r^{-(n+2)}}\partial_i(\|\vp\|^2-a^2)=\mathcal{O}_{r^{-(n+2)}}\partial_i(\|\vp\|^2)&=0\label{3constraint}
	\end{align}
	Clearly by considering also the $\partial_t\vec{\Pi}^{(k')}=0$ for any $k'< n-1$,
	\begin{align}
	\|\vp\|^2D_i\vp&=O(r^{-(n+3)}).
	\end{align}
	\begin{align}
	\Rightarrow\;\mathcal{O}_{r^{-(n+2)}}(\|\vp\|^2D_i\vp)&=0\nonumber\\
	\Leftrightarrow\;\mathcal{O}_1(\|\vp\|^2)\mathcal{O}_{r^{-(n+2)}}(D_i\vp)+\mathcal{O}_{r^{-1}}(\|\vp\|^2)\mathcal{O}_{r^{-(n+1)}}(D_i\vp)+...&=0.
	\end{align}
	As $\mathcal{O}_{r^{-k}}(D_i\vp)=0$ $\forall k\leq n+1$ and $\mathcal{O}_1(\|\vp\|^2)=a^2\neq 0$,
	\begin{align}
	\mathcal{O}_{r^{-(n+2)}}(D_i\vp)=0.
	\end{align}
\end{proof}
\begin{corollary}\label{pimass}
	If $\{((A,\phi),(\pi,\Pi))\}\in\mathcal{C}$ and the Poincar\'{e} action leaves the boundary and fall-off conditions invariant, then $\vec{\Pi}^{(n)}=0$ $\forall n\in\mathbb{N}$.
\end{corollary}
\begin{proof}
	This result is shown with the intermediate proposition in the proof of \ref{mass} with the additional assumption $\mathcal{O}_{r^{-n}}(D_i\vp)=0$ $\forall n\in\mathbb{N}$.
\end{proof}

\begin{corollary}\label{apimass}
	If $\{((A,\phi),(\pi,\Pi))\}\in\mathcal{C}$ and the Poincar\'{e} action leaves the boundary and fall-off conditions invariant, then
	\begin{align}
	\va_i&=A_i\vp-\frac{1}{e\|\vp\|^2}\vp\times\partial_i\vp+o(r^{-\mathbb{N}}),\\
	\vpi^i&=\pi^i\vp+\text{o}(r^{-\mathbb{N}}).
	\end{align}
\end{corollary}
\begin{proof}
	The first equation follows directly from $\mathcal{O}_{r^{-\mathbb{N}}}(D_i\vp)=0$ $\forall\{((A,\phi),(\pi,\Pi))\}\in\mathcal{C}$. The second equation follows by demanding Poincar\'{e} invariance of the first equation and using corollary \ref{pimass}.
\end{proof}

Before investigating the consequences of the proposition \ref{mass} and the corollaries \ref{pimass} and \ref{apimass} we finish this section by examining the Poincar\'{e} transformations of the leading components of the fields to find useful properties.\\

Let's introduce $\bar{D}_{\bar{a}}^{(0)}=\bar{\nabla}_{\bar{a}}+e\va^{(0)}_{\bar{a}}\times$ with $\bar{\nabla}$ being the Levi-Civita connection of the round metric $\bar{\gamma}$.
A well known identity from Riemannian geometry is: For a antisymmetric $(2,0)$ tensor $F^{ik}$, a Riemannian metric (or Lorentzian metric) $g$ and the corresponding Levi-Civita connection $\nabla$ the following identity holds\footnote{One can show this using the equations (3.1.9), (3.1.14) and (3.4.9) in \cite{WaldGR}.}:
\begin{align}
\nabla_k F^{ik}=\frac{1}{\sqrt{g}}\partial_k(F^{ik}\sqrt{g}).
\end{align}
In our case we have $\partial_{\bar{a}}(\sqrt{\bar{\gamma}}b\vec{F}_{(0)}^{\bar{a}r})=\sqrt{\bar{\gamma}}\bar{\nabla}_{\bar{a}}(b\vec{F}_{(0)}^{\bar{a}r})=\bar{\nabla}_a(\sqrt{\bar{\gamma}}b\vec{F}_{(0)}^{\bar{a}r})$. Additionally we have $\vec{F}^{\bar{a}r}=D_{(0)}^{\bar{a}}\va_{(0)}^r=\bar{D}^{\bar{a}}_{(0)}\va^{r}_{(0)}$, because $\va^{r}_{(0)}$ is a $\mathbb{R}^3$ valued function on $S^2$. We use this for the Poincar\'{e} transformations of the leading order Yang-Mills momenta. The leading order Poincar\'{e} transformations are:
\begin{align}
\delta_{(\xi^{\perp},\vect{\xi})}\va^{(0)}_r&=\frac{b}{\sqrt{\bar{\gamma}}}\bar{\pi}_r\vp_{\infty}+\vp_{\infty}Y^{\bar{a}}\partial_{\bar{a}}\bar{A}_r+\bar{A}_rY^{\bar{a}}\partial_{\bar{a}}\vp_{\infty}+e\bar{A}_r\vec{\zeta}^{(0)}\times\vp_{\infty}\label{poinar}\\
\delta_{(\xi^{\perp},\vect{\xi})}\va^{(0)}_{\bar{a}}&=\frac{b}{\sqrt{\bar{\gamma}}}\bar{\pi}_{\bar{a}}\vp_{\infty}+Y^{\bar{b}}\partial_{\bar{b}} \va^{(0)}_{\bar{a}}+\va^{(0)}_{\bar{b}}\partial_{\bar{a}} Y^{\bar{b}}-(\partial_{\bar{a}}\bar{\zeta})\vp_{\infty}-(\partial_{\bar{a}}\vec{\zeta}^{(0)\perp})^{\perp}+e\bar{A}_{\bar{a}}\vec{\zeta}^{(0)\perp}\times\vp_{\infty}\label{poinaa}\\
\delta_{(\xi^{\perp},\vect{\xi})}\vpi_{(0)}^r&=\bar{D}^{(0)}_{\bar{a}}(b\sqrt{\bar{\gamma}}\bar{D}_{(0)}^{\bar{a}}\va_r^{(0)})+\partial_{\bar{a}}(Y^{\bar{a}}\vpi^r_{(0)})+e\bar{\pi}^r\vec{\zeta}^{(0)}\times\vp_{\infty}\label{poinpir}\\
\delta_{(\xi^{\perp},\vect{\xi})}\vpi_{(0)}^{\bar{a}}&=\bar{D}^{(0)}_{\bar{b}}(b\sqrt{\bar{\gamma}}\vec{F}_{(0)}^{\bar{b}\bar{a}})+\partial_{\bar{b}}(Y^{\bar{b}}\vpi^{\bar{a}}_{(0)})-\vpi^{\bar{b}}\partial_{\bar{b}} Y^{\bar{a}}+e\bar{\pi}^{\bar{a}}\vec{\zeta}^{(0)}\times\vp_{\infty}\label{poinpia}\\
\delta_{(\xi^{\perp},\vect{\xi})}\vp_{\infty}&=Y^{\bar{a}}\partial_{\bar{a}}\vp_{\infty}+e\vec{\zeta}^{(0)}\times\vp_{\infty}.
\end{align}
We have seen before that these transformations are consistent with the boundary conditions $\va_i^{(0)}=\bar{A}_i\vp_{\infty}-(ea^2)^{-1}\vp_{\infty}\times\partial_i\vp_{\infty}$ and $\vpi^{(0)}_i=\bar{\pi}^i\vp_{\infty}$.\\

Let's look for important properties of the Poincar\'{e} transformations:\\

(i) It is easy to see that the radial and angular sector of $\bar{A}_i$ and $\bar{\pi}_i$ do not mix under Poincar\'{e} transformations.\\

(ii) The Poincar\'{e} transformations leave certain \textit{parity conditions} invariant. A parity transformation is given by the map
\begin{align}
a:\mathbb{R}^3\longrightarrow\mathbb{R}^3,\;\,x\mapsto -x.
\end{align}
This map is completely determined by a map
\begin{align}
a:S^2\longrightarrow S^2, (\theta,\varphi)\mapsto (\pi-\theta,\varphi+\pi).
\end{align}
\begin{defi}
	A function $f:S^2\longrightarrow\mathbb{R}$ is called \textit{even} under parity $\Leftrightarrow$ $a^*f=f$ $\Leftrightarrow$ $f(a(x))=f(x)$. Respectively $f$ is called \textit{odd} under parity $\Leftrightarrow$ $a^*f=-f$.
	
	A vector field $X\in\Gamma(TS^2)$ is called even under parity $\Leftrightarrow$ $Da(X)=X$ and odd under parity $\Leftrightarrow$ $Da(X)=-X$.
	A one form $\alpha\in\Omega^1(S^2)$ is called even under parity $\Leftrightarrow$ $a^*\alpha=\alpha$ and odd under parity $\Leftrightarrow$ $a^*\alpha=-\alpha$.
	In this way one generalizes this easily on all tensor fields.
\end{defi}
Let's make the following observation. In carthesian coordinates $\{x^i\}$ on $\mathbb{R}^3$ it is $a^*(\text{d}x^i)=-\text{d}x^i$, hence the components $\alpha_i$ of a one form $\alpha=\alpha_i\text{d}x^i$ are odd (even) $\Leftrightarrow$ $\alpha$ is even (odd). The same is true for the components of a vector field. The components $\omega_{ij}$ of a two form $\omega$ are odd (even) $\Leftrightarrow$ $\omega$ is odd (even). In these coordinates the derivative operator $\partial_i$ always changes the parity, because additionaly to $a^*(\text{d}x^i)=-\text{d}x^i$, $a^*(\text{d}\lambda)=\text{d}(a^*\lambda)$ for some form $\lambda$. Inspired by this observation we make the definition:
\begin{defi}\label{angularparity}
Call $\alpha_{\bar{a}}$, $X^{\bar{a}}$, $\omega_{\bar{a}\bar{b}}$ parity odd (even) $\Leftrightarrow$ $\alpha_i$, $X^i$, $\omega_{ij}$ are parity odd (even).\footnote{As $a^*\text{d}\theta=-\text{d}\theta$ and $a^*\text{d}\varphi=\text{d}\varphi$, $\alpha_{\bar{a}}$ being odd (even) means: $\alpha_{\theta}$ is even (odd) and $\alpha_{\varphi}$ is odd (even). For 2-forms, $\omega_{\bar{a}\bar{b}}$ being odd (even) means $\omega_{\theta\varphi}$ is even (odd).}
\end{defi}
This is a well posed definition, because the actual objects are the tensor fields with their parity conditions. This definition is just a different code for the same notion. In this code the derivative operator $\partial_{\bar{a}}$ always changes the parity.\\

It is easy to see that $b$ and $Y^{\bar{a}}$ are odd under parity.

Let's assume the asymptotic fields $\{\bar{A},\bar{\pi}\}$ have a definite parity, such that the parities of $\bar{A}_r$ and $\bar{\pi}^r$, as well as $\bar{A}_{\bar{a}}$ and $\bar{\pi}^{\bar{a}}$, are opposite.Take a look at equation \eqref{poinar}. Since $Y^{\bar{a}}\partial_{\bar{a}}$ is an parity-preserving operator and $\delta\va_r^{(0)}=(\delta\bar{A}_r)\vp_{\infty}+\bar{A}_r(\delta \vp_{\infty})$ you can take  $\vp_{\infty}$ out of the equation and with the property that $b$ and $Y^{\bar{a}}$ are odd, it is clear that this transformation preserves the parity of $\bar{A}_r$.

Take a look at equation \eqref{poinaa}. The same argument as before makes sure, that the first three terms preserve the parity. The last term might be a improper gauge transformation. If so, $\bar{\zeta}$ has to be of opposite parity then $\bar{A}_{\bar{a}}$. Alternatively we can allow a degree of freedom of opposite parity that has the form of a improper gauge transformation. We will expand on that in the next chapter.

One can write the transformation of the radial momentum as
\begin{align}
\delta_{(\xi^{\perp},\vect{\xi})}\vpi^r_{(0)}=\partial_{\bar{a}}(b\sqrt{\bar{\gamma}}\partial^{\bar{a}}\bar{A}^r)+\partial_{\bar{a}}(Y^{\bar{a}}\vpi_{(0)}^r)+e\bar{\pi}^r\vec{\zeta}^{(0)}\times\vp_{\infty},
\end{align}
where $\sqrt{\bar{\gamma}}$ is clearly even under parity.
This transformation preserves the parity condition of $\bar{\pi}^r$.

Take a look at equation \eqref{poinpia}. The term $\bar{D}_{\bar{b}}^{(0)}(b\sqrt{\bar{\gamma}}\vec{F}_{(0)}^{\bar{b}\bar{a}})$ might be problematic, because the covariant derivative $\bar{D}_{\bar{b}}^{(0)}$ mixes in general different parities.\footnote{In the paper \cite{Tanzi:2020}, concerning Yang-Mills theory without the Higgs field, this fact excludes the choice of parity conditions that were natural in electromagnetism (\cite{Henneaux-ED}).}
Here, using \eqref{asymptoticF}:
\begin{align}
\bar{D}_{\bar{b}}^{(0)}(b\sqrt{\bar{\gamma}}\vec{F}^{\bar{b}\bar{a}}_{(0)})=\partial_{\bar{b}}(b\sqrt{\bar{\gamma}}\bar{F}^{\bar{b}\bar{a}})\vp_{\infty}=\partial_{\bar{b}}(b\sqrt{\bar{\gamma}}((\partial^{\bar{b}}\bar{A}^{\bar{a}}-\partial^{\bar{a}}\bar{A}^{\bar{b}})+(ea^4)^{-1}\vp_{\infty}\cdot(\partial^{\bar{b}}\vp_{\infty}\times\partial^{\bar{a}}\vp_{\infty}))).
\end{align}
The first term preserves the parity conditions. The second term is more delicate to handle and will be treated in the section \ref{parity}.

\section{The Higgs mechanism}
Let's study the consequences of the proposition \ref{mass} and the corollaries \ref{pimass} and \ref{apimass}. The statements of the corollaries expressed differently are:

(i) \textit{The dynamics of physical states of the Higgs field in $\SU(2)$ Yang-Mills-Higgs theory in 4 dimensions falls off faster then any inverse monomial in the radius $r$.}

(ii) \textit{The physical non-abelian degrees of freedom of the Yang-Mills field in $\SU(2)$ Yang-Mills-Higgs theory in 4 dimensions fall off faster then any inverse monomial in the radius $r$.}

Both of these statements are the consequence of demanding Poincar\'{e} invariance of the theory. \textit{Physical} means that these statements are only true for states on the constraint surface.

As we actually only care about physical states and the results of the corollaries have the right shape (to not induce degeneracies of the symplectic form), we demand
\begin{align}
\vec{\Pi}=\text{o}(r^{-\mathbb{N}}),\;\,\vec{A}_i=A_i\vp-\frac{1}{e\|\vp\|^2}\vp\times\partial_i\vp+\va_i^{\perp},\;\vpi^i=\pi^i\vp+\vpi^{\perp i},\,\;\va_i^{\perp},\vpi^{\perp i}\in \text{o}(r^{-\mathbb{N}})
\end{align}
to be boundary conditions, where the $\perp$ is defined to be the projection perpendicular to $\vp$. The conditions for $\vpi^i$ and $\va_i$ are already Poincar\'{e} invariant, while the condition for $\vec{\Pi}$ has to be expanded by yet another condition, which is
\begin{align}
\|\vp\|^2=a^2+\text{o}(r^{-\mathbb{N}}),
\end{align}
because by proposition \ref{mass},
\begin{align}
\delta_{(\xi^{\perp},0)}\vec{\Pi}=D_i(\xi^{\perp}D^i\vp)-\xi^{\perp}\frac{\partial V}{\partial\vp}+e\vec{\zeta}\times\vec{\Pi}=\text{o}(r^{-\mathbb{N}})\nonumber\\
\Leftrightarrow\;\frac{\partial V}{\partial\vp}=4(\|\vp\|^2-a^2)\vp=\text{o}(r^{-\mathbb{N}})\;\Leftrightarrow \|\vp\|^2=a^2+\text{o}(r^{-\mathbb{N}}).
\end{align}
It is possible to demand even stronger boundary conditions, because we find out in the following that any order $\vp^{(n)}$ of the asymptotic expansion of $\vp$ for $n\geq 1$ is a proper gauge degree of freedom.
\begin{prop}\label{goldstone}
Any $\vp$ is a result of a proper gauge transformation of a $\vp'=\vp_{\infty}+o(r^{-\mathbb{N}})$.
\end{prop}
\begin{proof}
	Let's perform a sequence of infinitesimal proper gauge transformations starting at $\vp':=\vp_{\infty}+\text{o}(r^{-\mathbb{N}})$ which fixes $\vp_{\infty}$. At first order,
	\begin{align}
	\delta_{\vec{\epsilon}^{\,'}}\vp'=e\vec{\epsilon}^{\,'}\times\vp_{\infty}+\text{o}(r^{-\mathbb{N}})=r^{-1}e\vec{\epsilon}^{\,'(1)}\times\vp_{\infty}+r^{-2}e\vec{\epsilon}^{\,'(2)}\times\vp_{\infty}+...\;.
	\end{align}
	Let's call $e\vec{\epsilon}^{\,'(1)}\times\vp_{\infty}=:\vp^{(1)}_{\text{inf}}$, $e\vec{\epsilon}^{\,'(2)}\times\vp_{\infty}=:\vp^{(2)\perp}_{\text{inf}}$, ..., where the "inf" is a label for infinitesimal. This is justified because $e\vec{\epsilon}^{\,'(1)}\times\vp_{\infty}$ lies in the tangent space $T_0\{\vp^{(1)}|\vp^{(1)}\cdot\vp_{\infty}=0\}_{\vp_{\infty}}\cong\{\vp^{(1)}|\vp^{(1)}\cdot\vp_{\infty}=0\}_{\vp_{\infty}}$, $e\vec{\epsilon}^{\,'(2)}\times\vp_{\infty}$ lies in the tangent space $T_0\{\vp^{(2)\perp}|\vp^{(2)\perp}\cdot\vp_{\infty}=0\}_{\vp_{\infty}}\cong\{\vp^{(2)\perp}|\vp^{(2)\perp}\cdot\vp_{\infty}=0\}_{\vp_{\infty}}$,... . Let's compute the second order:
	\begin{align}
	\delta^2_{\vec{\epsilon}^{\,'}}\vp'&=e\vec{\epsilon}^{\,'}\times\vp_{\infty}+e\vec{\epsilon}^{\,'}\times\delta_{\vec{\epsilon}^{\,'}}\vp'+\text{o}(r^{-\mathbb{N}})\nonumber\\
	&=r^{-1}e\vec{\epsilon}^{\,'(1)}\times\vp_{\infty}+r^{-2}e\vec{\epsilon}^{\,'(2)}\times\vp_{\infty}+...\nonumber\\&+r^{-2}e\vec{\epsilon}^{\,'(1)}\times\vp^{(1)}_{\text{inf}}+...\;.
	\end{align}
	$e\vec{\epsilon}^{\,'(1)}\times\vp^{(1)}_{\text{inf}}$ has the property $e(\vec{\epsilon}^{\,'(1)}\times\vp^{(1)}_{\text{inf}})\cdot\vp_{\infty}=-\|\vp^{(1)}_{\text{inf}}\|^2.$ This is the infinitesimal version of the fact that by $\|\vp\|^2=a^2+\text{o}(r^{-\mathbb{N}})$, $\vp^{(2)}\cdot\vp_{\infty}=-\|\vp^{(1)}\|^2$.
	
	We would like to make sure that by the successive application of a gauge transformation $\delta_{\vec{\epsilon}^{\,'}}$ every possible value of $\vp$ with fixed $\vp_{\infty}$ can be reached.
	
	By the condition $\|\vp\|^2=a^2+\text{o}(r^{-\mathbb{N}})$, at every order $\sim r^{-k}$, $\vp^{(k)}\cdot\vp_{\infty}$ is determined by lower orders of $\vp$, while the $\vp^{(k)\perp}$ are the free parameters\footnote{The value of $\vp'^{(k)\perp}$ determines again through expressions of the form $\vp'^{(k)}\cdot\vp'^{(k')}$, $k'\geq 1$ higher orders of $\vp'\cdot\vp_{\infty}$.}. Every possible value of $\vp^{(k)\perp}$ is reached via a successive application of $r^{-k}\vec{\epsilon}^{\,'(k)}\times$ onto $\vp'$, because $\{\vec{\epsilon}^{\,'(k)}\times\vp_{\infty}|\vec{\epsilon}^{\,'(k)}\}_{k\geq 1}\cong\Gamma(\vp_{\infty}^*(TS^2))$ for any $k$.
	
	Additionally $\delta_{\vec{\epsilon}^{\,'}}\|\vp\|^2=2e(\vec{\epsilon}^{\,'}\times\vp)\cdot\vp=0$ for any $\vec{\epsilon}^{\,'}$ and $\vp$. Hence we can reach any $\vp$ with $\|\vp\|^2=a^2+\text{o}(r^{-\mathbb{N}})$ with a proper gauge transformation from $\vp'=\vp_{\infty}+\text{o}(r^{-\mathbb{N}})$.
\end{proof}
\begin{corollary}\label{theasymptoticgaugespace}
	$\vp^{(n)}$ is a proper gauge degree of freedom $\forall n\geq 1$.
\end{corollary}
\begin{proof}	
	The state of the system at time $t$ does not depend on the initial values of $\vp^{(n)}$ ($n\geq 1$) at time $t=0$, because:
	\begin{itemize}
		\item The time evolution of the $\vp^{(n)}$ is a proper gauge transformation, because $\vec{\Pi}^{(n)}=0$ $\forall n\in\mathbb{N}$.
		\item By the proposition \ref{goldstone}, the state at $t=0$ does not depend on the $\vp^{(n)}$.
	\end{itemize} 
	This is a necessary condition for a variable to be a proper gauge degree of freedom.
\end{proof}
Having that we can sharpen the boundary conditions by demanding
\begin{align}
	\vp=\vp_{\infty}+\text{o}(r^{-\mathbb{N}}).
\end{align}
This conditions restricts the allowed set of gauge parameters to
\begin{align}
	\vec{\epsilon}=\epsilon\vp_{\infty}+\vec{\epsilon}_{(0)}^{\perp}+\text{o}(r^{-\mathbb{N}}).
\end{align}
Let's collect the results and reframe them with respect to the usual notion of the Higgs mechanism. We found that the physical degrees of freedom of the Higgs
field and the non-abelian part of the Yang-Mills field with the exception of the vacuum degrees of freedom ($\vp_{\infty}$, $-(ea^2)^{-1}(\vp_{\infty}\times\partial_i\vp_{\infty})$) fall off faster then any inverse monomial in $r$. This fall-off suggests by comparison to the behaviour of a massive scalar field \cite[Section 3.1]{Tanzi:2021} that these degrees of freedom are massive.

Actually the \textit{mass} of a field is defined via the prefactor of the quadratic term of the field in the Lagrangian. So, let's make a precise statement about the massive fields in this theory. Let's decompose $\vp=\vp_{\infty}+\phi\,\vp_{\infty}+\vp^{\perp}$, where $\phi,\vp^{\perp}\in \text{o}(r^{-\mathbb{N}})$ and $\vp^{\perp}\perp\vp_{\infty}$. Here we have set the partial gauge $\vp^{(n)}=0$ $\forall n\geq 1$. The potential term in the Lagrangian takes the form
\begin{align}
	V(\vp)=4a^2\phi^2+2a\,\phi\|\phi\,\vp_{\infty}+\vp^{\perp}\|^2+\|\phi\vp_{\infty}+\vp^{\perp}\|^4.
\end{align}
The first term is the mass term. It should be noted that only the component pointing in the direction of the vacuum has a mass term. The perpendicular component $\vp^{\perp}$ is the massless \textit{Goldstone boson} of the spontaneously broken $\SU(2)$. $\vp^{\perp}$ is actually massless, because it has no mass term in $V(\vp)$, as well as no mass term in $D_i\vp\cdot D^i\vp$. A possible mass term is not an interaction term, so we can assume for a moment that $\va_i=-(ea^2)^{-1}\vp_{\infty}\times\partial_i\vp_{\infty}$\footnote{Even if this is topologically not possible everywhere in $\mathbb{R}^3$.} to calculate
\begin{align}
	D_i\vp&=\partial_i\vp-a^{-2}(\vp_{\infty}\times\partial_i\vp_{\infty})\times\vp\nonumber\\&=(\partial_i\phi)\vp_{\infty}+\phi\,\partial_i\vp_{\infty}+\partial_i\vp^{\perp}-\phi\partial_i\vp_{\infty}-a^{-2}\vp_{\infty}(\partial_i\vp^{\perp}\cdot\vp_{\infty})\nonumber\\&=(\partial_i\vp)\vp_{\infty}+(\partial_i\vp^{\perp})^{\perp}.\\
	\Rightarrow\; D_i\vp\cdot D^i\vp&=a^2(\partial_i\phi)(\partial^i\phi)+(\partial_i\vp^{\perp})^{\perp}\cdot(\partial_i\vp^{\perp})^{\perp}.
\end{align}
These are only the kinetic energy terms. Hence $\vp^{\perp}$ is massless.
\begin{prop}
There exists a gauge where $\vp^{\perp}=0$.
\end{prop}
\begin{proof}
Clearly there exists a smooth $U:\mathbb{R}^3\longrightarrow\SU(2)$ with $U\in o(r^{-\mathbb{N}})$ (a proper gauge transformation) such that $U(x)\vp(x)=\phi'(x)\vp_{\infty}(x)$. 
\end{proof}
Hence the Goldstone boson is no physical degree of freedom, but pure proper gauge.\\

Let's find the mass term of the non-abelian part of the Yang-Mills field in the Lagrangian. For that take a look to the terms which are quadratic in $\va_i$ in
\begin{align}
	D_i\vp\cdot D^i\vp&=(...)+e^2(\va_i\times\vp)\cdot(\va^i\times\vp)\nonumber\\&=(...)+e^2(\va_i\cdot\va^i)\|\vp\|^2-(\va_i\cdot\vp)(\va^i\cdot\vp)\nonumber\\
	&=(...)+(a^2+o(r^{-\mathbb{N}}))(\va_i^{\perp}\cdot\va^{i\perp}-2(ea^2)^{-1}\va_i^{\perp}\cdot(\vp_{\infty}\times\partial^i\vp_{\infty})\nonumber\\&+(ea^2)^{-2}(\vp_{\infty}\times\partial_i\vp_{\infty})\cdot(\vp_{\infty}\times\partial^i\vp_{\infty})).
\end{align}
The term $a^2\va_i^{\perp}\cdot\va^{i\perp}$ is the mass term of the non-abelian part of the Yang-Mills field.\\

Let's finish this section by collecting the boundary conditions at this point:
\begin{align}
	\va_r&=(r^{-1}\bar{A}_r+r^{-2}A_r^{(1)}+r^{-3}A_r^{(2)}...)\vp_{\infty}+o(r^{-\mathbb{N}}),\nonumber\\ \vpi^r&=(\bar{\pi}^{r}+r^{-1}\pi^{r}_{(1)}+r^{-2}\pi^{r}_{(2)}+...)\vp_{\infty}+o(r^{-\mathbb{N}}),\nonumber\\
	\va_{\bar{a}}&=-(ea^2)^{-1}\vp_{\infty}\times\partial_{\bar{a}}\vp_{\infty}+(\bar{A}_{\bar{a}}+r^{-1}A_{\bar{a}}^{(1)}+r^{-2}A_{\bar{a}}^{(2)}+...)\vp_{\infty}+o(r^{-\mathbb{N}}),\nonumber\\
	\vpi^{\bar{a}}&=(r^{-1}\bar{\pi}^{\bar{a}}+r^{-2}\pi^{\bar{a}}_{(1)}+r^{-3}\pi^{\bar{a}}_{(2)}+...)\vp_{\infty}+o(r^{-\mathbb{N}})\nonumber\\
	\vp\;&=\vp_{\infty}+o(r^{-\mathbb{N}})\nonumber\\
	\vec{\Pi}\;&=o(r^{-\mathbb{N}}),\label{fatboundary}\nonumber\\
	\vec{\epsilon}\;&=\vec{\epsilon}_{(0)}^{\perp}+(\bar{\epsilon}+r^{-1}\epsilon^{(1)}+r^{-2}\epsilon^{(2)}+...)\vp_{\infty}+o(r^{-\mathbb{N}}).
\end{align}
It is easy to see that these boundary conditions are invariant under the gauge transformation generated by the given $\epsilon$.
\begin{bem}\label{trickyspatial}
	Spatial translations of the Higgs field do not leave $\vp^{(1)}=0$ invariant, i.e.\
	\begin{align}
		\delta_{(0,\vect{\xi})}\vp^{(1)}=\mathcal{O}_{r^{-1}}(\xi^i\partial_i\vp)=\frac{1}{r}\bar{\gamma}^{\bar{a}\bar{b}}\partial_{\bar{b}}W\partial_{\bar{a}}\vp_{\infty}.
	\end{align}
	Hence we have to choose a slightly different (but equivalent) action of the Poincar\'{e} algebra for the translations to preserve the boundary conditions.
	If the spatial translation is paired with a proper gauge transformation generated by
	\begin{align}
		\vec{\epsilon}^{(1)}=-(ea^2)^{-1}\bar{\gamma}^{\bar{a}\bar{b}}\partial_{\bar{b}}W\,\vp_{\infty}\times\partial_{\bar{a}}\vp_{\infty},
	\end{align}
	then $\delta_{(0,\vect{\xi})}\vp^{(1)}+e\vec{\epsilon}^{(1)}\times\vp_{\infty}=0$. We have to choose the spatial translations to be $(\delta_W)^{\text{new}}:=\delta_W+\delta_{\vec{\epsilon}^{(1)}}$ before we perform the symplectic reduction to allow for an action of the Poincar\'{e} algebra on the set of variables restricted by the boundary conditions \eqref{fatboundary}.
\end{bem}

\section{Global structure of the phase space}\label{globalstructure}
On the path to construct a phase space we have yet to make sure that the symplectic form is finite. For that it helps to understand the global structure of the \textit{proto} phase space $\mathbb{P}^{\text{proto}}$ (defined by the conditions \eqref{fatboundary}).

In \ref{vectorbundle} we show that $\mathbb{P}^{\text{proto}}$ fibres over the space of Higgs vacua $\vp_{\infty}\in\mathcal{F}_{\{\vp_{\infty}\}}:=C^{\infty}(S^2,S^2)$ and $\mathcal{F}_{\{\vp_{\infty}\}}$ has connected components labelled by the magnetic charge $m\in\mathbb{Z}$. Let's denote the fibres by $\mathbb{P}^{\text{proto}}_{\vp_{\infty}}$. We will use this notation with the subscript of $\vp_{\infty}$ also for the fibres in the redefined spaces on our path to a well defined phase space.\\

In the following we will show that for each $m\in\mathbb{Z}$ the variation of $\vp_{\infty}\in\mathcal{F}_{\{\vp_{\infty}\}}^{(m)}$ is a proper gauge degree of freedom in $\mathbb{P}^{\text{proto}}$. First of all, based on the conditions \eqref{fatboundary} it is clear that $\vec{\epsilon}^{\perp}_{(0)}$ generate proper gauge transformations.
\begin{prop}\label{epsilonorbit}
	For any two $\vp_{\infty},\vp_{\infty}'\in\mathcal{F}^{(m)}_{\{\vp_{\infty}\}}$ and $m\in\mathbb{Z}$: $\exists$ a proper gauge transformation $U:\mathbb{R}^3\cong\Sigma\longrightarrow\SU(2)$ such that $U\rhd\mathbb{P}^{\text{proto}}_{\vp_{\infty}}=\mathbb{P}^{\text{proto}}_{\vp_{\infty}'}$ and $U\rhd|_{\mathbb{P}^{\text{proto}}_{\vp_{\infty}}}$ is injective.
\end{prop}
\begin{proof}
	As the first step we see that every tangent space $T_{\vp_{\infty}}\mathcal{F}_{\{\vp_{\infty}\}}$  consists only of proper gauge transformation generators.
	
	Let $\vp_{\infty}\in\mathcal{F}_{\{\vp_{\infty}\}}^{(m)}$ for some $m\in\mathbb{Z}$. Take a tangent vector $v\in T_{\vp_{\infty}}\mathcal{F}_{\{\vp_{\infty}\}}^{(m)}$. Using theorem 7.6. and corollary 7.6 in \cite{wockel},
	\begin{align}
		T_{\vp_{\infty}}\mathcal{F}_{\{\vp_{\infty}\}}^{(m)}\cong \{\delta\vp_{\infty}\in C^{\infty}(S^2,TS^2)|\pi(\delta\vp_{\infty})(x)=\vp_{\infty}(x)\;\forall x\in S^2\},
	\end{align}
	where $\pi:TS^2\longrightarrow S^2$ is the natural projection.
	
	Hence $v$ can be represented as a smooth function $\vec{v}:S^2\longrightarrow \mathbb{R}^3$ such that $\vec{v}(x)\perp\vp_{\infty}(x)$ $\forall x\in S^2$. Take an asymptotic proper gauge transformation generated by $\vec{\epsilon}_{(0)}^{\perp}$ that has the property $\vec{\epsilon}_{(0)}^{\perp}(x)\perp\vp_{\infty}(x)$ $\forall x\in S^2$. Take
	\begin{align}
		\vec{\epsilon}_{(0)}^{\perp}:=\frac{1}{ea^2}\vp_{\infty}\times\vec{v}.\;\Rightarrow\;\delta_{\vec{\epsilon}^{\perp}_{(0)}}\vp_{\infty}=e\vec{\epsilon}^{\perp}_{(0)}\times\vp_{\infty}=\vec{v}.
	\end{align}
	$\Rightarrow$ the tangent space $T_{\vp_{\infty}}\mathcal{F}_{\{\vp_{\infty}\}}^{(m)}$ is spanned by $\vp_{\infty}$-dependent proper gauge transformations.\\
	
	Take $\vp_{\infty},\vp_{\infty}'\in\mathcal{F}^{(m)}_{\{\vp_{\infty}\}}$. Then there exists a smooth curve $\gamma:[0,1]\longrightarrow\mathcal{F}_{\{\vp_{\infty}\}}^{(m)}$ with $\gamma(0)=\vp_{\infty}$ and $\gamma(1)=\vp_{\infty}'$. Every tangent vector to the curve is a field dependent generator of a proper gauge transformation. By arguments used in the proof of \ref{globalstructure}, there exists a smooth homotopy between $\vp_{\infty}$ and $\vp_{\infty}'$. Let $U:\mathbb{R}^3\longrightarrow\SU(2)$ be a gauge transformation that is constant in radial direction. Take the pullback $i^*_{S^2}(U):S^2\longrightarrow\SU(2)$. The action on the Higgs vacuum is given by a rotation matrix $\text{D}(x)$ such that $(i^*_{S^2}(U)\rhd\vp_{\infty})(x)=\text{D}(x)\vp_{\infty}(x)$. As $\vp_{\infty}'=\text{D}\vp_{\infty}$, $\text{D}$ is smooth, and because of the homotopy, $\text{D}$ lies in the identity component of $C^{\infty}(S^2,\text{SO}(3))$. Take the distinguished lift (along the fibres) of $\gamma$ into $\mathbb{P}^{proto}$ which is characterized by the tangent vectors $X^{\vec{\epsilon}^{\perp}_{(0)}}$ (the action of the gauge generators on the canonical variables), where $\vec{\epsilon}^{\perp}_{(0)}$ are tangent to $\gamma$. $\Rightarrow$ as all the $\vec{\epsilon}^{\perp}_{(0)}$ are proper gauge generators, $\text{D}$ represents a proper gauge transformation with $\text{D}\rhd(\mathbb{P}^{\text{proto}}_{\vp_{\infty}})\subset\mathbb{P}^{\text{proto}}_{\vp_{\infty}'}$. It remains to show that this map is bijective. This is easy to see for the variables $\{\vp,\vec{\Pi},\vpi^i\}$, because of the easy form of the action. For example $\text{D}\vpi^i=\pi^i\vp_{\infty}'+\text{D}\vpi^{i\perp}$ with $\text{D}\vpi^{i\perp}\perp\vp_{\infty}'$, which is a bijective map $D:(\pi^i,\pi^{i\perp})\mapsto(\pi^i,D\vpi^{i\perp})$ onto its image. Exactly similar it is for the variables $\vp$ and $\vec{\Pi}$. Let's take a look at the action on $\va_i$. We do not compute the action directly, but use the action on $\vec{F}_{ij}$ to get the information we need.\\
	$\text{D}\rhd\vec{F}_{ij}=\text{D}\vec{F}_{ij}$. With the boundary conditions \eqref{fatboundary} we have
	\begin{align}
	\vec{F}_{ij}=((\partial_i A_j-\partial_j A_i)+(ea^4)^{-1}(\text{Vol}_{\vp_{\infty}})_{ij})\vp_{\infty}+\vec{F}^{\perp}_{ij},
	\end{align}
where $\vec{F}_{ij}^{\perp}=\text{o}(r^{-\mathbb{N}})$ and $(\text{Vol}_{\vp_{\infty}})_{ij}:=\vp_{\infty}\cdot(\partial_i\vp_{\infty}\times\partial_j\vp_{\infty})$. Using that we get the equation:
\begin{align}
	((\partial_i A_j-\partial_j A_i)+(ea^4)^{-1}(\text{Vol}_{\vp_{\infty}})_{ij})\vp_{\infty}'+\text{D}\vec{F}_{ij}^{\perp}=\text{D}\vec{F}_{ij}\nonumber\\ \overset{!}{=}((\partial_i A'_j-\partial_jA'_i)+(ea^4)^{-1}(\text{Vol}_{\vp_{\infty}'})_{ij})\vp_{\infty}'+\text{D}\vec{F}_{ij}^{\perp},\nonumber\\ \Leftrightarrow \text{d}(A'-A)=(ea^2)^{-1}a^{-2}(\text{Vol}_{\vp_{\infty}}-\text{Vol}_{\vp_{\infty}'}).
\end{align}
	Hence by the equation above, $A'_i=A_i+\partial_i\Phi+A^{\text{gau}}_i$, where $\text{d}A^{\text{gau}}=(ea^2)^{-1}a^{-2}(\text{Vol}_{\vp_{\infty}}-\text{Vol}_{\vp_{\infty}'})$ and $A^{\text{gau}}$ does not include an exact term.\footnote{One can show also geometrically that the two form $\text{Vol}_{\vp_{\infty}}-\text{Vol}_{\vp_{\infty}'}$ is exact.} Additionally $\partial_i\Phi=0$ because we chose the distinguished lift of the curve $\gamma$ in the bundle that has $X^{\vec{\epsilon}^{\perp}_{(0)}}$ as tangent vectors. A term $\partial_i\Phi$ would only be generated if we had added a gauge transformation that stabilizes the fibre $\mathbb{P}^{\text{proto}}_{\gamma(t)}$ for some $t\in[0,1]$.
	
	So $\text{D}:(A_i,\vec{A}_i^{\perp})\mapsto(A_i+A_i^{\text{gau}},\text{D}\va_i^{\perp})$, where $A_i^{\text{gau}}$ is fixed by $\text{D}$. This map is clearly bijective onto its image.
\end{proof}

\begin{corollary}\label{thegaugemanifold}
	Every connected component of $\mathcal{F}_{\{\vp_{\infty}\}}$ parametrizes a proper gauge degree of freedom.
\end{corollary}
\begin{proof}
	We use the same argument as in the proof of the corollary \ref{theasymptoticgaugespace}. The time evolution of $\vp_{\infty}$ is a proper gauge transformation and by the proposition \ref{epsilonorbit}, any two $\vp_{\infty},\vp_{\infty}'$ that lie in the same connection component of $\mathcal{F}_{\vp_{\infty}}$ are connected via a proper gauge transformation.
\end{proof}

\section{Parity conditions}\label{parity}
Let's take a look at the formal symplectic potential on $\mathbb{P}^{\text{proto}}$. This is at a point $p=(A,\pi,\phi,\Pi)\in\mathbb{P}^{\text{proto}}$:
\begin{align}
	\Theta_p=\int\text{d}^3x\,\vpi^i\cdot\extder\va_i+\int\text{d}^3x\,\vec{\Pi}\cdot\extder\vp.
\end{align}
Evaluating the potential on a tangent vector $v\in T_p\mathbb{P}^{\text{proto}}$ of a curve $\gamma(s)$ with $\gamma(0)=p$ while using the boundary conditions \eqref{fatboundary}:
\begin{align}
	\Theta_p(v)&=\int\text{d}^3x\,\vpi^i\cdot\partial_s\va_i+\int\text{d}^3x\,\vec{\Pi}\cdot\partial_s\vp\nonumber\\
	&=(\text{finite})+\int\text{d}^3x\,\frac{1}{r^3}\vpi^{i}_{(0)}\cdot\partial_s\va_i^{(0)}.
\end{align}
The non-finite term is logarithmically divergent.
Consider the non-finite term of the symplectic potential:\\

\underline{r-component}:
\begin{align}
	\vpi_r^{(0)}\cdot\partial_s\va^r_{(0)}=\bar{\pi}^r\bar{A}_r \vp_{\infty}\cdot Y^{\bar{a}}\partial_{\bar{a}}\vp_{\infty}+a^2\bar{\pi}^r\partial_s\bar{A}_r=a^2\bar{\pi}^r\partial_s\bar{A}_r.
\end{align}
\underline{$\bar{a}$-component}:
\begin{align}
	\label{SympotA}
	\vpi_{\bar{a}}^{(0)}\cdot\partial_s\va^{\bar{a}}_{(0)}&=\bar{\pi}^{\bar{a}}\vp_{\infty}\cdot\partial_s(\vp_{\infty}\bar{A}_{\bar{a}}-\frac{1}{ea^2}\vp_{\infty}\times\partial_{\bar{a}}\vp_{\infty})\nonumber\\
	&=a^2\bar{\pi}^{\bar{a}}\partial_s\bar{A}_{\bar{a}}-\frac{1}{ea^2}\bar{\pi}^{\bar{a}}\vp_{\infty}\cdot(\partial_s\vp_{\infty}\times\partial_{\bar{a}}\vp_{\infty}+\vp_{\infty}\times\partial_s\partial_{\bar{a}}\vp_{\infty})\nonumber\\
	&=a^2\bar{\pi}^{\bar{a}}\partial_s\bar{A}_{\bar{a}}-\frac{1}{ea^2}\bar{\pi}^{\bar{a}}\vp_{\infty}\cdot(\partial_s\vp_{\infty}\times\partial_{\bar{a}}\vp_{\infty}).
\end{align}
The class of functions that makes the integral $\int\text{d}^3 x\frac{1}{r^3}\vpi^{(0)}_i\cdot\partial_s\va^i_{(0)}$ vanish might be rather non trivial, but it certainly vanishes if the integrand is an odd function on $S^2$. We will stick with choosing the integrand to be an odd function, because as we will see this will actually be necessary in the angular components (\ref{pieven}), it is compatible with Poincar\'{e} transformations (section \ref{aspoin}) and it it allows for physically relevant solutions.

Let's propose \textit{parity conditions} on $\bar{A}_r$, $\bar{A}_{\bar{a}}$, $\bar{\pi}^r$ and $\bar{\pi}^{\bar{a}}$.
Consider the r-component first. $\bar{\pi}^r\partial_s\bar{A}_r$ has to be odd. A constraint on the possible choices is that the electrically charged Dyon should be a state in phase space.

The Julia-Zee Dyon \cite[equation 2.10, III. boundary conditions]{Julia.Zee:1975} has the radial electric field $\vec{E}_r^{(0)}=\bar{E}_r\vec{n}$, where $\bar{E}_r=C\frac{1}{er^2}$ and $\vec{n}=\vp_{\infty}$ is the unit normal vector on $S^2$ (see \ref{Dyon}). $\bar{E}_r$ is the electric field of the Coulomb solution of electromagnetism. Calculating the corresponding canonical momentum in spherical coordinates gives $\bar{\pi}^r=\text{sin}(\theta)$. The change of $\theta$ after applying the antipodal map is $\theta\mapsto\pi-\theta$ and for $\theta\in[0,\pi]$: $\text{sin}(\pi-\theta)=\text{sin}(\theta)$, which means $\bar{\pi}^r$ is a even function.

Therefore one chooses:
\begin{align}
	\bar{A}_r(-x)=-\bar{A}_r(x)\, ,\;\;\bar{\pi}^r(-x)=\bar{\pi}^r(x).
\end{align}
Since $\bar{A}_r$ as well as $\bar{\pi}^r$ are gauge invariant, these are appropriate boundary conditions.

These conditions also determine what the improper gauge transformations of the theory are:
\begin{align}\label{impgenerator}
	\Gamma[\vec{\epsilon}]\approx -\oint\text{d}^2 x\, \bar{\epsilon}\bar{\pi}^r=-\oint\text{d}^2 x\,\bar{\epsilon}_{\text{even}}\bar{\pi}^r,
\end{align}
i.e. $\bar{\epsilon}_{\text{improper}}=\bar{\epsilon}_{\text{even}}$.\\

Let's consider the angular components \eqref{SympotA}.
We choose the condition
\begin{align}
	\bar{\pi}^{\bar{a}}_{\text{even}}=0.\label{pieven}
\end{align}
This condition might actually necessary\footnote{This is because of the following reasons. In the appendices B.1 and B.2 of \cite{Henneaux-ED} it was shown that in electromagnetism this condition is necessary to have not a divergence in the magnetic field as one approaches null infinity by an infinite sequence of boosts of the Cauchy surface. The same argumentation applies here to the case $\vp_{\infty}=\text{const}$, because in that case the asymptotic form of the action is indistinguishable from the electromagnetic action. By a proper gauge transformation (prop. \ref{epsilonorbit}) this applies to the whole $m=0$ sector. It is a conjecture that this argumentation also applies to the $m\neq 0$ sector.}.
The condition $\bar{\pi}^{\bar{a}}_{\text{even}}=0$ has to be invariant under boosts. In general $\delta_b\bar{\pi}^{\bar{a}}$ involves a term $\partial_{\bar{b}}(b\text{Vol}_{\vp_{\infty}}^{\bar{b}\bar{a}})$. If $\text{Vol}_{\vp_{\infty}}^{\bar{b}\bar{a}}:=\vp_{\infty}\cdot(\bar{\nabla}^{\bar{b}}\vp_{\infty}\times\bar{\nabla}^{\bar{a}}\vp_{\infty})$ is not odd this breaks the parity condition. But apparently there is no problem at all, which we will see in the remainder of this section. Let's regard the two cases $m=0$ and $m\neq 0$.\\

\underline{m=0}: By the corollary \ref{thegaugemanifold}, for each connection component $(\mathbb{P}^{\text{proto}})^{(m)}$ the base manifold $\mathcal{F}^{(m)}_{\{\vp_{\infty}\}}$ parametrizes a proper gauge degree of freedom. Let's phrase the parity condition firstly in the simple gauge $\vp_{\infty}=\vec{\tau}_3=(0,0,1)$, where $\va_{\bar{a}}^{(0)}=\bar{A}_{\bar{a}}\vec{\tau}_3$ and $\vpi^{\bar{a}}_{(0)}=\bar{\pi}^{\bar{a}}\vec{\tau}_3$. In order to make $\Omega_{\vec{\tau}_3}<\infty$, while $\bar{\pi}^{\bar{a}}_{\text{even}}=0$, choose
\begin{align}
	\bar{A}_{\bar{a}}=\bar{A}_{\bar{a}}^{\text{even}}+\partial_{\bar{a}}\Phi_{\text{even}}\;\text{and}\;\bar{\pi}^{\bar{a}}=\bar{\pi}^{\bar{a}}_{\text{odd}}.
\end{align}
Then it follows $\vec{F}_{(0)}^{\bar{a}\bar{b}}=2\bar{\nabla}^{[\bar{a}}\bar{A}_{\text{even}}^{\bar{b}]}\vec{\tau}_3$ $\Rightarrow$ $\delta_b\bar{\pi}_{\bar{a}}$ stays odd. These conditions let the angular integral over \eqref{SympotA} vanish, because the second term of \eqref{SympotA} vanishes in this gauge and $\partial_{\bar{a}}\bar{\pi}^{\bar{a}}=0$ is a boundary condition (see \eqref{boundary} and use $\bar{\Pi}=0$).

Transfer these conditions onto another fibre $\mathbb{P}^{\text{proto}}_{\vp_{\infty}}\subset(\mathbb{P}^{\text{proto}})^{(0)}$ by a proper gauge transformation $U\rhd$, being the canonical lift of the mapping $\vec{\tau}_3\mapsto\vp_{\infty}$ (which is exactly the $U$ from the proof of proposition \ref{epsilonorbit}). Then (using the proof of prop. \ref{epsilonorbit}),
\begin{align}
	U\rhd\va_{\bar{a}}^{(0)}&=\tilde{A}_{\bar{a}}\vp_{\infty}-(ea^2)^{-1}\vp_{\infty}\times\partial_{\bar{a}}\vp_{\infty}\;\text{with}\;\tilde{A}_{\bar{a}}=\bar{A}_{\bar{a}}+A_{\bar{a}}^{\text{gau}},\\
	U\rhd\vpi^{\bar{a}}_{(0)}&=\bar{\pi}^{\bar{a}}\vp_{\infty},
\end{align}
where $A^{\text{gau}}_{\bar{a}}$ fulfills $\text{d}A^{\text{gau}}=-(ea^2)^{-1}a^{-2}\text{Vol}_{\vp_{\infty}}$ while $A^{\text{gau}}$ does not include any term of the form $\text{d}\lambda$ by the choice of the canonical lift.

With these parity conditions it is easy to see that $\Theta_{\vp_{\infty}}<\infty$ $\forall\vp_{\infty}\in\mathcal{F}_{\vp_{\infty}}^{(0)}$ $\Rightarrow$ $\Omega_{\vp_{\infty}}<\infty$ $\forall\vp_{\infty}\in\mathcal{F}_{\vp_{\infty}}^{(0)}$.\\

\underline{$m\neq0$}: At first we make the following observation:
\begin{prop}\label{oddvol}
	For every $m\in\mathbb{Z}\backslash\{0\}:\exists \vp_{\infty}\in\mathcal{F}_{\{\vp_{\infty}\}}^{(m)}$ such that $\text{Vol}_{\vp_{\infty}}$ is odd.
\end{prop}
\begin{proof}
	Cases:
	
	$m=1$: Just take $\vp_{\infty}:=\text{id}:x\mapsto x$. Clearly $[\vp_{\infty}]=1$ (See definition \ref{locbdeg} and proposition \ref{locbdegp} in the appendix \ref{appendix1}). Then $\vp_{\infty}(-x)=-x=-\vp_{\infty}(x)$, i.e. $\vp_{\infty}$ is odd and therefore $(\text{Vol}_{\vp_{\infty}})_{\bar{a}\bar{b}}=\vp_{\infty}\cdot(\partial_{\bar{a}}\vp_{\infty}\times\partial_{\bar{b}}\vp_{\infty})$ is odd. By definition (\ref{angularparity}) this is equivalent to $\text{Vol}_{\vp_{\infty}}$ being odd as a two form.\\
	
	$m=-1$: Choose a northpole and a southpole $p_N,p_S\in S^2$ (antipodal points) with the corresponding coordinates $(h,\varphi)$, where $h:S^2\longrightarrow[-1,1]$ is the projection onto the axis connecting $p_N$ and $p_S$ when $S^2$ is embedded into $\mathbb{R}^3$, while $\varphi: S^2\setminus\text{Geo}_{p_N,p_S}\longrightarrow (0,2\pi)$ is the angular coordinate with respect to this axis, where $\text{Geo}_{p_N,p_S}$ is some closed geodesic segment connecting $p_N$ and $p_S$. Define $\vp_{\infty}$ via the map $f:(-1,1)\times(0,2\pi)\longrightarrow (-1,1)\times (0,2\pi)$,
	\begin{align}
		f(h,\varphi)=(h,-\varphi+2\pi).
	\end{align}
	The map $\vp_{\infty}|_{S^2\setminus\text{Geo}_{p_N,p_S}}:=(h,\varphi)^{-1}\circ f\circ (h,\varphi)$ extended to $S^2$ by mapping $\text{Geo}_{p_N,p_S}$ with the identity map onto $\text{Geo}_{p_N,p_S}$ on the target sphere does hit every point of $S^2$ exactly once. Look at the derivative:
	\begin{align}
		Df_{(h,\varphi)}=\begin{pmatrix}
			1 & 0\\
			0 & -1
		\end{pmatrix}\;\Rightarrow\det(Df_{(h,\varphi)})=-1\;\forall (h,\varphi).\label{136}
	\end{align}
	Since $(h,\varphi)$ is a orientation preserving diffeomorphism, the winding number has to be $[\vp_{\infty}]=-1$.
	
	Now show that $\text{Vol}_{\vp_{\infty}}$ is odd under parity. By definition (\ref{degree2}), $\text{Vol}_{\vp_{\infty}}=\vp_{\infty}^*\omega$ where $\omega$ is the standart volume form on $S^2$. Let's choose coordinates $\{y^{\bar{a}}\}$ on a region $R\setminus L\subset S^2$, where $a(R)\subset R$ and $\omega=\text{d}y^1\wedge\text{d}y^2$ \footnote{An example for such a region $R$ would be some ring segment around the equator bounded by two latitudes which are their mutual image under the antipodal map. $L$ is some line segment of a longitude. It is easy to see that such coordinates exist, by segmenting $R$ into regular equal-volume elements that are cut by longitudes and latitudes.}. Now choose those coordinates on the target sphere ($\{y_T^{\bar{a}}\}$ on $R_T\setminus L_T$) as well as on the domain sphere ($\{y_D^{\bar{a}}\}$ on $R_D\setminus L_D$) of $\vp_{\infty}$. Let $x\in R_D\setminus(L_D\cup a(L_D))$ and $\vp_{\infty}(x)\in R_T\setminus(L_T\cup a(L_T))$, then
	\begin{align}\label{theprepullback}
		a^*\text{Vol}_{\vp_{\infty}}|_x&=((\vp_{\infty}\circ a)^*\omega)|_x\nonumber\\&=\text{det}[D\vp_{\infty}|_{a(x)}]\text{det}[Da|_x]\,(\text{d}y^1_D\wedge\text{d}y^2_D)|_x,
	\end{align}
where it is obvious that $\text{det}[D\vp_{\infty}|_{a(x)}]$ and $\text{det}[Da|_x]$ are coordinate-independent expressions.
It is easy to see that $\text{det}[Da|_x]=-1$ $\forall x\in S^2$, as well as $\text{det}[D\vp_{\infty}|_{a(x)}]=\text{det}[Df|_{(h,\varphi)(a(x))}]=-1=\text{det}[D\vp_{\infty}|_{x}]$ $\forall x\in S^2$.
\begin{align}
	\Rightarrow\;a^*\text{Vol}_{\vp_{\infty}}|_x=-\text{det}[D\vp_{\infty}|_x]\,(\text{d}y^1_D\wedge\text{d}y^2_D)|_x=-\text{det}[D\vp_{\infty}|_x]\,\omega|_{x}=-\vp_{\infty}^*\omega\,|_x=-\text{Vol}_{\vp_{\infty}}|_x.\label{thepullback}
\end{align}
Hence $\text{Vol}_{\vp_{\infty}}$ is odd under parity.\\

$m=2$: Let's stay in the coordinates $(h,\varphi)$. Let
\begin{align}\label{Dfzwei}
	f(h,\varphi):&=\begin{cases}
		(2h^2-1,\varphi) & |h\in (0,1)\\
		(2h^2-1,-\varphi+2\pi) & |h\in (-1,0)
	\end{cases},\\
	\Rightarrow Df_{(h,\varphi)}&=\begin{cases}
		\begin{pmatrix}4h & 0\\ 0 & 1\end{pmatrix} & |h\in(0,1)\\
		\begin{pmatrix}4h & 0\\ 0 & -1\end{pmatrix} & |h\in(-1,0)
	\end{cases}.
\end{align}
$f$ is extended to $h=0$ via $\vp_{\infty}((h,\varphi)^{-1}(\{(0,\varphi)\}_{\varphi\in(0,2\pi)}))=\{p_S\}.$ It is easy to see that $\vp_{\infty}$ is smooth\footnote{The only problematic points could be the points of $h=0$. By the definition of $\vp_{\infty}$, $\partial_{\varphi}\vp_{\infty}|_{h=0}=0$, and $\partial_h\vp_{\infty}|_{h=0}=4h$. You have to change the coordinates on the target sphere, to see that the derivative is continuous in the neighbourhood of the $h=0$ set. With that it follows clearly, that $\vp_{\infty}$ is smooth.}. With $S^2_N$ being the closed northern hemisphere with respect to the coordinates $(h,\varphi)$ and $S^2_S$ being the closed southern hemisphere, $\vp_{\infty}|_{S^2_N}$ and $\vp_{\infty}|_{S^2_S}$ are injective and have $S^2$ in their image respectively. As by \eqref{Dfzwei} $\text{det}[Df|_{h}]=\text{det}[Df|_{-h}]>0$ $\forall h$ $\Rightarrow$ $[\vp_{\infty}]=2$.

Now again, $\text{det}[D\vp_{\infty}|_{a(x)}]=\text{det}[Df|_{(h,\varphi)(a(x))}]=\text{det}[Df|_{-h}]=\text{det}[Df|_{h}]=\text{det}[D\vp_{\infty}|_x]$. As \eqref{theprepullback} holds in general for every $\vp_{\infty}\in\mathcal{F}_{\{\vp_{\infty}\}}$, \eqref{thepullback} is again applicable which gives the desired result in this case.\\

	$m=-2$: Take the $f$ from $m=2$ and reverse the orientation of the image like in the case $m=-1$.\\
	
$m\in\mathbb{N}\setminus\{0\}$: Define $f$, using a polynomial $g:[-1,1]\longrightarrow [-1,1]$ of m'th order that has $m-1$ stationary points $h_i\in(-1,1)$ ($h_i<h_{i+1}$ $\forall i$) that are maxima or minima with $g(h_i)\in\{\pm 1\}$. Additionaly $g$ is an even function if $m$ is even and $g$ is an odd function if $m$ is odd. Let
\begin{align}
	f(h,\varphi):=\begin{cases}
		(g(h),\varphi) & |h\in(h_1,1)\\
		(g(h),-\varphi+2\pi) & | h\in(h_2,h_1)\\
		(g(h),\varphi) & | h\in(h_3,h_2)\\
		... & | ...\;\;\;\; .
	\end{cases}
\end{align}
With a such defined smooth $f$ (same argument as before) we find again $\text{det}[Df|_{-h}]=\text{det}[Df|_h]$ $\forall h$. Now execute the same argumentation as in the case $m=2$ to find $[\vp_{\infty}]=m$ and $\text{Vol}_{\vp_{\infty}}$ being odd.\\

$m\in-\mathbb{N}\setminus\{0\}$: Again, reverse the orientation of $f$.
\end{proof}
\begin{theorem}
	If $m\neq 0$, assuming the boundary condition \eqref{fatboundary}, a sufficient condition for the symplectic potential to be finite is:
	\begin{align}\label{theoremconditions}
		\bar{A}_{\bar{a}}=\bar{A}_{\bar{a}}^{\text{even}}+\partial_{\bar{a}}\Phi_{\text{even}}+A^{\text{gau}}_{\bar{a}},\;\;\bar{\pi}^{\bar{a}}=\bar{\pi}^{\bar{a}}_{\text{odd}},
	\end{align}
	where $\text{d}A^{\text{gau}}=\text{Vol}_{\vp_{\infty}}-\text{Vol}_{\vp_{\infty}'}$ for some $\vp_{\infty}\in\mathcal{F}^{(m)}_{\{\vp_{\infty}\}}$ and $\vp_{\infty}'\in\mathcal{F}^{(m)}_{\{\vp_{\infty}\}}$ such that $\text{Vol}_{\vp_{\infty}'}$ is odd.
\end{theorem}
\begin{proof}
	Let $m\in\mathbb{Z}\setminus\{0\}$ be fixed and $\vp_{\infty}\in\mathcal{F}_{\{\vp_{\infty}\}}^{(m)}$. Take the corresponding $\text{Vol}_{\vp_{\infty}}$. As $\text{Vol}_{\vp_{\infty}}$ is a 2-form it can always be written as the sum of a exact form and a residual part:
	\begin{align}
		\text{Vol}_{\vp_{\infty}}=\text{Vol}_{\vp_{\infty}}^{\text{ex}}+\text{Vol}_{\vp_{\infty}}^{\text{res}}=\text{d}A^{\text{gau}}+\text{Vol}_{\vp_{\infty}}^{\text{res}}
	\end{align}
for some 1-form $A^{\text{gau}}$.
It was shown in \ref{oddvol} that for every $m\in\mathbb{Z}\setminus\{0\}:$ $\exists$ $\vp_{\infty}'\in\mathcal{F}^{(m)}_{\{\vp_{\infty}\}}$ such that $\text{Vol}_{\vp_{\infty}'}$ is odd.
Let's assume $H^2_{\text{dR}}(S^2)\ni[\text{Vol}_{\vp_{\infty}}-\text{Vol}_{\vp_{\infty}'}]\neq 0$. By the theorem of de Rham \cite[Theorem 18.14]{Lee}, $\text{H}^2_{\text{dR}}(S^2)\cong \text{H}^2(S^2;\mathbb{R}):=\text{Hom}(\text{H}^2(S^2),\mathbb{R})$ and 
\begin{align}
	\bigg([S^2]\mapsto\int_{S^2}\alpha\bigg)\mapsto[\alpha]
\end{align}
is an isomorphism. $\text{H}^2(S^2)\cong\mathbb{Z}$ is the second homology group of $S^2$.\footnote{By \cite[Theorem 6 in chapter 4]{Spanier}, $\text{H}^2(S^2)\cong\mathbb{Z}$. Using the Hurewicz theorem \cite[Theorem 2 in chapter 7]{Spanier}), $\pi_2(S^2)\cong \text{H}^2(S^2)$, because $S^2$ is simply connected.} Since $\oint(\text{Vol}_{\vp_{\infty}}-\text{Vol}_{\vp_{\infty}'})=m-m=0$, by the isomorphism, $0$ has to map to $0\neq[\text{Vol}_{\vp_{\infty}}-\text{Vol}_{\vp_{\infty}'}]$ which is a contradiction. Hence, $[\text{Vol}_{\vp_{\infty}}-\text{Vol}_{\vp_{\infty}'}]=0$. $\Rightarrow$ $A^{\text{gau}}$ with $\text{d}A^{\text{gau}}=\text{Vol}_{\vp_{\infty}}-\text{Vol}_{\vp_{\infty}'}$ exists. $\Rightarrow$ $\text{Vol}^{\text{res}}_{\vp_{\infty}}=\text{Vol}_{\vp_{\infty}'}$.
Lets consider a variation paramerized by $s$:
\begin{align}\label{varagau}
	\partial_s(\text{Vol}_{\vp_{\infty}})_{\bar{a}\bar{b}}=&\vp_{\infty}\cdot(\partial_{\bar{a}}(\partial_s\vp_{\infty})\times\partial_{\bar{b}}\vp_{\infty})+\vp_{\infty}\cdot(\partial_{\bar{a}}\vp_{\infty}\times\partial_{\bar{b}}(\partial_s\vp_{\infty}))\nonumber\\
	=&\partial_{\bar{a}}(\vp_{\infty}\cdot(\partial_s\vp_{\infty}\times\partial_{\bar{b}}\vp_{\infty}))-\partial_{\bar{b}}(\vp_{\infty}\cdot(\partial_s\vp_{\infty}\times\partial_{\bar{a}}\vp_{\infty}))\nonumber\\
	\text{Hence}\;\;&\text{d}[\vp_{\infty}\cdot(\partial_s\vp_{\infty}\times\text{d}\vp_{\infty})]=\text{d}(\partial_s A^{\text{gau}})+\text{d}[\vp_{\infty}'\cdot(\partial_s\vp_{\infty}'\times\text{d}\vp_{\infty}')],\nonumber\\
	\Rightarrow& \vp_{\infty}\cdot(\partial_s\vp_{\infty}\times\text{d}\vp_{\infty})=\partial_s A^{\text{gau}}+\vp_{\infty}'\cdot(\partial_s\vp_{\infty}'\times\text{d}\vp_{\infty}')+\text{d}\lambda.
\end{align}
In the next step we show that the conditions stated in the theorem are sufficient for the symplectic potential to be finite. Let $\bar{\pi}^{\bar{a}}=\bar{\pi}^{\bar{a}}_{\text{odd}}$ and $\bar{A}_{\bar{a}}=\bar{A}^{\text{even}}_{\bar{a}}+\partial_{\bar{a}}\Phi_{\text{even}}+A^{\text{gau}}_{\bar{a}}$ with $\text{d}A^{\text{gau}}=\text{Vol}_{\vp_{\infty}}-\text{Vol}_{\vp_{\infty}'}$. In the equation \eqref{varagau} we choose $\lambda=0$, because we lift the gauge transformation $\vp_{\infty}'\mapsto\vp_{\infty}$ in the canonical way. Now using \eqref{varagau} in the expression of the symplectic potential \eqref{SympotA} gives
\begin{align}
	\eqref{SympotA}=a^2\bar{\pi}^{\bar{a}}\partial_s(\bar{A}_{\bar{a}}^{\text{even}}+\partial_{\bar{a}}\Phi_{\text{even}})-\frac{1}{ea^2}\bar{\pi}^{\bar{a}}\vp_{\infty}'\cdot(\partial_s\vp_{\infty}'\times\partial_{\bar{a}}\vp_{\infty}'),
\end{align}
which is an odd expression with the exception of $a^2\bar{\pi}^{\bar{a}}\partial_{\bar{a}}\Phi_{\text{even}}$, which vanishes by using partial integration and the boundary condition $\partial_{\bar{a}}\bar{\pi}^{\bar{a}}=0$.

\end{proof}
By that proposition and the result in the case $m=0$, for every $\vp_{\infty}\in\mathcal{F}_{\{\vp_{\infty}\}}$ the symplectic form $\Omega_{\vp_{\infty}}$ exists under the assumption of the above given parity conditions.\\

Let's state the boundary- and fall-off conditions at this point. In order for the symplectic form to be finite we add to the conditions \ref{fatboundary} the additional conditions:\\

\begin{tabular}[h]{l|l|l|l|l|l}
	m & $\text{Vol}_{\vp_{\infty}}$ & $\bar{A}_r$ & $\bar{\pi}^r$ & $\bar{A}_{\bar{a}}$ & $\bar{\pi}^{\bar{a}}$\\
	\hline
	= 0 & $\oint\text{Vol}_{\vp_{\infty}}=0$ & odd & even & $=\bar{A}^{\text{even}}_{\bar{a}}+\partial_{\bar{a}}\Phi_{\text{even}}+A^{\text{gau}}_{\bar{a}}$, s.t. $\text{d}A^{\text{gau}}=\text{Vol}_{\vp_{\infty}}$ & odd\\
	$\neq$ 0 & $\oint\text{Vol}_{\vp_{\infty}}=m$ & odd & even &$=\bar{A}^{\text{even}}_{\bar{a}}+\partial_{\bar{a}}\Phi_{\text{even}}+A^{\text{gau}}_{\bar{a}}$,   & odd \\ & & & & s.t. $\text{d}A^{\text{gau}}=\text{Vol}_{\vp_{\infty}}-\text{Vol}_{\vp_{\infty}'}$ and $\text{Vol}_{\vp_{\infty}'}$ is odd &	 \;\;\; .
\end{tabular}\\
These additional parity conditions are invariant under the action of the Poincar\'{e} algebra \ref{angularparity}.
\begin{defi}\label{phasespace}
	The set of smooth functions $\{(A,\pi,\phi,\Pi)\}$ that fulfills the conditions \ref{fatboundary} and the additional parity conditions above, denoted by $(\mathbb{P},\Omega)$, is called the \textit{unbroken phase space}.
\end{defi}
$(\mathbb{P},\Omega)$ is a symplectic manifold.

\section{The boost problem}
\subsection{The Lorentz boost of the symplectic form}\label{bo}
In order to have a well stated phase space of a relativistic field theory, the Poincar\'{e} group has to act as a Hamiltonian action. By the result in \cite{Woodhouse:GQ}, chapter 3.2 and 3.3, for the Poincar\'{e} group to act Hamiltonian it is sufficient to show that the group acts via symplectomorphisms. At the Lia algebra level this is $\liephase_{X^{(\xi^{\perp},\vect{\xi})}}\Omega=0$, where $X^{(\xi^{\perp},\vect{\xi})}$ is the fundamental vector field of the group action, corresponding to the Lie-algebra elements $(\xi^{\perp},\vect{\xi})$\footnote{A vanishing Lie derivative of the symplectic form on a locally convex space does in general not ensure the existence of a symplectomorphism, but this functional analytic proplem is beyond the scope here. (\cite{Neeb}, remark II.2,10)}.
It is easy to see that $\liephase_{X^{(0,\vect{\xi})}}\Omega=0$, but a straightforeward calculation shows:
\begin{align}
	\liephase_{X^(\xi^{\perp},0)}\Omega=&\underbrace{a^2\oint\text{d}^2x\sqrt{\bar{\gamma}}\,\extder\bar{A}_r\boldsymbol{\wedge}\bar{\nabla}^{\bar{a}}(b\extder\bar{A}_{\bar{a}})}_{=(1)}\underbrace{-\oint\text{d}^2x\sqrt{\bar{\gamma}}\,\extder\bar{A}_r\boldsymbol{\wedge}\frac{1}{ea^2}\bar{\nabla}^{\bar{a}}(b\vp_{\infty}\cdot(\extder\vp_{\infty}\times\partial_{\bar{a}}\vp_{\infty}))}_{=(2)}\nonumber\\
	&\underbrace{-\oint\text{d}^2x\sqrt{\bar{\gamma}}\bar{A}_r\bar{\nabla}^{\bar{a}}(b\extder\vp_{\infty}\boldsymbol{\wedge}\cdot(\frac{1}{ea^2}\vp_{\infty}\times\partial_{\bar{a}}\extder\vp_{\infty}))}_{=(3)}\nonumber\\
	&\underbrace{+\oint\text{d}^2x\,\extder\vec{\zeta}^{(0)}\boldsymbol{\wedge}\cdot\extder\vec{\pi}^r_{(0)}-\int\text{d}^3x\,\extder\vec{\zeta}\boldsymbol{\wedge}\cdot\extder\vec{\mathscr{G}}}_{(4)}.\label{expr}
\end{align}
By using the parity conditions \eqref{phasespace} we can write $\bar{A}_{\bar{a}}=\bar{A}_{\bar{a}}^{\text{even}}+\partial_{\bar{a}}\Phi_{\text{even}}+A^{\text{gau}}_{\bar{a}}$. Using the parity condition $\bar{A}_r=\bar{A}_r^{\text{odd}}$ we have $a^2\oint\text{d}^2x\,\sqrt{\bar{\gamma}}\extder\bar{A}_r\boldsymbol{\wedge}\bar{\nabla}^{\bar{a}}(b\extder\bar{A}_{\bar{a}}^{\text{even}})=0$. The integral
\begin{align}
	B_1:=a^2\oint\text{d}^2x\,\sqrt{\bar{\gamma}}\,\extder\bar{A}_r\boldsymbol{\wedge}\bar{\nabla}^{\bar{a}}(b\extder\partial_{\bar{a}}\Phi_{\text{even}})
\end{align}
is exactly the boundary term that appears in $\text{U}(1)$ electromagnetism \cite{Henneaux-ED}. This term is treated in the section \ref{asym}.

Let's take a look at the term of (1) that corresponds to $A_{\bar{a}}^{\text{gau}}$. By the equation \eqref{varagau} (choosing $\lambda=0$),
\begin{align}
	\extder A_{\bar{a}}^{\text{gau}}=\frac{1}{ea^4}(\vp_{\infty}\cdot(\extder\vp_{\infty}\times\partial_{\bar{a}}\vp_{\infty})-\vp_{\infty}'\cdot(\extder\vp_{\infty}'\times\partial_{\bar{a}}\vp_{\infty}')),
\end{align}
where the term involving $\vp_{\infty}'$ does vanish if and only if $m=0$. In the case where it does not vanish, the $\vp_{\infty}'$ is such that $(\text{Vol}_{\vp_{\infty}'})_{\bar{a}\bar{b}}$ is odd. Using \eqref{varagau} again we can conclude that $\vp_{\infty}'\cdot(\extder\vp_{\infty}'\times\partial_{\bar{a}}\vp_{\infty}')$ is even. Hence sticking this expression in the integral (1) gives $0$. From $a^2\oint\text{d}^2x\,\sqrt{\bar{\gamma}}\,\extder\bar{A}_r\boldsymbol{\wedge}\bar{\nabla}^{\bar{a}}(b\extder A_{\bar{a}}^{\text{gau}})$ it remains the expression $a^2\oint\text{d}^2x\,\sqrt{\bar{\gamma}}\,\extder\bar{A}_r\boldsymbol{\wedge}\bar{\nabla}^{\bar{a}}(b\vp_{\infty}\cdot(\extder\vp_{\infty}\times\partial_{\bar{a}}\vp_{\infty}))$, which does cancel with the integral (2).

As $\vp_{\infty}$ and also $\text{Vol}_{\vp_{\infty}}$ does in general not fulfill any parity condition, the integral $B_2:=(3)$ does not vanish.

The integrals $B_3:=(4)$ do not vanish for any field configuration, if and only if $\vec{\zeta}$ is field dependent.
A general gauge transformation that fulfills the fall-off and boundary conditions \ref{fatboundary} to accompany the boost depends at least on $\vp_{\infty}$.

\subsection{Solving the boost problem I}\label{booster}
In the previous section we found that the Lorentz boosts do not act canonical on phase space. Rather  for a general $\vp_{\infty}\in\mathcal{F}_{\{\vp_{\infty}\}}$ and an $\bar{A}_{\bar{a}}$ that has an odd component, $\liephase_{X^{(\xi^{\perp},0)}}\Omega=B_1+B_2+B_3\neq 0$. $B_1\neq0$ happens when one allows for improper gauge transformations $\partial_{\bar{a}}\Phi_{\text{even}}$ in the theory (see \eqref{phasespace}). This boundary term can be handled by introducing a new degree of freedom $\bar{\psi}$ on the sphere at infinity and a boundary term for the symplectic form. It will turn out that this new degree of freedom is actually the missing improper gauge transformation that was locked out by the condition $\bar{\pi}^r_{\text{odd}}=0$ (\ref{impgenerator}) before. This method to handle $B_1$ will be carried out in the section \ref{asym}.

The boundary term $B_2$ is made uncritical by a conceptual approach. We use the fact that inside each connection component of $\mathbb{P}$, $\vp_{\infty}$ is only a proper gauge degree of freedom (corollary \ref{thegaugemanifold}). If we would work on the \textit{reduced phase space} which is the set of actual physical states, i.e. the quotient $\mathcal{C}/\mathsf{Gau}$, there would be no boundary term $B_2$, because this term involves $\extder\vp_{\infty}$, which is not anymore present in the theory. But working on the reduced phase space is not realistic because it is unclear whether the quotient $\mathcal{C}/\mathsf{Gau}$ has any manifold structure.

Rather we take the partial gauge that fixes some $\vp_{\infty}$ in each connected component $\mathcal{F}_{\{\vp_{\infty}\}}^{(m)}$ as a boundary condition.
\begin{defi}
	The \textit{pre-final} phase space of $\SU(2)$ Yang-Mills-Higgs theory is
	\begin{align}
		\bar{\mathbb{P}}^{\text{pre}}:=\bigsqcup_{m\in\mathbb{Z}}\mathbb{P}_{\vp_{\infty}^{(m)}},
	\end{align}
	where $\vp_{\infty}^{(m)}$ is some element in $\mathcal{F}_{\{\infty\}}^{(m)}$.
\end{defi}
\begin{prop}\label{adjustpoin}
	There is an action of $\mathfrak{poin}$ on $\mathbb{P}$ that acts tangential to all fibres $\mathbb{P}_{\vp_{\infty}}$ which is the composition of the usual $\mathfrak{poin}$ action (\ref{Poinis}) with an infinitesimal proper gauge transformation.
\end{prop}
\begin{proof}
	Certainly the action of $\vect{\xi}$ does not leave $\vp_{\infty}$ invariant is $\vect{\xi}$ as $\delta_{(0,\vect{\xi})}\vp_{\infty}=\xi^i\partial_i\vp_{\infty}$. Finding a correcting proper gauge transformation that brings the translation part of $\delta_{(0,\vect{\xi})}\vp_{\infty}$ to zero was already done before (\ref{trickyspatial}). Let's now choose a proper gauge transformation $\vec{\epsilon}^{\perp}$ such that $\delta_{Y}\vp_{\infty}+\delta_{\vec{\epsilon}^{\perp}}\vp_{\infty}=0$. This equation is fulfilled by the choice $\vec{\epsilon}^{\perp}=(ea^2)^{-1}\vp_{\infty}\times Y^{\bar{a}}\partial_{\bar{a}}\vp_{\infty}$. It is now easy to see that $\delta_{(0,\vect{\xi})}+\delta_{\vec{\epsilon}^{\perp}}$ acts tangential to $\mathbb{P}_{\vp_{\infty}}$ $\forall\vp_{\infty}\in\mathcal{F}_{\{\vp_{\infty}\}}$.
	
	If a gauge transformation $\vec{\zeta}$ is accompanying the action of the boost $\xi^{\perp}$, $\delta_{(\xi^{\perp},0)}\vp_{\infty}=e\vec{\zeta}_{(0)}^{\perp}\times\vp_{\infty}$, we have to choose here $\vec{\zeta}_{(0)}^{\perp}=0$.
\end{proof}
Hence $\mathfrak{poin}$ acts on $\bar{P}^{\text{pre}}$.
The generators of gauge transformations on $\bar{\mathbb{P}}^{\text{pre}}$ are
\begin{align}
	\vec{\epsilon}=\epsilon\vp_{\infty}+o(r^{-\mathbb{N}}),
\end{align}
depending on the $\vp_{\infty}\in\{\vp_{\infty}^{(m)}\}_{m\in\mathbb{Z}}$.
Choosing the space $\bar{\mathbb{P}}^{\text{pre}}$ is associated to the treatment of $\vp_{\infty}$ as a non-dynamical \textit{background field}.
The symplectic form on $\mathbb{P}_{\vp_{\infty}}$ for some $\vp_{\infty}$ is with $\extder\vp_{\infty}=0$:
\begin{align}
	\Omega_{\vp_{\infty}}=a^2\int\text{d}^3x\,\extder A_i\boldsymbol{\wedge}\extder\pi^i+o(r^{-\mathbb{N}}).\\
	\Rightarrow\;\liephase_{(\xi^{\perp},0)}\Omega_{\vp_{\infty}}=(B_1(+B_3))|_{\extder\vp_{\infty}=0},
\end{align}
where $B_3$ does not vanishes if $\vec{\zeta}$ is field dependenent in some $\mathbb{P}_{\vp_{\infty}}$. In the next subsection we will see that we actually have to choose  $B_3$ field dependent for the method in the next subsection to work.

\subsection{Solving the boost problem II}
\label{asym}
This section is an exact reprise of the method by Henneaux and Troessaert for electromagnetism in \cite{Henneaux-ED}, because they have to treat exactly the same boundary term.

On the pre-final phase space $\bar{\mathbb{P}}^{\text{pre}}$ the Lorentz boost of the symplectic form in the background $\vp_{\infty}$ is:
\begin{align}
	\liephase_{X^{(\xi^{\perp},0)}}\Omega_{\vp_{\infty}}&=a^2\oint\text{d}^2x\sqrt{\bar{\gamma}}\,\extder\bar{A}_r\boldsymbol{\wedge}\bar{\nabla}^{\bar{a}}(b\extder(\bar{A}_{\bar{a}}^{\text{even}}+\partial_{\bar{a}}\Phi_{\text{even}}))\nonumber\\
	&-\int\text{d}^3x\,\extder\vec{\zeta}\boldsymbol{\wedge}\cdot\extder\vec{\mathscr{G}}+a^2\oint\text{d}^2x\,\extder\bar{\zeta}\boldsymbol{\wedge}\extder\bar{\pi}^r.\label{boostproblem}
\end{align}
For the boosts to be canonical, this expression has to vanish. By \ref{phasespace} $\bar{A}_r$ is odd under parity and hence the term including $\bar{A}_{\bar{a}}^{\text{even}}$ vanishes, while the term including the improper gauge transformation $\partial_{\bar{a}}\Phi_{\text{even}}$ does not vanish. Let's treat this term at first and we will see that the remaining terms will also be treated by this ansatz.
Henneaux and Troessaert introduced a new field $\bar{\psi}\in C^{\infty}(S^2,\mathbb{R})$ together with an alternative symplectic form $\Omega_{\vp_{\infty}}':=\Omega_{\vp_{\infty}}+\omega$, where
\begin{align}
	\omega=-a^2\oint\text{d}^2x\sqrt{\bar{\gamma}}\,\extder \bar{A}_r\boldsymbol{\wedge}\extder\bar{\psi}.
\end{align}
$\bar{\psi}$ should have the following transformation behaviour under boosts:
\begin{align}
	\delta_b\bar{\psi}=\bar{\nabla}^{\bar{a}}(b\bar{A}_{\bar{a}})+b\bar{A}_r.\label{cureboost}
\end{align}
 With that,
\begin{align}
	\liephase_{X^{(\xi^{\perp},0)}}\omega&=\extder(i_{X^{(\xi^{\perp},0)}}\omega)\nonumber\\&=-a^2\oint\text{d}^2x\,b\extder\bar{\pi}_r\boldsymbol{\wedge}\extder\bar{\psi}-a^2\oint\text{d}^2x\,\sqrt{\bar{\gamma}}\extder\bar{A}_r\boldsymbol{\wedge}\bar{\nabla}^{\bar{a}}(b\extder\bar{A}_{\bar{a}}).
\end{align}
\begin{align}
	\Rightarrow\;\liephase_{X^{(\xi^{\perp},0)}}\Omega'_{\vp_{\infty}}&=-a^2\oint\text{d}^2x\,b\extder\bar{\pi}_r\boldsymbol{\wedge}\extder\bar{\psi}+a^2\oint\text{d}^2x\,\extder\bar{\zeta}\boldsymbol{\wedge}\extder\bar{\pi}^r\nonumber\\
	&-\int\text{d}^3x\,\extder\vec{\zeta}\boldsymbol{\wedge}\cdot\extder\vec{\mathscr{G}}.
\end{align}
For the first two terms to cancel we choose the field dependent $\bar{\zeta}=-b\bar{\psi}.$ The third term is handled after introducing the canonical pair $(\vec{\psi},\vec{\pi}_{\psi})$ with an affiliated additional term for the symplectic form,
\begin{align}
	\Omega^{\psi}=\int\text{d}^3x\,\extder\vec{\pi}_{\psi}\boldsymbol{\wedge}\cdot\extder\vec{\psi},
\end{align}
while demanding the constraint $\vec{\pi}_{\psi}\approx 0$, with the fall-off behaviour
\begin{align}
	\vec{\psi}=(r^{-1}\bar{\psi}+r^{-2}\psi^{(1)}+...)\vp_{\infty}+\text{o}(r^{-\mathbb{N}}),\;\vec{\pi}_{\psi}=(r^{-4}\bar{\pi}_{\psi}+...)\vp_{\infty}+\text{o}(r^{-\mathbb{N}}),
\end{align}
where the fall-off of $\vec{\pi}_{\psi}$ is the fastes decay that is invariant under the yet to be defined transformations under boosts. The additional constraint function is
\begin{align}
	G_{\vec{\mu}}=\int\text{d}^3x\,\vec{\mu}\cdot\vec{\pi}_{\psi}.
\end{align}
By the fall-off of $\vec{\pi}_{\psi}$ we allow for a fall-off $\vec{\mu}=(\bar{\mu}+r^{-1}\mu^{(1)}+...)\vp_{\infty}+o(r^{-\mathbb{N}})$ of the gauge parameter. $G_{\vec{\mu}}$ is not functionally differentiable with respect to the extended symplectic form $\bar{\Omega}_{\vp_{\infty}}:=\Omega'_{\vp_{\infty}}+\Omega^{\psi}$. We have to add a boundary term, such that
\begin{align}
	\bar{G}_{\vec{\mu}}=G_{\vec{\mu}}-a^2\oint\text{d}^2x\sqrt{\bar{\gamma}}\,\bar{\mu}\bar{A}_r\label{newimpropers}
\end{align}
is f.d.w.r.t. $\bar{\Omega}_{\vp_{\infty}}$. The corresponding gauge transformation is $\delta_{\vec{\mu}}\vec{\psi}=\vec{\mu}$ where the improper part is given by $\bar{\mu}_{\text{odd}}$. Hence up to $\bar{\psi}_{\text{odd}}$, $\vec{\psi}$ is a proper gauge degree of freedom.

The last adjustment in this section is to choose the boost- accompanying gauge parameter $\vec{\zeta}$, the boost transformations of $\vec{\psi}$ and $\vpi_{\psi}$ and an adjustment of the transformation of $\vpi^i$ that vanishes on the constraint surface:
\begin{align}\label{psiboost}
	&\delta_{(\xi^{\perp},0)}\vec{\psi}=D_i(\xi^{\perp}\va^i),\;\;\delta_{(\xi^{\perp},0)}\vpi_{\psi}=-\xi^{\perp}\vec{\mathscr{G}},\nonumber\\
	&\vec{\zeta}=-\xi^{\perp}\vec{\psi},\;\;\delta_{(\xi^{\perp},0)}\vpi^i=(...)-\xi^{\perp}D_i\vpi_{\psi}.
\end{align}
This choice results in $\liephase_{X^{(\xi^{\perp},0)}}\bar{\Omega}_{\vp_{\infty}}=0$.
\begin{defi}
	The \textit{final} phase space of $\SU(2)$ Yang-Mills-Higgs theory is
	\begin{align}
		\bigg(\bar{\mathbb{P}}:=\bigsqcup_{m\in\mathbb{Z}}\mathbb{P}_{\vp_{\infty}^{(m)}}\times\{(\vec{\psi},\vec{\pi}_{\psi})\}_{\vp_{\infty}^{(m)}},(\bar{\Omega}_{\vp_{\infty}^{(m)}})_{m\in\mathbb{Z}}\bigg).
	\end{align}
\end{defi}

\section{Asymptotic symmetries}
In the last section in order for the boosts to act canonically we introduced an additional gauge degree of freedom parametrized by $\vec{\mu}$ that is generated by $\bar{G}_{\vec{\mu}}$. Together with $\bar{G}_{\vec{\epsilon}}$ this combines to
\begin{align}
	\bar{G}_{(\vec{\epsilon},\vec{\mu})}=\int\text{d}^3x\,(\vec{\epsilon}\cdot\vec{\mathscr{G}}+\vec{\mu}\cdot\vpi_{\psi})-a^2\oint\text{d}^2x\,(\bar{\epsilon}\bar{\pi}^r+\sqrt{\bar{\gamma}}\bar{\mu}\bar{A}_r).
\end{align}
The pair $(\vec{\epsilon},\vec{\mu})$ with the asymptotic component $(\bar{\epsilon}_{\text{even}},\bar{\mu}_{\text{odd}})$ generate improper gauge transformations. $\bar{\epsilon}_{\text{even}}$ and $\bar{\mu}_{\text{odd}}$ combine to a single smooth function on $S^2$. It is easy to see that $\bar{G}_{(\vec{\epsilon},\vec{\mu})}$ actually generate symmetries, i.e. $\{\bar{G}_{\vec{\epsilon},\vec{\mu}},H\}\approx 0$ $\forall(\vec{\epsilon},\vec{\mu})$.
\begin{bem}
	At first sight the introduction of the new physical variable $\bar{\psi}_{\text{odd}}$ seems adhoc, but as $\delta_{\bar{\mu}}\bar{\psi}_{\text{odd}}=\bar{\mu}_{\text{odd}}$ this new variable brings only the odd part of $\bar{\epsilon}$ back as a physical degree of freedom. Also $\bar{A}_r$ is a gauge-invariant conserved quantity. By the introduction of $\bar{\mu}$ we get a corresponding symmetry.
\end{bem}
The generators of the Poincar\'{e} transformations have the general form:
\begin{align}
	P_{(\xi^{\perp},\vect{\xi})}=\int\text{d}^3x\sqrt{g}\,(\xi^{\perp}\mathscr{P}_0+\xi^i\mathscr{P}_i)+\oint\text{d}^2x\,\mathscr{B}.\label{Boostgen}
\end{align}
$\mathscr{P}_{0}$ is just the Hamiltonian density with adjusted gauge generating part for the transformation \eqref{psiboost} and an additional term that generates a fitting transformation for the asymptotic boost transformation of $\bar{\psi}$ (\ref{cureboost}):
\begin{align}
	\mathscr{P}_0&=\frac{\vpi^i\cdot\vpi_i}{2\sqrt{g}}+\frac{\sqrt{g}}{4}\vec{F}^{ij}\cdot\vec{F}_{ij}+\frac{\vec{\Pi}\cdot\vec{\Pi}}{2\sqrt{g}}+\frac{1}{2}D^i\vp\cdot D_i\vp+V(\vp)-\vec{\psi}\cdot\vec{\mathscr{G}}-\va_i\cdot  \bar{D}^i\vec{\pi}_{\psi}.
\end{align}
The generators of the spatial Poincar\'{e} transformations are:
\begin{align}
	\mathscr{P}_i=&\vpi^{j\perp}\cdot\partial_i\va_j^{\perp}-\partial_j(\vpi^{j\perp}\cdot\va_i^{\perp})+\vec{\Pi}^{\perp}\cdot\partial_i\vp^{\perp}+\vec{\pi}_{\psi}^{\perp}\cdot\partial_i\vec{\psi}^{\perp}\nonumber\\
	&+a^2(\pi^j\partial_i A_j-\partial_j(\pi^jA_i)+\Pi\partial_i\phi+\pi_{\psi}\partial_i\psi),\label{spatialgenerator}
\end{align}
where we made the decompositions
\begin{align}
	\va_i&=A_i\vp_{\infty}-(ea^2)^{-1}\vp_{\infty}\times\partial_i\vp_{\infty}+\va_i^{\perp},\;\;\vp=\phi\vp_{\infty}+\vp^{\perp}\\
	\vpi^i&=\pi^i\vp_{\infty}+\vpi^{i\perp},\;\;\vec{\Pi}=\Pi\vp_{\infty}+\vec{\Pi}^{\perp}.
\end{align}
The form of $\mathscr{P}_i$ differs from the usual form (to generate the transformations $\delta_{(0,\vect{\xi})}$, defined in \ref{Poinis}), because we had to adjust the action $\delta_{(0,\vect{\xi})}$ by a proper gauge transformation to have it act tangentially to $\bar{\mathbb{P}}$ (see prop. \ref{adjustpoin}).

To make $P_{(\xi^{\perp},\vect{\xi})}$ functionally differentiable with respect to $\bar{\Omega}_{\vp_{\infty}}$ we have to add the same boundary terms as in electromagnetism \cite[Section 4.5]{Henneaux-ED}, i.e.\
\begin{align}
	\mathscr{B}&=b(\bar{\psi}\bar{\pi}^r+\sqrt{\bar{\gamma}}\bar{A}_{\bar{a}}\bar{\nabla}^{\bar{a}}\bar{A}_r)+Y^{\bar{a}}(\bar{\pi}^r\bar{A}_{\bar{a}}+\sqrt{\bar{\gamma}}\bar{\psi}\partial_{\bar{a}}\bar{A}_r).
\end{align}
\begin{theorem}The asymptotic symmetry algebra is build by the generators of the Poincar\'{e} group $P_{(\xi^{\perp},\vect{\xi})}$ and the generators of improper gauge transformations $G_{(\epsilon,\mu)}$, fulfilling the Poisson relations:
	\begin{align}
		\{P_{(\xi_1^{\perp},\vect{\xi}_1)},P_{(\xi_2^{\perp},\vect{\xi}_2)}\}= P_{(\hat{\xi}^{\perp},\hat{\vect{\xi}})}\; ,\; \{\bar{G}_{(\vec{\epsilon},\vec{\mu})},P_{(\xi^{\perp},\vect{\xi})}\}= \bar{G}_{(\hat{\vec{\epsilon}},\hat{\vec{\mu}})}\; ,\; \{\bar{G}_{(\vec{\epsilon}_1,\vec{\mu}_1)},\bar{G}_{(\vec{\epsilon}_2,\vec{\mu}_2)}\}\approx 0,\label{algebra}
	\end{align}
	where	
	\begin{align}\label{liepoin}
		\hat{\xi}^{\perp}&=\xi^i_1\partial_i\xi^{\perp}_2-\xi^i_2\partial_i\xi_1^{\perp}\; ,\; \hat{\xi}^i=g^{ij}(\xi^{\perp}_1\partial_j\xi_2^{\perp}-\xi_2^{\perp}\partial_j\xi_1^{\perp})+\xi_1^j\partial_j\xi_2^i-\xi_2^j\partial_j\xi_1^i\nonumber\\
		\hat{\vec{\mu}}&=\nabla^i(\xi^{\perp}D_i\vec{\epsilon})-(\xi^i\partial_i\mu)\vp_{\infty}-(\xi^j\partial_j\vec{\mu}^{\perp})^{\perp}\; ,\;\hat{\vec{\epsilon}}=\xi^{\perp}\vec{\mu}-(\xi^j\partial_j\epsilon)\vp_{\infty}-(\xi^j\partial_j\vec{\epsilon}^{\perp})^{\perp}.
	\end{align}
	Note: $(\xi^j\partial_j\epsilon)\vp_{\infty}+(\xi^j\partial_j\vec{\epsilon}^{\perp})^{\perp}=\delta_{(0,\vect{\xi})}\vec{\epsilon}$ and $(\xi^i\partial_i\mu)\vp_{\infty}+(\xi^j\partial_j\vec{\mu}^{\perp})^{\perp}=\delta_{(0,\vect{\xi})}\vec{\mu}$.
\end{theorem}
\begin{proof}
	As $P_{(0,\vect{\xi})}$ has a form that requires a splitting of the variables in a perpendicular to $\vp_{\infty}$- and parallel to $\vp_{\infty}$- part, to directly compute the Poisson brackets would be really tedious. Instead we use a trick to reduce the calculation to the calculation one has to do to get the result in electromagnetism \cite[equations (4.36)- (4.39)]{Henneaux-ED}. Let's start with $\{P_{(\xi_1^{\perp},\vect{\xi}_1)},P_{(\xi_2^{\perp},\vect{\xi}_2)}\}$. One can write $P_{(\xi^{\perp},\vect{\xi})}=a^2P^{\text{em}}_{(\xi^{\perp},\vect{\xi})}+P^{\perp}_{(\xi^{\perp},\vect{\xi})}$, where $P^{\text{em}}_{(\xi^{\perp},\vect{\xi})}$ is the term that depends exclusively on the variables $A_i$ and $\pi^i$. It is exactly the same expression as the Poincar\'{e} generator in electromagnetism.
	
	As the action of $\mathfrak{poin}$ on $\bar{\mathbb{P}}$ is canonical and therefore by \cite{Woodhouse:GQ} Hamiltonian, there is a Lie anti-homomorphism:
	\begin{align}
		\rho:(\mathfrak{poin},[\,,\,])\longrightarrow (C^{\infty}(\bar{\mathbb{P}}),\{\,,\,\}),
	\end{align}
such that with the faithful representation of $\mathfrak{poin}$ as $(\xi^{\perp},\vect{\xi})$:
\begin{align}
\{P_{(\xi_1^{\perp},\vect{\xi}_1)},P_{(\xi_2^{\perp},\vect{\xi}_2)}\}=-P_{(\tilde{\xi}^{\perp},\tilde{\vect{\xi})}}.	
\end{align}
for some $(\tilde{\xi}^{\perp},\tilde{\vect{\xi}})$. By using the equation (4.36) in \cite{Henneaux-ED}, $\{P_{(\xi_1^{\perp},\vect{\xi}_1)},P_{(\xi_2^{\perp},\vect{\xi}_2)}\}=a^2P^{\text{em}}_{(\hat{\xi}^{\perp},\hat{\vect{\xi}})}+(...)$. We also decompose $-P_{(\tilde{\xi}^{\perp},\tilde{\vect{\xi}})}=-a^2P_{(\tilde{\xi}^{\perp},\tilde{\vect{\xi}})}^{\text{em}}-P_{(\tilde{\xi}^{\perp},\tilde{\vect{\xi}})}^{\perp}$. Also for the phase space of electromagnetism there is a Lie anti-homomorphism $\rho^{\text{em}}$. It is easy to see that $\rho$ and $\rho^{\text{em}}$ are actually faithful. Using that, $P_{(\tilde{\xi}^{\perp},\tilde{\vect{\xi}})}^{\text{em}}=P_{(\hat{\xi}^{\perp},\hat{\vect{\xi}})}^{\text{em}}$ and therefore $(...)=P_{(\tilde{\xi}^{\perp},\tilde{\vect{\xi}})}^{\perp}$.

It is a straightforward calculation to show
\begin{align}
	\{\bar{G}_{(\vec{\epsilon},\vec{\mu})},P_{(\xi^{\perp},0)}\}= \bar{G}_{(\hat{\vec{\epsilon}},\hat{\vec{\mu}})(\vect{\xi}=0)}\;\text{and}\;\{\bar{G}_{(\vec{\epsilon}_1,\vec{\mu}_1)},\bar{G}_{(\vec{\epsilon}_2,\vec{\mu}_2)}\}\approx 0.
\end{align}
To calculate $\{\bar{G}_{(\vec{\epsilon},\vec{\mu})},P_{(0,\vect{\xi})}\}$ we use again a trick. Let's extend $\bar{\mathbb{P}}$ again with $\vp_{\infty}$ being a degree of freedom. On this extended space $\bar{\mathbb{P}}^{\uparrow}$ we use the symplectic form $\bar{\Omega}^{\uparrow}=\Omega+\Omega^{\psi}+\omega$. With that $X^{(0,\vect{\xi})}_{\text{former}}$ (vector field corresponding to spatial Poincar\'{e} transformations before adjusting to being tangential to $\bar{\mathbb{P}}_{\vp_{\infty}}$) is canonical. $X^{(\vec{\epsilon},\vec{\mu})}$ is canonical iff $(\vec{\epsilon},\vec{\mu})$ is field-independent. Let's write $(P_{(0,\vect{\xi})})^{\uparrow}_{\vp_{\infty}}=(P_{(0,\vect{\xi})})^{\text{former}}_{\vp_{\infty}}+(\bar{G}_{(\vec{\epsilon}^{\perp},0)})_{\vp_{\infty}}$, where $\vec{\epsilon}^{\perp}=(ea^2)^{-1}\vp_{\infty}\times Y^{\bar{a}}\partial_{\bar{a}}\vp_{\infty}$ (from the proof of proposition \ref{adjustpoin}). The arrow on the LHS means the lift of $P_{(0,\vect{\xi})}$ to the embedding of $\bar{\mathbb{P}}_{\vp_{\infty}}$ into $\bar{\mathbb{P}}^{\uparrow}$. Obviously $\vec{\epsilon}^{\perp}$ depends on the variable $\vp_{\infty}$, but along the embedding $\vp_{\infty}$ is fixed and we only consider the infinitesimal gauge transformation along this embedding, i.e. we can choose it field independet via still choosing $\vec{\epsilon}^{\perp}(\vp_{\infty})$ also at a different $\vp_{\infty}'$. With that the generator $\bar{G}_{(\vec{\epsilon}^{\perp},0)}$ actually exists. Now,
\begin{align}
\{\bar{G}_{(\vec{\epsilon},0)},(P_{(0,\vect{\xi})})^{\uparrow}_{\vp_{\infty}}\}=\underbrace{\{\bar{G}_{(\vec{\epsilon},0)},(P_{(0,\vect{\xi})})^{\text{former}}_{\vp_{\infty}}\}}_{=(1)}+\underbrace{\{\bar{G}_{(\vec{\epsilon},\vec{\mu})},\bar{G}_{(\vec{\epsilon}^{\perp},0)}\}}_{=(2)}.
\end{align}
Straightforward calculations shows $(1)=-\bar{G}_{(\xi^j\partial_j\vec{\epsilon},0)}$ and $(2)=-\bar{G}_{(\vec{\epsilon}^{\perp}\times\vec{\epsilon},0)}$, such that $(1)+(2)=-\bar{G}_{(\delta_{(0,\vect{\xi})}\vec{\epsilon},0)}$ using the definition of $\delta_{(0,\vect{\xi})}\vec{\epsilon}$ stated in the theorem.

Another straightforward calculations shows $\{\bar{G}_{(0,\vec{\mu})},(P_{(0,\vect{\xi})})^{\uparrow}_{\vp_{\infty}}\}=-\bar{G}_{(0,\delta_{(0,\vect{\xi})}\vec{\mu})}$.
\end{proof}
We denote the semi-direct sum of two Lie algebras $L$ and $L'$ by $L\oplus_{\text{semi}}L'$.
\begin{corollary}
	The Poisson algebra \eqref{algebra} in the previous theorem has the structure of a semi-direct sum of two Lie algebras.
\end{corollary}
\begin{proof}
	Take as $L'$ the Poisson algebra of the Poincar\'{e} generators. Take as $L$ the Poisson algebra of improper gauge transformations. Looking at the relations \eqref{algebra} one sees immediately the structure of the semi-direct sum.
\end{proof}
As $\rho^{-1}_{\mathfrak{poin}}:L'\longrightarrow\mathfrak{poin}$ and $\rho^{-1}_{\text{imp}}:L\longrightarrow C^{\infty}(S^2,\mathfrak{u}(1))$ are faithful Lie algebra homomorphisms, the \textit{asymptotic symmetry algebra} of $\SU(2)$ Yang-Mills-Higgs theory has the structure of a semi-direct sum of an abelian algebra with the Poincar\'{e} algebra:
\begin{align}
	\mathfrak{asym}=C^{\infty}(S^2,\mathfrak{u}(1))\oplus_{\text{semi}}\mathfrak{poin}.
\end{align}

\section{Conclusions}
The study of  $\SU(2)$ Yang-Mills-Higgs theory in the Hamiltonian framework, starting from first principles, led us to clear insights into symmetry breaking and the Higgs mechanism at a classical level. Symmetry breaking happens as a consequence of a finite Hamiltonian. The Higgs mechanism is a consequence of the Poincar\'{e} invariance of this theory. It forces the non-abelian degrees of freedom of the theory to become massive (fast fall-off) while asymptotically electromagnetic degrees of freedom as well as a topological charge -- the magnetic charge -- remain. The result is similar to the abelian Higgs mechanism \cite{Tanzi:2021}.

Because of the presence of the magnetic charge it was not completely straightforward to 
apply the results of Henneaux and Troessaert \cite{Henneaux-ED} to the asymptotic degrees of freedom. Instead we had to carefully extract proper gauge degrees of freedom from the asymptotic structure and had to apply an asymptotic partial gauge in every magnetic sector as a prerequisite to the definition of the phase space in order for the boosts acting canonically.

By the same method as in \cite{Henneaux-ED} 
we found the asymptotic symmetry algebra of electromagnetism in every magnetic sector, including global angle-dependent $\mathfrak{u}(1)$ transformations.

We highlight that the extraction of the improper gauge transformations depends in a subtle way on the underlying functional analytic structure of the phase space. In a future work we will focus on developing a rigorous mathematical background for all these Hamiltonian theories.

In a recent work \cite{Emnote} by Fuentealba, Henneaux and Troessaert the fall-off conditions for the gauge parameter is weakened. It is found that by the extention to include $\sim\text{log}(r)$ and $\sim r$ growing terms in the angle-dependent $\mathfrak{u}(1)$ transformations together with conjugacy relations between these new terms with the terms of $\mathcal{O}_1$ and $\mathcal{O}_{r^{-1}}$, $C^{\infty}(S^2,\mathfrak{u}(1))$ will be decoupled from $\mathfrak{poin}$ in $\mathfrak{asym}$. This yields in particular that one can give a definition of the angular momentum that is free from ambiguities coming from $C^{\infty}(S^2,\mathfrak{u}(1))$.

As the asymptotic structure in the $\SU(2)$ Yang-Mills-Higgs case is that of electromagnetism of \cite{Henneaux-ED}, with the only difference being the presence of topological charges, we are optimistic that this generalization works out also for $\SU(2)$ Yang-Mills-Higgs theory.

Another interesting generalization would be to consider other gauge groups together with  symmetry-breaking fields. Does the method presented here generalize to these other cases? A case of particular interest is 
obviously the electroweak gauge group $\SU(2)_L\times\text{U}(1)_Y$.

Last but not least we mention the possibility 
of a different spacetime dimension. 
It is known that the behaviour found in \cite{Tanzi:2020} critically depends on 
that dimension being four. In the case considered here the choice of dimension certainly influences the topological 
structure of the Higgs vacuum, the impact of which remains to be analysed.

\appendix
\section{Fall-off from compactification}\label{AppendixFalloff}
The general idea is to compactify the Minkowski space and the fields should be smooth at the boundary. The usual method to compactify Minkowski space is conformal compactification of the null directions \cite[Chapter 11.1]{WaldGR}. Compactification means to embed Minkowki space $M$ into a manifold with boundary $\bar{M}$, such that the interior of $\bar{M}$ is diffeomorphic to $M$ and $\partial \bar{M}=\mathscr{J}^+\cup\mathscr{J}^-\cup i^0\cup i^-\cup i^+$, where $\mathscr{J}^{+,-}\cong S^2\times\mathbb{R}$ are called future- and past \textit{null infnity} and $i^{0,+,-}\cong\{p\}$ are called \textit{spatial infinity} and future- and past \textit{timelike infinity}. The null infinities are reached when one follows any null geodesic, the timelike infinities are reached when one follows any timelike geodesic and spatial infinity is reached when one follows any spacelike geodesic. A natural assumption for fields is smoothness at the boundary. In \cite{WaldSati} Satishchandran and Wald study the asymptotic behaviour of fields at null infinity.

The idea here is to make a smoothness assumption for the fields at spatial infinity. In the usual conformal compactification, spatial infinity consists only of a single point $i^0=\{p\}$. This is to restrictive for our purposes, which can be seen by an easy argument by Ashtekar and Hansen in their seminal paper \cite{Ashtekar}. The argument is the following: Take the Maxwell field of a free falling electric charge and a Cauchy hypersurface $\Sigma$ in which the charge is at rest. Look at the conformally rescaled field on the compact manifold $\Sigma\cup\, i^0\cong S^3$ without boundary. As the total electric charge in a closed manifold has to be zero, there has to be a mirror electric charge at $i^0$, which means that the rescaled Maxwell field has to diverge at $i^0$ and is in particular not smooth at $i^0$. Electrically charged solutions should be allowed by the boundary conditions for the fields (see Dyon solution \ref{Dyon}). A solution for this problem is to consider an extended manifold as spatial infinity.
In the paper \cite{Ashtekar} two different blow ups of $i^0$ are discussed. One of them contains the other. We consider the smaller one, because it is fitting for fields on Minkowki space where the metric itself is not dynamical.

The construction is the following: Take the usual compactification of Minkowski space with $i^0=\{p\}$ and consider the \textit{unit hyperboloid} $\mathcal{H}\subset T_p\bar{M}$ that consists of all spatial directions in $T_p\bar{M}$. Take $\mathcal{H}$ as the new spatial infinity. Let $\Sigma$ be a Cauchy hypersurface in $M$. With $\mathcal{H}$ as spatial infinity, $\Sigma$ will be embedded into the compact manifold with boundary $\bar{\Sigma}$, such that $\partial\bar{\Sigma}=S\cap\mathcal{H}\cong S^2$ where $S\subset T_p\bar{M}$ is a spatial 3-dimensional subspace.

Consider a conformal transformation $\Omega$ of the Riemannian metric $g$ on $\Sigma$, where $\tilde{g}:=\Omega^2 g$. We would like to choose the conformal factor $\Omega$ such that we can extend $\tilde{g}$ onto $\bar{\Sigma}$. Let $\bar{\gamma}$ be the round metric on $S^2$, then, after choosing spherical coordinates on $\Sigma\cong\mathbb{R}^3$, $g=\text{d}r\otimes\text{d}r+r^2\bar{\gamma}$. If $\Omega=\frac{1}{r}$ at large distances, $\lim_{r\rightarrow\infty}\tilde{g}=\bar{\gamma}$. We need to find an appropriate conformal transformation behaviour of the Yang-Mills field $\va_i$. As the conformally transformed theory should describe the same physics the Lagrangian function should have the same value, whether one works with the fields or the conformally transformed fields, i.e. for the Yang-Mills field:
\begin{align}\label{confcond}
\int_{\Sigma}\text{d}^3x\,\vec{F}^{\mu\nu}\cdot\vec{F}_{\mu\nu}=\int_{\Sigma}\text{d}^3x\,\tilde{g}^{\mu\rho}\tilde{g}^{v\sigma}\vec{F}'_{\rho\sigma}\cdot\vec{F}'_{\mu\nu},
\end{align}
where $\vec{F}'$ is the (3D)-conformally transformed Yang-Mills curvature. This relation is true for any configuration of the Yang-Mills field, hence this relation is fulfilled locally, i.e. $\vec{F}^{\mu\nu}\cdot\vec{F}_{\mu\nu}=\tilde{g}^{\mu\rho}\tilde{g}^{v\sigma}\vec{F}'_{\rho\sigma}\cdot\vec{F}'_{\mu\nu}$. As $\tilde{g}^{ij}=\Omega^{-2}g^{ij}$ and $\tilde{g}^{00}=g^{00}$, it has to be $\vec{F}'_{ij}=\Omega^2\vec{F}_{ij}$. As $\vec{F}'_{ij}=2^{-1}(\partial_{[i}\va'_{j]}+e\va'_{[i}\times\va'_{j]})$ and $\Omega^2=r^{-2}$, it has to be $\va'_i=\Omega\va_i$. The conformal factor depends on the index position, i.e. $\va'^{\,i}=\Omega^{-1}\va^i$.

The fundamental assumption is smoothness of $\va'$ at the spatial boundary $S\cap\mathcal{H}\cong S^2$. Working in the compact space $\bar{\Sigma}\cong\bar{\mathbb{B}}_{1}(0)\subset\mathbb{R}^3$ we can use Taylors theorem for expanding $\va'$ at the boundary $S^2$. Take the functions $\va'^{\,i}$ at first. In the neighbourhood of $S^2$ in $\bar{\mathbb{B}}_1(0)$ we choose the coordinates $(\Omega,x^{\bar{a}})$, where $\Omega=0$ at the boundary. If $\va'^{\,i}$ is smooth at the boundary, by Taylors theorem, it has in a neighbourhood of $S^2$ the following form:
\begin{align}
\va'^{\,i}(\Omega,x^{\bar{a}})=\sum_{n=0}^{k}\frac{\partial^n_{\Omega}\vec{A}'^{i}(\Omega,x^{\bar{a}})|_{\Omega=0}}{n!}\,\Omega^n+o(\Omega^k)\quad\forall k\in\mathbb{N}\backslash\{0\}.
\end{align}
Call $\frac{\partial_{\Omega}A'^{i}(\Omega,x^{\bar{a}})|_{\Omega=0}}{n!}=:\va'^{\,i}_{(n)}(x^{\bar{a}})$. Doing the same expansion for $\va'_i=\Omega^{2}\va'^{\,i}$, it has to be $\va'_i=\sum_{n=2}^k\va'^{\,(n)}_i\Omega^{n}+o(\Omega^k)$ $\forall k\geq 2$. Let's transform these expansions to the initial fields, i.e. with $\va'^{\,i}=\Omega^{-1}\va^i$ it follows:
\begin{align}
\label{Afalloff}
\va^i(r,x^{\bar{a}})=\sum_{n=0}^k\va'^{\,i}_{(n)}(x^{\bar{a}})\Omega^{n+1}+o(\Omega^{k+1})=\sum_{n=1}^{k}\frac{1}{r^{n}}\va^i_{(n-1)}(x^{\bar{a}})+o(r^{-k})\quad\forall k\in\mathbb{N}\backslash\{0\},
\end{align}
where $\va^i_{(n-1)}(x^{\bar{a}})=\va'^{\,i}_{(n)}(x^{\bar{a}})$. This matches also with the expansion of $\Omega^{-1}\va'_i=\va_i=\va^i$. Now by taking as the fundamental assumption for the asymptotic behaviour of $\va_i$ that $\va_i$ can be smoothly extended to the compactification of $\Sigma$ in the form of a closed ball $\bar{\mathbb{B}}_1(0)\subset\mathbb{R}^3$, $\va_i$ necessary has to have the fall-off behaviour \eqref{Afalloff}.\\

To make the fall-off conditions Poincar\'{e} invariant (see \ref{aspoin}) we have to choose the fall-off for the canonical momentum of the Yang-Mills field to be an asymptotic expansion starting at one order higher in $r^{-1}$. The resulting fall-off conditions we get by this second approach are:
\begin{align}
\va_i(r,x^{\bar{a}})&=\sum_{n=1}^{\infty}\frac{1}{r^n}\va_i^{(n-1)}(x^{\bar{a}})+o(r^{-\mathbb{N}}),\\
\vpi^i(r,x^{\bar{a}})&=\sum_{n=1}^{\infty}\frac{1}{r^{n+1}}\vpi^i_{(n-1)}(x^{\bar{a}})+o(r^{-\mathbb{N}}),
\end{align}
where $o(r^{-\mathbb{N}})$ is a different notation for $o(r^{-n})$ $\forall n\in\mathbb{N}$.\\

As for the Yang-Mills field we choose the extension of the Higgs field onto $\bar{\Sigma}\cong\bar{\mathbb{B}}_1(0)$ to be smooth at the boundary. Let's determine the spatial conformal transformation of $\vp$. In order for the Lagrangian to stay the same we certainly have $D'^{\,i}\vp'\cdot D'_i\vp'=D^i\vp\cdot D_i\vp$, where ' marks the conformal transformation. We have seen before, that $\va'^{\, i}=\Omega^{-1}\va^i$ and $\va'_i=\Omega\va_i$, so especially $(\va'^{\,i}\times\vp')\cdot(\va'_i\times\vp')=(\va^{\,i}\times\vp)\cdot(\va_i\times\vp)$ $\Leftrightarrow$ $\vp(x)=\vp'(x)$ $\forall x\in\Sigma$. If $\vp$ is uniformly smooth as $r\rightarrow\infty$, which is certainly the case when $\vp'$ is smooth at the boundary, $\partial_i\vp'=\Omega\,\partial_i\vp$. $\partial^i\vp'=\Omega^{-1}\partial^i\vp$ by $\tilde{g}=\Omega^2 g$. It is also clearly $V(\vp')=V(\vp)$. By the assumption that $\vp'$ is smooth at the boundary, there exists a Taylor expansion in $\Omega$, i.e.
\begin{align}
\vp'(\Omega,x^{a})=\vp_{\infty}(x^{\bar{a}})+\sum_{n=1}^k\Omega^n\vp^{(n)}(x^{\bar{a}})+o(\Omega^k)\;\forall k\in\mathbb{N}.
\end{align}
This is equivalent to the fall-off behaviour
\begin{align}
\vp(r,x^{\bar{a}})=\vp_{\infty}(x^{\bar{a}})+\sum_{n=1}^{\infty}\frac{1}{r^n}\vp^{(n)}(x^{\bar{a}})+o(r^{-\mathbb{N}}).
\end{align}

\section{Winding number}\label{appendix1}
Let $f:M\rightarrow N$ be a smooth map between two manifolds $M,N$ of the same dimension $n$, which are also connected, closed and oriented.

In this chapter we define two different notions of a \textit{winding number} of $f$ and point out that these notions are equivalent.\\

At first, let's define the \textit{Brouwer degree}.

The following definitions and proposition A.1 until the main theorem A.2 give a quick introduction into the notion of the Brouwer degree, following the introduction Pierre de Harpe has given in his paper \cite{laHarpe}.

Let $f:M\rightarrow N$ be $C^{\infty}$ and let $y\in N$ be a regular value of $f$, i.e. $Df_x$ is invertible $\forall x\in f^{-1}(\{y\})$.The result of Sards theorem for $C^{\infty}$ functions between manifolds of the same dimension is that $N\backslash\{\text{regular values}\}$ is a Lebesque measure zero set\footnote{Null sets can obviously be defined independent of the choice of a Lebesque measure coming from a particular atlas.}. Therefore $\{\text{regular value}\}\subset N$ is dense.

Using the inverse function theorem, $f$ is a diffeomorphism around every $x\in f^{-1}(\{y\})$. Therefore $f^{-1}(\{y\})$ has to be a discrete set. $M$ is compact, which forces $f^{-1}(\{y\})$ to be finite. At every $x\in f^{-1}(\{y\})$ the diffeomorphism $f|_{U_x}$ can either be orientation preserving or orientation reversing. Define $\varepsilon_x(f)$ to be $1$ if $f|_{U_x}$ is orientation preserving and $-1$ if it is reversing.
\begin{defi}\label{locbdeg}
	The \textit{local brouwer degree} of a $C^{\infty}$-map $f:M\rightarrow N$ at a regular value $y\in N$ is:
	\begin{align}
	\text{deg}_y(f)=\sum_{x\in f^{-1}(\{y\})}\varepsilon_x(f)\in\mathbb{Z}.
	\end{align}
\end{defi}
The sum is well defined, because $f^{-1}(\{y\})$ is finite. The Brouwer degree is actually a well defined topological property of the map, i.e.
\begin{prop}\label{locbdegp}
	(1) $\text{deg}_y(f)$ is independent of the choice of regular value $y\in N$.
	
	(2) If $f':M\rightarrow N$ is $C^1$ and homotopic to $f$, then $\text{deg}(f')=\text{deg}(f)$.\label{degreeprop}
\end{prop}
\begin{proof} See \cite[Proposition 1.2.]{laHarpe}

\end{proof}

Now we make another definition of a \textit{degree}. The following definition is taken from the book \cite{Manton}, section 3.2.
\begin{defi}\label{degree2}
	Let $M$, $N$ be closed, oriented and connected manifolds of the same dimension. Let $f:M\longrightarrow N$ be a smooth map. Let $\Omega$ be a normalized volume form on $N$, i.e.\ $\int_{N}\Omega=1$. This volume form might come from a Riemannian metric, which always exists on smooth finite dimensional manifold. Consider the pull back $f^*\Omega$. The \textit{degree} of $f$ is defined as
	\begin{align}
	\text{deg}(f):=\int_{M}f^*\Omega.
	\end{align}
\end{defi}
\begin{beispiel}
	Take the Higgs vacuum $\vp_{\infty}:S^2\longrightarrow S^2$ and the standart volume form $\Omega$ on the target $S^2$. It is a basic calculation to see that
	\begin{align}
	\text{deg}(\vp_{\infty})=\frac{1}{4\pi}\int_{S^2}\text{d}^2x\,\frac{1}{2}\varepsilon^{\bar{a}\bar{b}}\vp_{\infty}\cdot(\partial_{\bar{a}}\vp_{\infty}\times\partial_{\bar{b}}\vp_{\infty}),
	\end{align}
	which exactly the magnetic charge from the section \ref{gauge}.
\end{beispiel}
The following theorem is taken from \cite{Manton}, section 3.2:
\begin{theorem}
	For smooth maps $f:M\longrightarrow N$, the degree defined in \ref{degree2} is the same as the Brouwer degree (definition \ref{locbdeg}).
\end{theorem}
\begin{proof}
	See \cite[Section 3.3]{Manton}.
\end{proof}

\section{Topological structure of the phase space}

On the path to construct a phase space we have yet to make sure that the symplectic form is finite. For that it helps to understand the global structure of the \textit{proto} phase space $\mathbb{P}^{\text{proto}}$ (defined by the conditions \eqref{fatboundary}).
\begin{prop}\label{vectorbundle} Assuming that $\mathbb{P}^{\text{proto}}$ already has the structure of a smooth manifold, $\mathbb{P}^{\text{proto}}$ fibres over the Fr\'{e}chet manifold $\mathcal{F}_{\{\vp_{\infty}\}}\cong C^{\infty}(S^2,S^2)$, which is the set of Higgs vacua $\vp_{\infty}$. In the $C^{\infty}$ compact-open topology on $C^{\infty}(S^2,S^2)$ (which underlies the manifold structure),
\begin{align}
	\mathcal{F}_{\{\vp_{\infty}\}}=\bigsqcup_{m\in\mathbb{Z}}\mathcal{F}_{\{\vp_{\infty}\}}^{(m)},
\end{align}
where $\mathcal{F}_{\{\vp_{\infty}\}}^{(m)}$ are path-connected components.
\end{prop}
\begin{proof}
	We stay short on the mathematical details here, but all the information is available in the references.
	
	By \cite[Theorem 7.6 (b)]{wockel}, $C^{\infty}(S^2,S^2)$ is a Fr\'{e}chet manifold in the $C^{\infty}$ compact-open topology.
	
	By the theorem of Hopf \cite[p.51]{milnor}, $f,g\in C^{\infty}(S^2,S^2)$ have the same degree $m\in\mathbb{Z}$ (see \ref{appendix1}) if and only if $f$ and $g$ are smoothly homotopic. By the exponential law for function spaces \cite[Theorem 7.6 (e)]{wockel}, $f$ and $g$ are connected by a smooth homotopy if and only if $f$ and $g$ are connected by a smooth path in $C^{\infty}(S^2,S^2)$. $\Rightarrow$ $f$ and $g$ lie in the same path-connected component.
	
	Let $f,g\in C^{\infty}(S^2,S^2)$ be connected by a continuous path $\gamma$ in the $C^{\infty}$ c.-o. topology. As the $C$ c.-o. topology is coarser as the $C^{\infty}$ c.-o. topology, $\gamma$ is also a continuous path in the $C$ c.-o. topology. By the exponential law for the $C$ c.-o. topology \cite[Theorem 1.6]{brown}, $f$ and $g$ are connected by a continuous curve if and only if $f$ and $g$ are connected by a continuous homotopy. $\Rightarrow$ $\text{deg}(f)=\text{deg}(g)$ (see \ref{appendix1}).\\
	
	To see that $\mathbb{P}^{\text{proto}}$ fibres over $\mathcal{F}_{\{\vp_{\infty}\}}$ we have to construct a surjective submersion $P:\mathbb{P}^{\text{proto}}\longrightarrow\mathcal{F}_{\{\vp_{\infty}\}}$. Take the projection map $P:(\va,\vpi,\vp,\vec{\Pi})\mapsto\vp_{\infty}$ which can be written as $P=(0,0,P_{\vp_{\infty}},0)$, where
	\begin{align}
	    P_{\vp_{\infty}}: \vp=\text{o}(r^{-\mathbb{N}})+\vp_{\infty}\mapsto\vp_{\infty}.
	\end{align}
	In local coordinates on the manifold of the Higgs field $\{\vp\}$ modelled on the locally convex space $E$ (see \cite{Neeb}) it is easy to compute the derivative in the direction $\vec{h}\in E$, i.e.\ $DP_{\vp_{\infty}}(\vp,\vec{h})=P_{\vp_{\infty}}(\vec{h})$. Using \cite[Lemma II.4]{LocFunc} it is easy to see that $DP_{\vp_{\infty}}$ is continuous in $(\vp,\vec{h})\in E\times E$\footnote{In more detail, use that $\sup_{x\in S^2}\|P_{\vp_{\infty}}(\vec{h})(x)\|\leq\sup_{x\in\mathbb{R}^3}\|h(x)\|$ is constant in $\vp$ and the relation $C^{\infty}(\mathbb{R}^3,\mathbb{R}^3)\times C^{\infty}(\mathbb{R}^3,\mathbb{R}^3)\cong C^{\infty}(\mathbb{R}^3,\mathbb{R}^3\oplus\mathbb{R}^3)$ using the standart topologies.}. Hence $P_{\vp_{\infty}}$ is Bastiani $C^1$ \cite{Neeb}. Clearly $DP_{\vp_{\infty}}$ as well as $P_{\vp_{\infty}}$ are surjective. It follows that $P$ is a surjective submersion.
	
\end{proof}

		\bibliography{biblio}{}

\providecommand{\href}[2]{#2}\begingroup\raggedright\begin{thebibliography}{10}

\bibitem{Henneaux-ED}
Marc Henneaux and Cédric Troessaert, \emph{Asymptotic symmetries of
  electromagnetism at spatial infinity},
  \href{https://doi.org/10.1007/jhep05(2018)137}{\emph{Journal of High Energy
  Physics} {\bfseries 5} (2018) 137}
  [\href{https://arxiv.org/abs/1803.10194}{{\ttfamily 1803.10194}}].

\bibitem{Tanzi:2020}
Roberto Tanzi and Domenico Giulini, \emph{Asymptotic symmetries of {Yang-Mills}
  fields in hamiltonian formulation},
  \href{https://doi.org/10.1007/JHEP10(2020)094}{\emph{Journal of High Energy
  Physics} {\bfseries 10} (2020) 94}
  [\href{https://arxiv.org/abs/2006.07268}{{\ttfamily 2006.07268}}].

\bibitem{Tanzi:2021}
Roberto Tanzi and Domenico Giulini, \emph{Asymptotic symmetries of scalar
  electrodynamics and of the abelian higgs model in hamiltonian formulation},
  \href{https://doi.org/10.1007/JHEP08(2021)117}{\emph{Journal of High Energy
  Physics} {\bfseries 8} (2021) }
  [\href{https://arxiv.org/abs/2101.07234}{{\ttfamily 2101.07234}}].

\bibitem{Dirac-book}
Paul Adrien~Maurice Dirac, \emph{Lectures on Quantum Mechanics}. Belfer
  Graduate School of Science, 1964.

\bibitem{TanziPHD}
Roberto Tanzi, \emph{Hamiltonian study of the asymptotic symmetries of gauge
  theories}, {\emph{PhD thesis} (2021) }
  [\href{https://arxiv.org/abs/2109.02350}{{\ttfamily 2109.02350}}].

\bibitem{Neeb}
Karl-Hermann Neeb, ``Infinite-dimensional lie groups.'' Monastir Summer School,
  \url{https://cel.archives-ouvertes.fr/cel-00391789/document}, 2005.

\bibitem{Mitchell2006NotesOP}
Stephen~A. Mitchell, \emph{Notes on principal bundles and classifying spaces},
  {\emph{\url{https://math.mit.edu/~mbehrens/18.906spring10/prin.pdf}} (2006)
  }.

\bibitem{Waida}
Spenta~R. Wadia, \emph{Hamiltonian formulation of non-abelian gauge theory with
  surface terms: Applications to the dyon solution},
  \href{https://doi.org/10.1103/PhysRevD.15.3615}{\emph{Phys. Rev. D}
  {\bfseries 15} (1977) 3615}.

\bibitem{HenneauxTeitelboim}
Marc Henneaux and Claudio Teitelboim, \emph{Quantization of Gauge Systems}.
  Princeton University Press, 1992.

\bibitem{Troessaert}
C\'{e}dric Troessaert, \emph{Canonical structure of field theories with
  boundaries and applications to gauge theories}, {\emph{arXiv} (2013) }
  [\href{https://arxiv.org/abs/1312.6427}{{\ttfamily 1312.6427}}].

\bibitem{Goddard:1978}
P~Goddard and D~I Olive, \emph{Magnetic monopoles in gauge field theories},
  \href{https://doi.org/10.1088/0034-4885/41/9/001}{\emph{Reports on Progress
  in Physics} {\bfseries 41} (1978) 1357}.

\bibitem{Julia.Zee:1975}
Bernard~L. Julia and Anthony Zee, \emph{Poles with both magnetic and electric
  charges in nonabelian gauge theory},
  \href{https://doi.org/10.1103/PhysRevD.11.2227}{\emph{Physical Review D}
  {\bfseries 11} (1975) 2227}.

\bibitem{Giulini:1995}
Domenico Giulini, \emph{Asymptotic symmetry groups of long-ranged gauge
  configurations},
  \href{https://doi.org/10.1142/s0217732395002210}{\emph{Modern Physics Letters
  A} {\bfseries 10} (1995) 2059}.

\bibitem{Woodhouse:GQ}
Nicholas Woodhouse, \emph{Geometric Quantization}. Clarendon Press, Oxford,
  1980.

\bibitem{Roeser}
Markus Röser, \emph{Yang-mills theory and stable bundles on kähler
  manifolds}, {\emph{Lecture notes at the Leibniz Universität Hannover} (2018)
  }.

\bibitem{WaldGR}
Robert~M Wald, \emph{{General relativity}}. Chicago Univ. Press, Chicago, IL,
  1984.

\bibitem{wockel}
Christoph Wockel, ``Infinite-dimensional and higher structures in differential
  geometry (lecture notes).''
  \url{https://www.math.uni-hamburg.de/home/wockel/teaching/data/HigherStructures2013/hs.pdf},
  2014.

\bibitem{Lee}
John~M. Lee, \emph{Introduction to Smooth Manifolds}. Springer, New York, NY,
  2012,
  \href{https://doi.org/10.1007/978-1-4419-9982-5}{10.1007/978-1-4419-9982-5}.

\bibitem{Spanier}
Edwin~H. Spanier, \emph{Algebraic Topology}. Springer, New York, NY, 1966,
  \href{https://doi.org/10.1007/978-1-4684-9322-1}{10.1007/978-1-4684-9322-1}.

\bibitem{Emnote}
Oscar Fuentealba, Marc Henneaux, and Cedric Troessaert, \emph{A note on the
  asymptotic symmetries of electromagnetism},
  \href{https://doi.org/10.1007/JHEP03(2023)073}{\emph{Journal of High Energy
  Physics} {\bfseries 73} (2023) }
  [\href{https://arxiv.org/abs/2301.05989}{{\ttfamily 2301.05989}}].

\bibitem{WaldSati}
Gautam Satishchandran and Robert~M. Wald, \emph{Asymptotic behavior of massless
  fields and the memory effect},
  \href{https://doi.org/10.1103/PhysRevD.99.084007}{\emph{Phys. Rev. D}
  {\bfseries 99} (2019) 084007}.

\bibitem{Ashtekar}
Abhay Ashtekar and R.~O. Hansen, \emph{A unified treatment of null and spatial
  infinity in general relativity. i. universal structure, asymptotic
  symmetries, and conserved quantities at spatial infinity},
  \href{https://doi.org/10.1063/1.523863}{\emph{Journal of Mathematical
  Physics} {\bfseries 19} (1978) 1542}.

\bibitem{laHarpe}
Pierre de~la Harpe, \emph{Brouwer degree, domination of manifolds, and groups
  presentable by products},  \href{https://arxiv.org/abs/1609.06637}{{\ttfamily
  1609.06637}}.

\bibitem{Manton}
Nicholas Manton and Paul Sutcliffe, \emph{Topological Solitons}, Cambridge
  Monographs on Mathematical Physics. Cambridge University Press, 2004,
  \href{https://doi.org/10.1017/CBO9780511617034}{10.1017/CBO9780511617034}.

\bibitem{milnor}
John Milnor, \emph{Topology from the differentiable viewpoint}. Princeton
  University Press (rev. 1997), 1965.

\bibitem{brown}
R.~Brown, \emph{Function spaces and product topologies},
  \href{https://doi.org/10.1093/QMATH/15.1.238}{\emph{Quarterly Journal of
  Mathematics} {\bfseries 15} (1964) 238}.

\bibitem{LocFunc}
Christian Brouder, Nguyen~Viet Dang, Camille Laurent-Gengoux, and Kasia
  Rejzner, \emph{{Properties of field functionals and characterization of local
  functionals}}, \href{https://doi.org/10.1063/1.4998323}{\emph{Journal of
  Mathematical Physics} {\bfseries 59} (2018) }
  [\href{https://arxiv.org/abs/https://pubs.aip.org/aip/jmp/article-pdf/doi/10.1063/1.4998323/16088263/023508\_1\_online.pdf}{{\ttfamily
  https://pubs.aip.org/aip/jmp/article-pdf/doi/10.1063/1.4998323/16088263/023508\_1\_online.pdf}}].

\end{thebibliography}\endgroup
		\bibliographystyle{JHEP-fullnames}   
		
	\end{singlespace}
\end{document}